\documentclass[journal]{IEEEtran}
\def\BibTeX{{\rm B\kern-.05em{\sc i\kern-.025em b}\kern-.08em
    T\kern-.1667em\lower.7ex\hbox{E}\kern-.125emX}}
\usepackage{epsfig,epstopdf,graphicx,psfrag,amsmath,cases}
\usepackage{latexsym,amssymb,amsmath,epsfig,algorithm,amsthm}
\usepackage{algorithmic}
\usepackage{color}
\usepackage{url}
\usepackage{amsthm}
\usepackage{scrtime}
\usepackage{bm}
\usepackage{cancel}
\usepackage{hyperref}
\usepackage{bbding}
\usepackage{multicol}
\usepackage{graphicx}
\usepackage{amsthm}
\usepackage{cite}
\usepackage{array}
\usepackage[utf8]{inputenc}
\usepackage[english]{babel}
\usepackage{setspace}
\usepackage{subcaption}
\addto\captionsenglish{}
\usepackage[font=small]{caption}

\begin{document}
\title{Robust and Secure Sum-Rate Maximization for Multiuser MISO Downlink Systems with Self-sustainable IRS}

\author{\IEEEauthorblockN{Shaokang Hu}}

\author{Shaokang~Hu,~\IEEEmembership{Student Member,~IEEE,}
           ~Zhiqiang~Wei,~\IEEEmembership{Member,~IEEE,}~Yuanxin~Cai,~\IEEEmembership{Student Member,~IEEE,}~Chang~Liu,~\IEEEmembership{Member,~IEEE,}~Derrick~Wing~Kwan~ Ng,~\IEEEmembership{Fellow,~IEEE}
           and~Jinhong~Yuan,~\IEEEmembership{Fellow,~IEEE}\thanks{This work has been presented in part at the IEEE Global Communications Conference, 2020\cite{hu2020sum}.} \thanks{The work of Z. Wei was supported by Alexander von Humboldt Foundation. The work of C. Liu was supported by the National Natural Science Foundation of China under Grant 61801082. D. W. K. Ng is supported by funding from the UNSW Digital Grid Futures Institute, UNSW, Sydney, under a cross-disciplinary fund scheme and by the Australian Research Council's Discovery Project (DP210102169). The work was partially supported by the Australian Research Council (ARC) Discovery Projects under Grant DP190101363 and by the ARC Linkage Project under Grant LP170101196.
           }
           \thanks{
           S. Hu, Y. Cai, C. Liu, D. W. K. Ng, and J. Yuan are with the School of Electrical Engineering and Telecommunications, University of New South Wales, Sydney, NSW 2052, Australia (e-mail: shaokang.hu@unsw.edu.au; yuanxin.cai@unsw.edu.au; chang.liu19@unsw.edu.au; w.k.ng@unsw.edu.au; j.yuan@unsw.edu.au).
           }
           \thanks{C. Liu also with the National Key Laboratory of Science and Technology on Communications, University of Electronic Science and Technology of China, Chengdu, China.
           }
           \thanks{Z. Wei is with the Institute for Digital Communications (IDC), Friedrich-Alexander University Erlangen-Nuremberg, Germany (e-mail: zhiqiang.wei@fau.de).
           }
           }

{\newtheorem{Thm}{Theorem}}
\newtheorem{theorem}{Theorem}
\newtheorem{Lem}{Lemma}
\newtheorem{Cor}{Corollory}
\newtheorem{Def}{Definition}
\newtheorem{Exam}{Example}
\newtheorem{Alg}{Algorithm}
\newtheorem{Prob}{Problem}
\newtheorem{Proof}{Proof}
\newtheorem{Remark}{Remark}
\newcommand{\abs}[1]{\lvert#1\rvert}
\newcommand{\norm}[1]{\lVert#1\rVert}

\maketitle
\begin{abstract}
This paper investigates robust and secure multiuser multiple-input single-output (MISO) downlink communications assisted by a self-sustainable intelligent reflection surface (IRS), which can simultaneously reflect and harvest energy from the received signals. We study the joint design of beamformers at an access point (AP) and the phase shifts as well as the energy harvesting schedule at the IRS for maximizing the system sum-rate. The design is formulated as a non-convex optimization problem taking into account the wireless energy harvesting capability of IRS elements, secure communications, and the robustness against the impact of channel state information (CSI) imperfection. Subsequently, we propose a computationally-efficient iterative algorithm to obtain a suboptimal solution to the design problem. In each iteration, $\mathcal{S}$-procedure and the successive convex approximation are adopted to handle the intermediate optimization problem. Our simulation results unveil that: 1) there is a non-trivial trade-off between the system sum-rate and the self-sustainability of the IRS; 2) the performance gain achieved by the proposed scheme is saturated with a large number of energy harvesting IRS elements; 3) an IRS equipped with small bit-resolution discrete phase shifters is sufficient to achieve a considerable system sum-rate of the ideal case with continuous phase shifts.
\end{abstract}
%
\section{Introduction}
\IEEEPARstart{T}{he} sixth-generation (6G) networks are expected to serve as a key enabler for realizing the future intelligent digital society in 2030, offering superior communication services compared with the current fifth-generation (5G) networks, such as ultra-high data rate, high energy efficiency, large global coverage, and highly secure communications \cite{zhang20196g}. To satisfy these requirements, various technologies have been proposed in recent years. In particular, massive multiple-input multiple-output (MIMO), millimeter wave (mmWave) communications, and ultra-dense networks (UDNs) are prominent candidates. Although these technologies are able to enhance the wireless network coverage and capacity to a certain extent, they generally incur increased hardware costs, energy consumption, and signal processing complexity, due to the aggressive deployment of large-scale antennas and active nodes in the systems \cite{wu2017overview,shafi20175g}. As such, it is still a fundamental challenge to achieve a sustainable capacity growth of future wireless networks in terms of the communication cost and energy efficiency.

To fulfill the stringent requirements of 6G networks\cite{zhang20196g}, the emerging intelligent reflecting surface (IRS)-assisted wireless communications \cite{wu2019intelligent} have received an unprecedented upsurge of attentions recently. Specifically, an IRS consists of a large number of low-cost passive reflection elements which can independently reflect the incident electromagnetic wave with controlled phase shifts. By intelligently adapting the phase shift of each IRS element to the conditions of communication channels, the reflected signals can be coherently combined at the desired receivers. As such, IRSs are able to reshape the signal propagation and to establish a favorable communication environment for achieving various purposes, such as enhancing the physical layer security, harnessing multiple access interference, and improving the quality of communication in terms of both spectral efficiency and energy efficiency\cite{wu2020intelligent}. For example, \cite{wu2019intelligent} confirmed that IRS-assisted communication systems can extend the signal coverage compared with conventional systems where the IRS is absent. Furthermore, the results in \cite{wu2019weighted} showed that the introduction of an IRS can significantly improve both the achievable system data rate and the total harvested power in simultaneous wireless information and power transfer (SWIPT) systems.

 Despite the fruitful results in the literature, e.g. \cite{liu2021deep,yu2021smart,9145224,wu2019intelligent,wu2020intelligent,abeywickrama2020intelligent,pan2020multicell,zhang2020joint}, most of the works idealistically assumed that the energy consumption of the IRS is negligible. However, typical power consumption values of each phase shifter are $1.5$, $4.5$, $6$, and $7.8$ mW for $3$-, $4$-, $5$-, and $6$-bit resolution phase shifters, respectively \cite{mendez2016hybrid,ribeiro2018energy,huang2018energy,huang2019reconfigurable}. Considering a hundred of IRS elements, the operational power of IRS with $3$-bit phase shifters becomes $0.15$ W which is comparable to the radiated communication power. In fact, IRSs are expected to be battery- or grid- powered \cite{wu2019towards}. However, powering IRSs via conventional powerline networks would not only increase the implementation cost, but would also decrease the flexibility of IRS deployment. On the other hand, battery-powered IRS are usually equipped with a limited energy storage leading to restrictive operational lifetime that creates a bottleneck in communication networks. Moreover, manually replacing batteries of IRS may be costly or even impossible due to environmental hazards. As a result, the amalgamation of energy harvesting with IRSs is a promising alternative for providing self-sustainable and uninterrupted communication services.
 In practice, the availability of major conventional energy harvesting resources, such as wind, geothermal, and solar, are usually limited by the weather and energy source locations. As a result, harvesting energy from radio frequency (RF) in wireless communication systems is more reliable and suitable than that from natural sources, due to the controllability of wireless power transfer (WPT)\cite{clerckx2018fundamentals}.
To enable uninterrupted
communication services, a self-sustainable IRS enabled by WPT was considered in \cite{lyu2020optimized,zou2020wireless} for improving the system sum-rate and capacity, respectively.
 However, \cite{lyu2020optimized} focused on the resource allocation design of a single-antenna AP and the obtained results cannot be directly extended to the case adopting a multi-antenna transmitter. On the other hand, although a multi-antenna AP system was considered in \cite{zou2020wireless}, the obtained results were developed for the case of a single receiver system, which are not applicable to systems serving with multiple receivers.
Moreover, all IRS elements in \cite{lyu2020optimized,zou2020wireless} are forced to adopt the same operating mode in every time instant, which is a strictly suboptimal setting decreasing the flexibility in resource allocation. Furthermore, both \cite{lyu2020optimized,zou2020wireless} ideally assumed a linear energy harvesting model for the conversion of the scavenged RF energy to the direct current (DC) power for simplicity. Yet, the conventional linear energy harvesting model fails to capture the non-linear features of practical energy harvesting circuits \cite{le2008efficient,guo2012improved} and resource allocation designs based on the simplified linear model can lead to severe resource allocation mismatches and performance degradation. Additionally, \cite{lyu2020optimized,zou2020wireless} assumed the availability of continuous phase shifters at IRS which can hardly implemented in particle due to high hardware implementation cost.
On the other hand, to facilitate effective energy harvesting, transmitters in SWIPT systems usually increase the power of the information-carrying signals, which may also increase the susceptibility to eavesdropping due to the higher potential of information leakage, which is not considered in \cite{lyu2020optimized,zou2020wireless}. Thus, communication security is a fundamental issue in wireless-powered IRS-assisted systems.

Conventionally, security provisioning methods rely on cryptographic techniques applied at the upper layers of wireless networks, which require high computational complexity leading to a large amount of energy consumption. Fortunately, different from cryptographic techniques, physical layer security is computationally-efficient and effective to safe guard communications by exploiting the inherent randomness of wireless channels \cite{sun2019physical}. In particular, there are several methods to enhance the physical layer security in wireless networks, e.g. cooperative jamming\cite{dong2009cooperative} and artificial-noise (AN)-aided beamforming \cite{sun2018robust}.
Indeed, exploiting AN or jamming signals in wireless networks is a double-edged sword. On the one hand,  it can facilitate secure communications by degrading the quality of eavesdroppers' channels. On the other hand, it may also cause inevitable interference to legitimate users. As a remedy, the application of IRS and AN has been investigated to guarantee secure communications, e.g.
\cite{shen2019secrecy,cui2019secure}. In practice, the degrees of freedom (DoF) offered by the IRS can provide a better interference management which alleviates the negative impacts caused by AN transmission. Indeed, to unleash the potential gain brought by the IRS, accurate channel state information (CSI) is needed for effective beamforming design. However, numerous literature, e.g. \cite{shen2019secrecy,cui2019secure}, over optimistically assumed that perfect CSI of eavesdropper channels is available at the AP which can hardly be obtained in practice. Especially, potential eavesdroppers may not interact with the AP and acquiring perfect CSI of the eavesdropping channel at the AP is challenging. In other words, adopting the results from \cite{shen2019secrecy,cui2019secure} for robust beamforming design may increase the risk of security breach even if an IRS is deployed. To this end, considering imperfect CSI of eavesdroppers, a secure IRS-assisted wireless network for maximizing the system sum-rate in the presence of multiple users and eavesdroppers was presented in \cite{yu2019robust}. Nevertheless, the design in \cite{lyu2020optimized,zou2020wireless,yu2019robust} still required the availability of perfect CSI of IRS communication channels. Unfortunately, due to the nearly passive implementation of the IRS, obtaining perfect CSI for the IRS related links is challenging, if not impossible. Also, the authors in \cite{zhou2020robust} showed that when the joint transmitter precoder and IRS phase shifts design ignores the existence of possible CSI errors, existing non-robust designs, e.g. \cite{shen2019secrecy,cui2019secure}, cannot always guarantee the required target rate which jeopardizes the system performance significantly. As a result, an efficient robust beamforming design to strike a balance between the system sum-rate and the IRS self-sustainability for secure communication is desired which has not been reported in the literature, yet.

{Motivated by the aforementioned observations, we consider a robust and secure IRS-assisted multiuser MISO downlink wireless system in the presence of multiple single-antenna potential eavesdroppers. In particular, by exploiting the large number of IRS elements, our design allows a portion of the IRS elements to reflect signals from the AP and the remaining to harvest energy for supporting the energy consumption of the IRS elements.}
To maximize the system sum-rate, we jointly optimize the beamforming for information signal and AN, the discrete phase shifts of the IRS, and the harvesting schedule at the IRS.
{The
resource allocation design is formulated as a non-convex mixed-integer optimization problem takes into account the imperfect CSI of all wireless channels, the non-linearity of energy harvesting, and the communication security quality of service (QoS) requirement}.
To address the design problem at hand, this paper proposes a computationally-efficient iterative suboptimal algorithm.
{In particular, we first propose a series of transformations to resolve the coupling between the optimization variables which paves the way for the development of a computationally efficient iterative algorithm based on the successive convex approximation (SCA) and semi-definite relaxation (SDR). Unlike the conventional alternation optimization (AO) method \cite{lyu2020optimized,zou2020wireless,abeywickrama2020intelligent,wu2019intelligent}, which can be easily trapped in inefficient stationary points over iterations in practice, our proposed algorithm is able to jointly optimize all these optimization variables in each iteration enjoying a higher system sum-rate.}
Our results not only show the non-trivial trade-off between the system sum-rate and the self-sustainability of IRS, but also unveil the impact of limited bit resolution of the IRS phase shifters on the system performance. In addition, our results reveal that deploying self-sustainable IRS can significantly enhance the physical layer security in wireless communication systems.

\emph{Notations}: The scalars, vectors, and matrices are represented by lowercase letter $x$, boldface lowercase letter $\mathbf{x}$, and boldface uppercase letter $\mathbf{X}$, respectively. $\mathbb{B}^{N \times M}$, $\mathbb{R}^{N \times M}$, and $\mathbb{C}^{N \times M}$ denote the space of $N \times M$ matrices with binary, real, and complex entries, respectively.  $\mathbb{H}^N$ denotes the set of all $N\times N$ Hermitian matrices. The modulus of a complex-valued scalar is denoted by $|\cdot|$.  The Euclidean norm, spectral norm, nuclear norm, and Frobenius norm of a matrix are denoted by $\|\cdot\|$, $\|\cdot\|_2$, $\|\cdot\|_*$, and $\|\cdot\|_{\mathrm{F}}$, respectively. The transpose, conjugate transpose, conjugate, expectation, rank, and trace of a matrix are denoted as $(\cdot)^{\mathrm{T}}$, $(\cdot)^{\mathrm{H}}$, $(\cdot)^{*}$, $\mathbb{E}\{\cdot\}$, $\mathrm{Rank}(\cdot)$, and $\mathrm{Tr(\cdot)}$, respectively. The maximum eigenvalue of a matrix is denoted by ${\lambda}_{\max}(\cdot)$ and the eigenvector associated with the maximum eigenvalue is denoted by $\bm{\lambda}_{\max}(\cdot)$. $[x]^+$ represents $\max\{0,x\}$. $\mathbf{X}\succeq\mathbf{0}$ and $\mathbf{X}\preceq\mathbf{0}$ mean that matrix $\mathbf{X}$ is positive semi-definite and negative semi-definite, respectively. $\mathrm{Re}(\cdot)$ stands for the real part of a complex number. $\mathrm{diag(\mathbf{x})}$ denotes a diagonal matrix with its diagonal elements given by vector $\mathbf{x}\in \mathbb{C}^{N\times1}$ and $\mathrm{vec}(\mathbf{X})$ results in a column vector obtained by sequentially stacking of the columns of matrix $\mathbf{X}$. Operator $[\mathbf{X}]_{nm}$ returns the element at the $n$-th row and the $m$-th column of the matrix $\mathbf{X}$. $j$ denotes the imaginary unit. For a continuous function $f(\mathbf{X})$, $\nabla_{\mathbf{X}}f(\cdot)$ represents the gradient of $f(\cdot)$ with respect to matrix $\mathbf{X}$. The distribution of a circularly symmetric complex Gaussian (CSCG) random variable with mean $\mu$ and variance $\sigma^2$ is denoted by $\mathcal{CN}(\mu, \sigma^2)$ and $\sim$ stands for ``distributed as''. $\mathbf{I}_N$ denotes an $N\times N$ identity matrix. $\otimes$ stands for the Kronecker product. 
\section{System Model}
\subsection{Signal Model}
\begin{figure}[t] 
  \centering
  \includegraphics[width=3.4in]{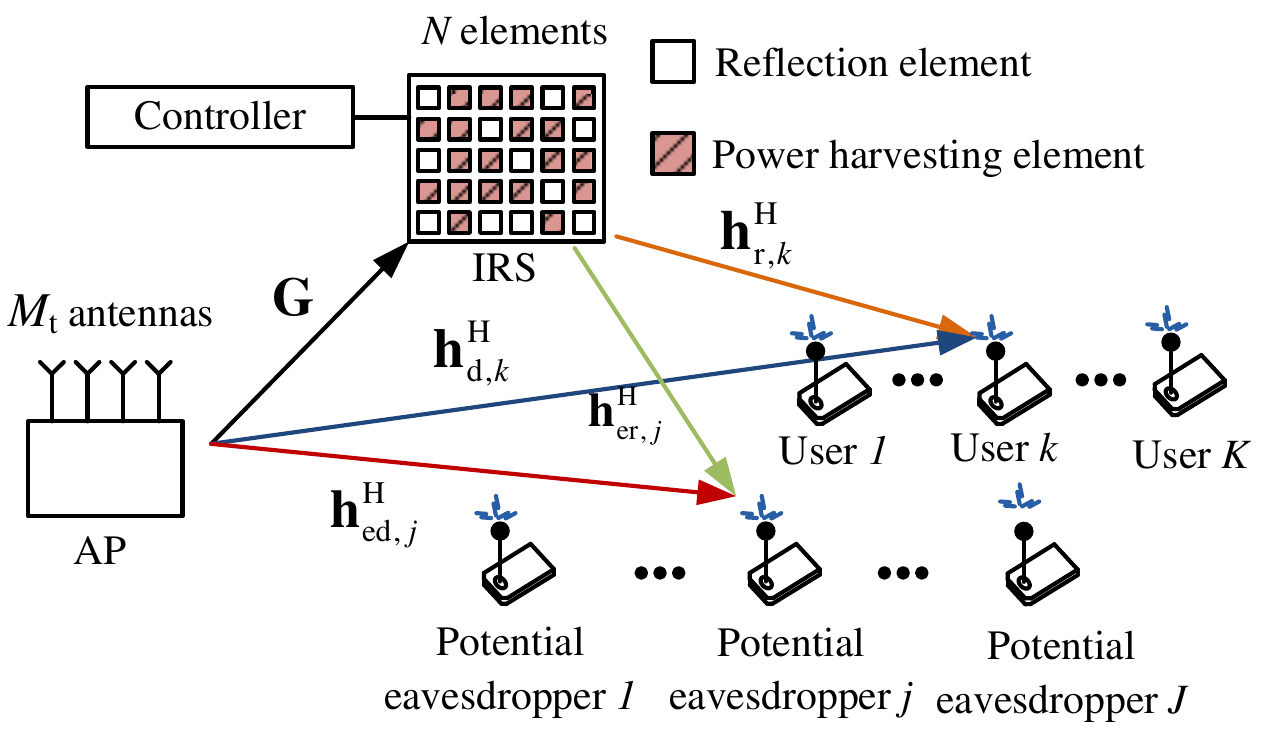}  
  \caption{\hspace{-1mm} {A downlink wireless communication system assisted by a self-sustainable IRS.}}
  \label{system_model} 
\end{figure}

As shown in Fig. \ref{system_model}, this paper considers a downlink MISO system assisted by a wireless-powered IRS. An AP equipping with $M_{\mathrm{t}}>1$ antennas transmits $K>1$ independent data streams to $K$ single-antenna users simultaneously, denoted by a set $ \mathcal{K} = \{1, \ldots ,K\} $. There are $J$ potential eavesdroppers equipped with a single-antenna each, denoted by a set $\mathcal{J} = \{1,\ldots, J\}$. We assume that all the eavesdroppers work independently and $M_{\mathrm{t}}> J$ holds to guarantee communication security, as commonly adopted in the literature\cite{sun2016multi,zhao2015joint,zhang2016secure}. Besides, the IRS panel consists of $N$ IRS elements, denoted by a set $\mathcal{N} = \{1, \ldots ,N\} $. The reflection matrix of the IRS is denoted as $\bm{\Theta}\in\mathbb{C}^{N\times N}$ and its details will be discussed in Section \ref{Section:IRS_MODEL}.

This paper assumes a quasi-static flat fading channel model. The baseband equivalent channels from the AP to the IRS (AP-IRS), from the IRS to the $k$-th user (IRS-user), from the AP to the $k$-th user (AP-user), from the IRS to the $j$-th potential eavesdropper (IRS-Eve), and from the AP to the $j$-th potential eavesdropper (AP-Eve) are denoted by $\bm{\mathbf{G}}\in\mathbb{C}^{N \times M_{\mathrm{t}}}$, $\mathbf{h}_{\mathrm{r},k}\in\mathbb{C}^{N \times 1}$, $\mathbf{h}_{\mathrm{d},k}\in\mathbb{C}^{M_{\mathrm{t}}\times 1}$, $\mathbf{h}_{\mathrm{re},j}\in\mathbb{C}^{N \times  1}$ and $\mathbf{h}_{\mathrm{ed},j}\in\mathbb{C}^{M_{\mathrm{t}}\times 1}$, respectively.
The transmitted signal from the AP is given by
\begin{align}\label{tx_signal}
    \mathbf{x} = \sum_{k\in\mathcal{K}}\mathbf{w}_k x_k+\mathbf{z},
\end{align}
where $\mathbf{w}_k\in\mathbb{C}^{M_{\mathrm{t}}\times1}$ is the precoding vector for the $k$-th user, $x_k\sim \mathcal{CN}(0,1)$, $\forall k \in \mathcal{K}$, with $\mathbb{E}\{|x_k|^2\} = 1$, is the modulated data symbol for the $k$-th user, and $\mathbf{z}\in\mathbb{C}^{M_{\mathrm{t}}\times 1}$ denotes the AN vector generated by the AP to deliberately combat the eavesdroppers. Specifically, $\mathbf{z}$ is a CSCG vector with $\mathbf{z}\sim\mathcal{CN}(\mathbf{0},\mathbf{Z})$, where $\mathbf{Z}\succeq\mathbf{0}$ is the covariance matrix of the AN vector. Note that the AN is assumed to be unknown to both the legitimate users and the potential eavesdroppers, which is to be optimized by the proposed algorithm to facilicate communication security provisioning and energy harvesting. We assume that the AP has a total transmit power $P_{\max}$, i.e., $\mathbb{E}\{\|\mathbf{x}\|^2\} = \sum_{k\in\mathcal{K}}\|\mathbf{w}_k\|^2+\mathrm{Tr}(\mathbf{Z})\le P_{\max}$.
In the system, each user receives signals via two links, i.e.,  from the indirect path via the IRS (AP-IRS-user) and the direct AP-user channel. Similarly, each potential eavesdropper also receives signals via two links, i.e., from the indirect path via the IRS (AP-IRS-Eve) and the direct AP-Eve channel. Thus, the signals received at the $k$-th user and at the $j$-th eavesdropper are given by\footnote{For a small-cell network with a cell radius of $200$ meters, the delay between the propagation path reflected by the IRS and the direct path is typically around $1$ $\mathrm{\mu s}$, which is much shorter than a symbol duration, e.g. $70$ $\mathrm{\mu s}$ in Long-Term Evolution systems\cite{arunabha2010fundamentals}. Thus, the potential inter-symbol interference caused by the two paths is not considered in \eqref{y}.}
\begin{align}
    &y_k \!\!=\!\! \Big(\!\mathbf{h}_{\mathrm{d},k}^\mathrm{H}\!+\!\mathbf{h}_{\mathrm{r},k}^\mathrm{H}\bm{\Theta}\mathbf{G}\!\Big)\!\Big(\!\sum_{i\in \mathcal{K}}\!\mathbf{w}_i x_i \!+\! \mathbf{z}\Big) +n_k, \forall k\in\mathcal{K},\text{ and}\label{y}\\
    &y_{\mathrm{eve},j}\!\! =\!\! \Big(\!\mathbf{h}_{\mathrm{ed},j}^\mathrm{H}\!+\!\mathbf{h}_{\mathrm{re},j}^\mathrm{H}\bm{\Theta}\mathbf{G}\!\Big)\!(\!\sum_{i\in \mathcal{K}}\!\mathbf{w}_i x_i \!+\! \mathbf{z})\!+\!{n}_{\mathrm{eve},j},\forall k\!\in\!\mathcal{K}, j\!\in\!\mathcal{J},\notag
\end{align}
respectively. Variables $n_k \sim \mathcal{CN}(0,\sigma_k^2)$ and ${n}_{\mathrm{eve},j}\sim \mathcal{CN}(0,\sigma_{\mathrm{eve},j}^2)$ denote the background noise at the $k$-th user with noise power $\sigma_k^2$ and at the $j$-th eavesdropper with noise power $\sigma_{\mathrm{eve},j}^2$, respectively.
Accordingly, the received signal-to-interference-plus-noise ratio (SINR) at the $k$-th user, $\forall k \in \mathcal{K}$, is given by
\begin{align}
    \mathrm{SINR}_k\!\!
     =\!\!\frac{|(\mathbf{h}_{\mathrm{d},k}^\mathrm{H}+\mathbf{h}_{\mathrm{r},k}^\mathrm{H}\bm{\Theta}\mathbf{G})\mathbf{w}_k|^2}
     {\sum\limits_{i\neq k}^K\!|\!(\mathbf{h}_{\mathrm{d},k}^\mathrm{H}\!\!+\!\mathbf{h}_{\mathrm{r},k}^\mathrm{H}\bm{\Theta}\mathbf{G})\mathbf{w}_i|^2
     \!\!+\!
    |(\mathbf{h}_{\mathrm{d},k}^\mathrm{H}\!\!+\!\mathbf{h}_{\mathrm{r},k}^\mathrm{H}\bm{\Theta}\mathbf{G})\mathbf{z}|^2\!\!+\!\sigma_k^2}\text{.}
\end{align}
The achievable rate (bits/s/Hz) of the $k$-th user is given by $R_k = \log_2\left(1+\mathrm{SINR}_k\right),\forall k \in \mathcal{K}\text{.}$
For security provisioning, we adopt the worst-case assumption on the eavesdropping capabilities of potential eavesdroppers. Specifically, the potential eavesdroppers are assumed to have unlimited computational resources such that they can cancel all the multiuser interference before decoding the desired information\cite{yu2019robust}. Therefore, the channel capacity between the AP and the $j$-th potential eavesdropper for decoding the signal of the $k$-th legitimate user is given by
\begin{align}
C_{k,j}\!\!=\!\log_2\Big(1\!+\!\frac{|(\mathbf{h}_{\mathrm{ed},j}^\mathrm{H}+\mathbf{h}_{\mathrm{re},j}^\mathrm{H}\bm{\Theta}\mathbf{G})\mathbf{w}_k|^2}{|(\mathbf{h}_{\mathrm{ed},j}^\mathrm{H}+\mathbf{h}_{\mathrm{re},j}^\mathrm{H}\bm{\Theta}\mathbf{G}) \mathbf{z}|^2+\sigma^2_{\mathrm{eve},j}}\Big), \forall k, j.
\end{align}
Hence, the achievable secrecy rate between the AP and the $k$-th legitimate user is
$R_{\mathrm{s},k}\!\!=\!\!\Big[R_{k}\!\!-\!\!\underset{j\in\mathcal{J}}{\max}\{C_{k,j}\}\Big]^+$,
and the system sum secrecy rate is given by $R_{\mathrm{s}}\!\!=\!\!\sum_{k\in\mathcal{K}}R_{\mathrm{s},k}$\cite{yu2019robust,9045989}.
\subsection{IRS Model}\label{Section:IRS_MODEL}
In this paper, we consider a realistic model of the IRS to achieve self-sustainability, which contains discrete phase shifters and a non-linear energy\footnote{In this paper, the unit of energy consumption is Joule-per-second. Therefore, the terms ``power'' and ``energy'' are interchangeable.} harvesting circuit.
\begin{figure}[t] 
  \centering
  \includegraphics[width=3.4in]{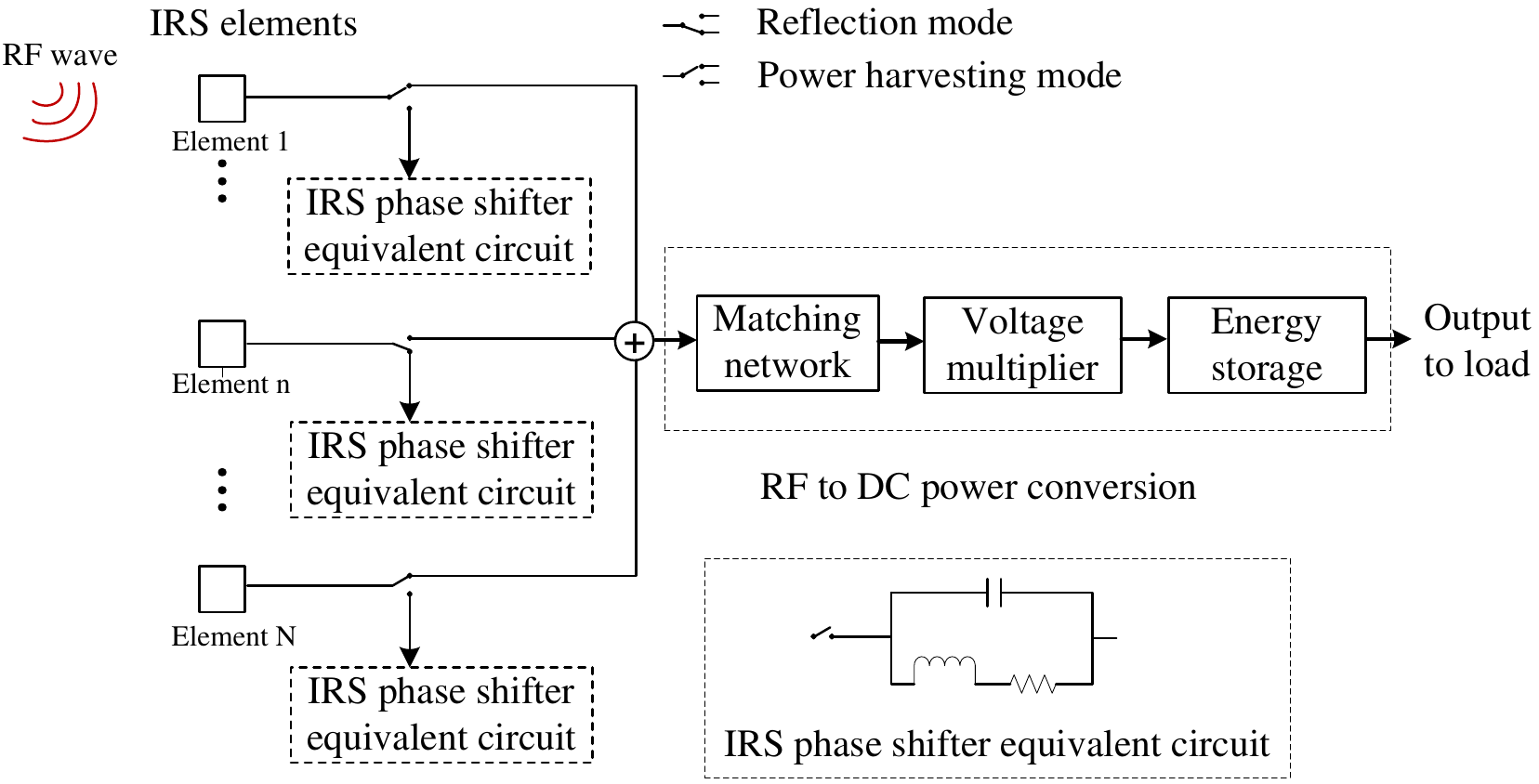}  
  \caption{{Energy harvesting block diagram of the self-sustainable IRS.}}
  \label{EH_circuit} 
\end{figure}
\subsubsection{Phase Shift Model}
The reflection matrix is defined as $\bm{\Theta}= \mathbf{A}\mathbf{\Phi}$, where $\mathbf{A}\in\mathbb{R}^{N\times N}$ is an IRS mode selection matrix and $\mathbf{\Phi} \in \mathbb{C}^{N\times N}$ is a diagonal matrix.
Matrix $\mathbf{A}= \mathrm{diag}(\alpha_{1},\ldots,\alpha_{n},\ldots,$ $\alpha_{N})\in \mathbb{B}^{N\times N}$, $\forall n \in \mathcal{N}$, and $\alpha_{n}\in\{0,1\}$ is an IRS mode selection variable which is defined as:
\begin{align}
    \alpha_{n} = \left\{
    \begin{array}{ll}
      \hspace{-2mm}  1, &\hspace{-2mm} \text{Reflection mode at IRS element} \,n\text{,}\\
      \hspace{-2mm}  0, &\hspace{-2mm} \text{Energy harvesting mode at IRS element} \,n\text{.}
    \end{array}
\right.
\end{align}
A simplified equivalent circuit of a self-sustainable IRS is shown in Fig. \ref{EH_circuit}. Specifically, several positive-intrinsic-negative (PIN) diodes are embedded in each IRS element such that the IRS elements can be switched to different modes \cite{wu2019towards}, i.e., energy harvesting mode or reflection modes\footnote{{Note that the IRS is equipped with a controller which switches the operation mode of the IRS elements between the reflection mode and power harvesting mode\cite{wu2020intelligent}.}}.
{In particular, IRS elements in reflection mode reflect all impinging signal waveforms, while the elements in energy harvesting mode harvest all the received energy carried by the signals. Once an IRS element is in reflection mode, all the received signals at that element are reflected and it cannot harvest any energy. Likewise, the elements operating in power harvesting mode do not reflect any received signal. Note that the mode selection policy will be optimized by the proposed algorithm to balance the number of energy harvesting and reflection IRS elements.}
The diagonal matrix $\mathbf{\Phi}$ is defined as $\mathbf{\Phi} =$$ \mathrm{diag}($$\beta_1 e^{j\theta_1}, $$ \ldots,$$\beta_n  e^{j\theta_n}$$ ,\ldots,\beta_N  e^{j\theta_N}) $$\in \mathbb{C}^{N\times N}$ with a phase shift $\theta_n \in [0,2\pi)$ and an amplitude coefficient $\beta_n \in [0,1], \forall n \in \mathcal{N}$. 
We assume that a discrete phase shifter is adopted in each IRS element and the phase shift interval $[0,2\pi)$ is uniformly quantized, i.e., 
$
\theta_n \in \mathcal{F} = \mathcal{\mathfrak{}} \Big\{0, \ldots,\bigtriangleup\theta,\ldots, \bigtriangleup\theta(B-1)\Big\}, \forall n \in \mathcal{N},
$
where $\mathcal{F}$ is a set of phase shifts, $\bigtriangleup\theta = 2\pi/B$, $B=2^b$ is the total number of realizable phase shift levels, and $b$ is the given constant bit resolution. For practical implementation of the IRS, the reflection coefficient $\beta_n $  is fixed to be $1$ in this paper, as commonly adopted in the literature, e.g.  \cite{wu2019intelligent,wu2019weighted,wu2019beamforming,huang2019reconfigurable}. {For the reflection mode, by varying the biasing voltage of each phase shifter via a direct current feeding line, the imping signals at each IRS element are reflected with some controlled phase shifts \cite{wu2019towards}.} {In particular, the power consumption of each IRS reflection element is related to its bit resolution \cite{mendez2016hybrid,ribeiro2018energy,huang2018energy,huang2019reconfigurable}. Thus, we denote $P_{\mathrm{IRS}}(b)$ as the amount of power consumed by each $b$-bit resolution IRS reflection element.} 
\subsubsection{Non-linear Energy Harvesting Model}

{As for the energy harvesting mode, the signals go through the radio frequency (RF)-to-direct current (DC) power conversion circuit, which consists of three parts, i.e., a matching network, a voltage multiplier, and an energy storage. In particular, the matching network enables the maximum power transmission from the antenna to a voltage multiplier.
    Also, the voltage multiplier converts the incident RF power into DC power and an energy storage guarantees the smoothed DC power delivery to the load\footnote{{Note that other advanced hardware circuit implementations of IRS are beyond the scope of the paper and interested readers may refer to \cite{abdelhalem2013rf,hameed2014hybrid} for more details.}} \cite{nintanavongsa2012design}.}
    The power consumed by the RF to DC power conversion circuit is denoted by $P_{\mathrm{c}}>0$, which is a constant value depending on the detailed circuit specifications, e.g. resistance, capacitance, and inductance. The total received RF power at the IRS is given by
\begin{align}
P_{\mathrm{PR}}=\mathbb{E}\Big(\|\mathbf{A}_{\mathrm{EH}}\Big(\mathbf{G}\Big(\sum_{k\in\mathcal{K}}\mathbf{w}_kx_k+\mathbf{z}\Big)+\mathbf{n}_a\Big)\|^2\Big),\label{y_eh}
\end{align}
where $ \mathbf{A}_{\mathrm{EH}}=\mathbf{I}_N-\mathbf{A}$ is the energy harvesting binary-valued matrix. Vector $\mathbf{n}_a\sim\mathcal{CN}(\mathbf{0},\sigma_a^2\mathbf{I}_N)$ is the thermal noise at the IRS with per IRS element noise power $\sigma_a^2$. In general, there are two tractable models for characterizing the property of RF-based energy harvesting in the literature, i.e., the linear model and the non-linear model\cite{clerckx2018fundamentals}.
However, the linear model fails to capture the characteristics of the inherent non-linearity of practical energy harvesting circuits\cite{boshkovska2017resource}. Motivated by this, we adopt a more realistic non-linear energy harvesting model proposed in \cite{boshkovska2017resource} at the IRS and it is given by
\begin{align}
P_{\mathrm{PH}}=\frac{\Psi - M_{\mathrm{P}}\Omega}{1-\Omega},
 \end{align}
 where $\Omega=\frac{1}{1+\exp(aq)}$ and $\Psi=\frac{M_{\mathrm{P}}}{1+\exp\Big(-a(P_{\mathrm{PR}}-q)\Big)}$ is a sigmoidal function whose input is the received RF power, $P_{\mathrm{PR}}$. Parameter $M_{\mathrm{P}}$ is a nonnegative constant which determines the maximum harvestable power. Also, constants $a>0$ and $q>0$ capture the joint effects of resistance, capacitance, and circuit sensitivity of the energy harvesting circuit. Note that parameters $M_{\mathrm{p}}$, $a$, and $q$ can be obtained easily via some standard curve fitting approach once the schematic of the adopted energy harvesting circuit is fixed.
\subsection{Channel State Information}
 In the literature, vaious approaches have been proposed for CSI estimation of individual links and/or cascaded links, e.g. \cite{chen2019channel,7106496}. To isolate the robust resource allocation design from any specific channel estimation design details, e.g. \cite{chen2019channel,hu2019two}, we consider the bounded CSI error model\footnote{This model guarantees the performance even for the worst case of channel estimation error, as long as $\|\Delta\mathbf{G}_{\mathrm{cu},k}\|_{\mathrm{F}}\leq\rho_{\mathrm{cu},k}$.} \cite{7106496} for the cascaded AP-IRS-user channels, the AP-IRS channel, and the direct AP-user channels. Specifically, we denote the cascaded AP-IRS-user channel for the $k$-th user as
\begin{align}
\mathbf{G}_{\mathrm{cu},k}=&\mathrm{diag}(\mathbf{h}_{\mathrm{r},k}^{\mathrm{H}})\mathbf{G} = \hat{\mathbf{G}}_{\mathrm{cu},k}+\Delta\mathbf{G}_{\mathrm{cu},k},\forall k\label{channelGcu:model}\\
            \text{ with } \mathbf{\Upsilon}_{\mathrm{cu},k} \overset{\Delta}{=} &\Big\{\!\Delta\mathbf{G}_{\mathrm{cu},k} \in\mathbb{C}^{N\times M_{\mathrm{t}}}\!\!:\|             \Delta\mathbf{G}_{\mathrm{cu},k}\|_{\mathrm{F}}\!\leq\!\rho_{\mathrm{cu},k}\Big\},\forall k,\notag
\end{align}
where $\hat{\mathbf{G}}_{\mathrm{cu},k},\forall k,$ is the estimation of the corresponding channel $\mathbf{G}_{\mathrm{cu},k}$. The channel estimation error of $\mathbf{G}_{\mathrm{cu},k}, \forall k$, is denoted by $\Delta\mathbf{G}_{\mathrm{cu},k}$. The continuous set $\mathbf{\Upsilon}_{\mathrm{cu},k},\forall k,$ collects all possible channel estimation errors, while constant $\rho_{\mathrm{cu},k},\forall k,$ denotes the maximum value of the norm of the CSI estimation error $\Delta\mathbf{G}_{\mathrm{cu},k},\forall k$.

Since the positions of AP and IRS are usually fixed, the AP-IRS channel, $\mathbf{G}$, is generally quasi-static\cite{hu2019two}. Therefore, the channel matrix, $\mathbf{G}$, can be estimated in a large timescale via existing algorithms, e.g. a dual-link pilot transmission scheme and a coordinate descent-based AP-IRS channel estimation algorithm\cite{hu2019two}. As a result, to capture the CSI imperfectness of $\mathbf{G}$, we express the AP-IRS channel as
\begin{align}
\mathbf{G}\!\!=\!\!\hat{\mathbf{G}}\!+\!\Delta\mathbf{G}
            \text{ with }\mathbf{\Upsilon}_{\mathrm{G}}\! \overset{\Delta}{=}\! \Big\{\!\Delta\mathbf{G} \!\in\!\mathbb{C}^{N\times M_{\mathrm{t}}}\!\!:\!\!\| \Delta\mathbf{G}\|_{\mathrm{F}}\!\leq\!\rho_{\mathrm{G}}\!\Big\},\label{ChannelG:error}
\end{align}
where $\hat{\mathbf{G}}$ is the estimation of channel $\mathbf{G}$, $\Delta\mathbf{G}$ represents the estimation error of $\mathbf{G}$, $\mathbf{\Upsilon}_{\mathrm{G}}$ is a continuous set which collects all the possible channel estimation errors, and $\rho_{\mathrm{G}}$ is the maximum value of norm of the  the estimation error $\Delta\mathbf{G}$. Generally, we have $\rho_{\mathrm{cu},k}\geq\rho_{\mathrm{G}},\forall k$.

Moreover, the direct AP-user channel can be estimated by exploiting conventional uplink pilot transmission \cite{hu2019two}. Similarly, we denote the direct AP-user channel for the $k$-th user as
\begin{align}
            \mathbf{h}_{\mathrm{d},k} =& \hat{\mathbf{h}}_{\mathrm{d},k}+\Delta\mathbf{h}_{\mathrm{d},k}\label{channelhd:model}\\
            \text{ with }\mathbf{\Upsilon}_{\mathrm{d},k} \overset{\Delta}{=}& \Big\{\Delta\mathbf{h}_{\mathrm{d},k} \in\mathbb{C}^{M_{\mathrm{t}}\times 1}:\| \Delta\mathbf{h}_{\mathrm{d},k}\|_2\leq\rho_{\mathrm{d},k}\Big\},\forall k,\notag
\end{align}
where $\hat{\mathbf{h}}_{\mathrm{d},k}$ is the estimate of corresponding channel $\mathbf{h}_{\mathrm{d},k}$ and the corresponding channel estimation error is $\Delta\mathbf{h}_{\mathrm{d},k}$. The continuous set $\mathbf{\Upsilon}_{\mathrm{d},k}$, collects all the possible channel estimation errors, while constant $\rho_{\mathrm{d},k}$, is the maximum value of the norm of the CSI estimation error.
Similarly, the $j$-th cascaded AP-IRS-Eve channel\footnote{Note that the CSI of active eavesdroppers can be estimated by their transmissions, while the CSI for passive eavesdroppers can be estimated through the local oscillator power leaked from the eavesdroppers' receiver RF front end\cite{mukherjee2012detecting}.} and the $j$-th direct AP-Eve channel are denoted as
$\mathbf{G}_{\mathrm{ce},j}=\mathrm{diag}(\mathbf{h}_{\mathrm{re},j}^{\mathrm{H}})\mathbf{G} = \hat{\mathbf{G}}_{\mathrm{ce},j}+\Delta\mathbf{G}_{\mathrm{ce},j}$ with
            $\mathbf{\Upsilon}_{\mathrm{ce},j} \overset{\Delta}{=}  \Big\{\Delta\mathbf{G}_{\mathrm{ce},j} \in\mathbb{C}^{N\times M_{\mathrm{r}}}:\|             \Delta\mathbf{G}_{\mathrm{ce},j}\|_{\mathrm{F}}\leq\rho_{\mathrm{ce},j}\Big\},\forall j$,
and $\mathbf{h}_{\mathrm{ed},j} =  \hat{\mathbf{h}}_{\mathrm{ed},j}+\Delta\mathbf{h}_{\mathrm{ed},j}$ with $\mathbf{\Upsilon}_{\mathrm{ed},j} \overset{\Delta}{=} \Big\{\Delta\mathbf{h}_{\mathrm{ed},j} \in\mathbb{C}^{M_{\mathrm{r}}\times 1}:\| \Delta\mathbf{h}_{\mathrm{ed},j}\|_2\leq\rho_{\mathrm{ed},j}\Big\},\forall j,$
respectively. $\hat{\mathbf{G}}_{\mathrm{ce},j},\Delta\mathbf{G}_{\mathrm{ce},j},\mathbf{\Upsilon}_{\mathrm{ce},j},\rho_{\mathrm{ce},j}$, $\hat{\mathbf{h}}_{\mathrm{ed},j},\Delta\mathbf{h}_{\mathrm{ed},j},\mathbf{\Upsilon}_{\mathrm{ed},j}$, and $\rho_{\mathrm{ed},j}$ are defined in the similar manner as the AP-IRS-user link, $\mathbf{G}_{\mathrm{cu},k}$, and the AP-user link, $\mathbf{h}_{\mathrm{d},k}$. Generally, we have $\rho_{\mathrm{ce},j}\geq\rho_{\mathrm{G}}, \forall j$.
\section{Problem Formulation}
{In this paper, we aim to maximize the system sum-rate while maintaining the self-sustainability of the IRS by jointly designing the precoding vector $\{\mathbf{w}_k\}_{k\in\mathcal{K}}$, the AN covariance matrix $\mathbf{Z}$, the mode selection indicators $\{\alpha_n\}_{n\in\mathcal{N}}$, and the discrete phase shifters $\{\theta_n\}_{n\in\mathcal{N}}$.} {The joint design can be formulated as the following optimization problem\footnote{{The considered problem formulation serves as a generalized optimization framework that subsumes various existing designs as subcases, e.g. robust precoding \cite{zhou2020robust}, secure precoding \cite{cui2019secure}, or energy-harvesting IRS \cite{lyu2020optimized}, etc.}}}:
\begin{align}
&\underset{\mathbf{w}_k,\,\mathbf{Z},\, \alpha_n,\,\theta_n}{\mathrm{maximize}} \,\, \sum_{k\in\mathcal{K}}\underset{\Delta\mathbf{h}_{\mathrm{d},k},\Delta\mathbf{G}_{\mathrm{cu},k} }{\min}\Big\{\log_2(1+\mathrm{SINR}_k)\Big\}  \label{proposed_formulation_origion} \\
&\mathrm{s.t.}\,\,\mathrm{C1}\hspace{-1mm}: \sum_{k\in\mathcal{K}}\|\mathbf{w}_k\|^2 + \mathrm{Tr}(\mathbf{Z}) \leq P_{\max},\,\,\mathrm{C2}\hspace{-1mm}: \theta_n \in \mathcal{F}, \forall n \in \mathcal{N}, \notag\\
&\hspace{6mm}\mathrm{C3}\hspace{-1mm}:  \sum_{ n \in \mathcal{N}} \alpha_{n}P_{\mathrm{IRS}}(b)+P_{\mathrm{c}} \leq \underset{\Delta\mathbf{G}}{\min}\big\{P_{\mathrm{PH}}\big\},\notag\\
&\hspace{6mm}\mathrm{C4}\hspace{-1mm}: \alpha_{n}\in \{0,1\},\forall n, \,\,\mathrm{C5}\hspace{-1mm}:\!\! \underset{\Delta\mathbf{h}_{\mathrm{ed},j},\Delta\mathbf{G}_{\mathrm{ce},j}}{\max}\! \! \{C_{k,j}\}\! \leq\! \tau_{k,j}, \forall k,j.\notag
\end{align}
The objective function in \eqref{proposed_formulation_origion} maximizes the minimum value of all the possible sum-rate due to CSI errors. Constraint $\mathrm{C1}$ ensures that the transmit power at the AP does not exceed its maximum transmit power budget $P_{\max}$. Constraint $\mathrm{C2}$ specifies that the phase shift of a $b$-bit resolution IRS reflecting element can only be selected from a discrete set $\mathcal{F}$. {Constraint $\mathrm{C3}$ indicates that the total power consumed at the IRS should not exceed its total harvested power\footnote{{If the harvested energy by the IRS is not even enough to supply one IRS reflection element, the proposed problem \eqref{proposed_formulation_origion} can be solved by ignoring constraint C3 and the related IRS cascaded links, which is essentially the case of beamforming design without the IRS.}} from the AP, $P_{\mathrm{PH}}$, while taking into account the imperfect CSI knowledge.} Constraint $\mathrm{C4}$ is imposed to guarantee that each IRS element can only operate in either reflection mode or energy harvesting mode. Constant $\tau_{k,j}$ in $\mathrm{C5}$ is the maximum tolerable information leakage to the $j$-th potential eavesdropper for wiretapping the signal transmitted to the $k$-th legitimate user considering the CSI error in the AP-Eve channel and the cascaded AP-IRS-Eve link. In particular, $\mathrm{C5}$ can guarantee that the system secrecy rate $R_\mathrm{s}$ is bounded from below, i.e. $R_\mathrm{s}\geq \sum_{k\in\mathcal{K}}\Big[R_{\mathrm{s},k}-\tau_{k,j}\Big]^+$. 
\begin{Remark}
 The proposed formulation in \eqref{proposed_formulation_origion} for maximizing the system sum-rate taking into account the information leakage constraints offers a higher flexibility in resources allocation than that of directly maximizing the system secrecy rate as in e.g. \cite{zhao2015robust,9045989}. In particular, the proposed problem formulation can control the levels of secrecy performance of individual eavesdropper via adjusting parameters $\tau_{k,j},\forall k,j$, which is a more flexible formulation for determining the levels of communication securities for heterogeneous practical applications and has been commonly adopted in the literature\cite{yu2019robust,8333737,7106496,5963524}.
 \end{Remark}
\section{Solution Of The Optimization Problem} \label{solution}
\begin{figure}[t] 
  \centering
  \includegraphics[width=3.4in]{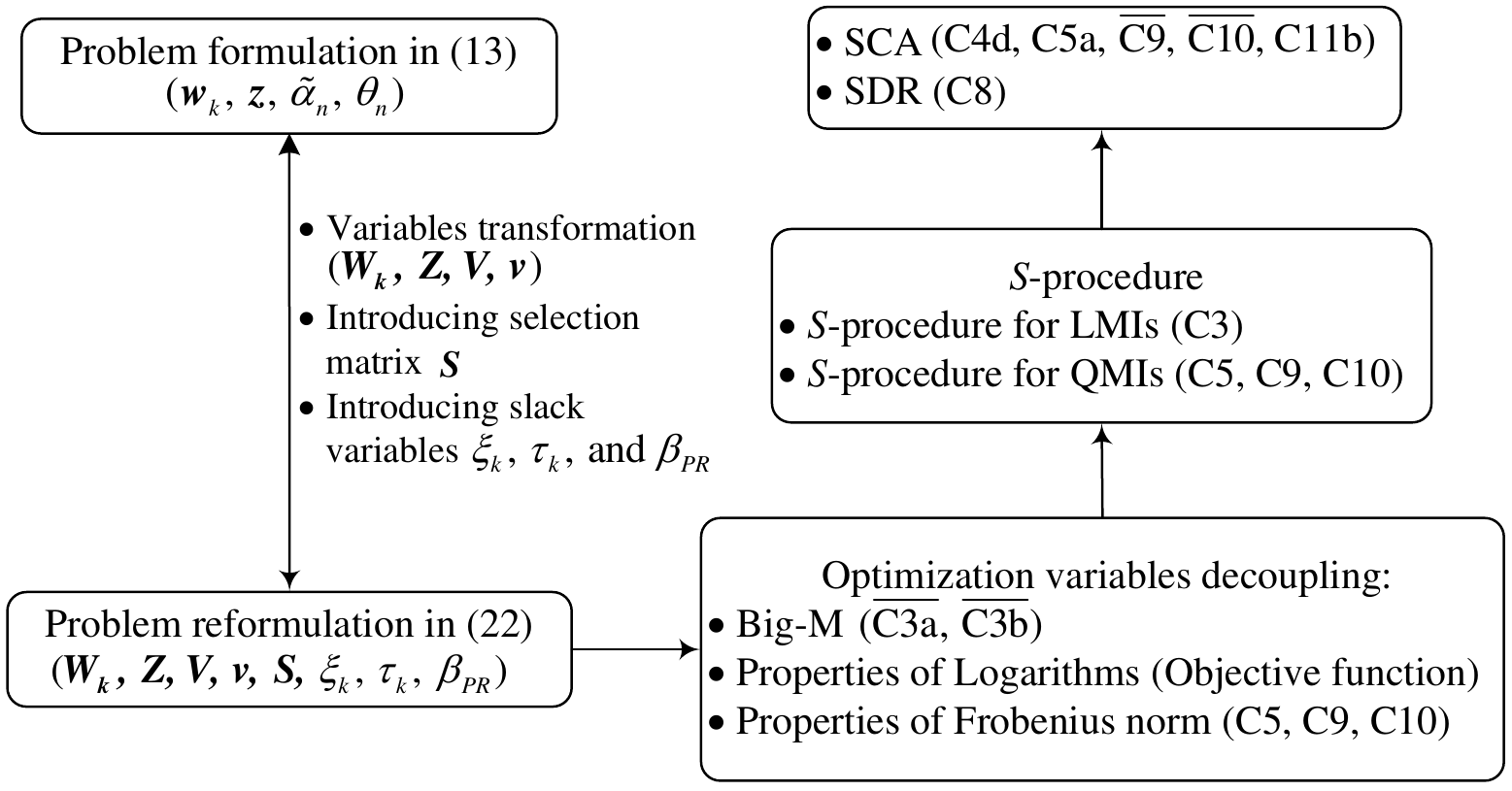}
  \caption{A flow chart of the proposed algorithm.}
  \label{flowchart2}
\end{figure}
The formulated problem is non-convex due to the coupling between variables $\mathbf{w}_k$, $\mathbf{Z}$, $\theta_n$, and $\alpha_n$ in constraints $\mathrm{C3}$, $\mathrm{C5}$, and the objective function in \eqref{proposed_formulation_origion}, the discrete phase shift in constraint $\mathrm{C2}$, and the binary variable $\alpha_n$ in constraint $\mathrm{C4}$. Besides, there are infinitely many possibilities in the objective function, constraint $\mathrm{C3}$, and constraint $\mathrm{C5}$, due to the CSI uncertainties in continuous variable sets. In general, solving such a problem optimally requires the use of the exhaustive search method which is computationally intractable even for a moderate size system. As a compromise approach, we aim to design a computationally-efficient suboptimal algorithm to obtain an effective solution.
Fig. \ref{flowchart2} summarizes the procedure of the proposed problem solving methodology. In order to address the optimization problem without using the AO method, we firstly introduce a set of slack variables and reformulate the problem in \eqref{proposed_formulation_origion} to problem in \eqref{optimization_problem_transformed}, which paves the way for decomposing the coupling variables $\{\mathbf{w}_k$, $\mathbf{Z}, \theta_n,$ $\alpha_n\}$. Secondly, we adopt the big-M reformulation to decouple the coupling between the discrete variables, i.e., the mode selection matrix, $\mathbf{S}$, and the continuous variables, i.e., precoder, $\mathbf{W}_k$, and AN matrix, $\mathbf{Z}$, in $\overline{\mathrm{C3a}}$ and $\overline{\mathrm{C3b}}$. Besides, decomposing the coupling between the continuous variables in the objective function and constraints $\mathrm{C5},\mathrm{C9}$, and $\mathrm{C10}$ requires exploiting the properties of logarithm and Frobenius norm, respectively. Thirdly, after decomposing all these coupling variables, the $\mathcal{S}$-procedure is adopted to tackle the infinite number of constraints caused by imperfect CSI in $\mathrm{C3}$, $\mathrm{C5}$, $\mathrm{C9}$, and $\mathrm{C3}$. Fourthly, we exploit the successive convex approximation (SCA) to deal with all functions in difference of convex (D.C.) form, i.e.,$\mathrm{C4d}$, $\mathrm{C5a}$, $\overline{\mathrm{C9}}$, $\overline{\mathrm{C10}}$, and $\mathrm{C11b}$. Lastly, we apply the semi-definite relaxation (SDR) method to tackle the rank-one constraint, i.e., $\mathrm{C8}$, for the precoder design. As such, a computationally efficient suboptimal solution can be obtained.
\subsection{Problem Transformation}
To facilitate the design of discrete IRS phase shifts, we first handle the coupling of model selection variables and phase shifts matrix, $\mathbf{A}\mathbf{\Phi}$, in the objective function, constraint $\mathrm{C3}$, and constraint $\mathrm{C5}$. To this end, we define an augmented mode selection matrix $\tilde{\mathbf{A}}$ to replace $\mathbf{A}\mathbf{\Phi}$. $\tilde{\mathbf{A}} = \mathrm{diag}(\tilde{\bm{\alpha}})=\mathrm{diag}(\tilde{\alpha}_1,$$ \ldots, \tilde{\alpha}_n, \ldots, $$\tilde{\alpha}_N)$, where $\tilde{\alpha}_n\in \tilde{\mathcal{F}} =\{0,e^{j0},$$\ldots,e^{j\bigtriangleup\theta},\ldots, $$\,e^{j\bigtriangleup\theta(B-1)}\}$ is the new mode selection optimization variable for the $n$-th IRS element and $\tilde{\mathcal{F}}$ is the generalized mode selection set with $B+1$ modes. When $\tilde{\alpha}_{n} = 0$, the $n$-th IRS element is in the energy harvesting mode, otherwise it is in the reflection mode. Then, we can express the received desired signal at the $k$-th user as
\begin{align}
&|(\mathbf{h}_{\mathrm{d},k}^\mathrm{H}+\mathbf{h}_{\mathrm{r},k}^\mathrm{H}\tilde{\mathbf{A}}\mathbf{G})\mathbf{w}_k|^2\notag\\
=&2\mathrm{Re}\{\tilde{\mathbf{v}}^\mathrm{H}\!\mathbf{G}_{\mathrm{cu},k}\!\mathbf{W}_k\mathbf{h}_{\mathrm{d},k}\!\}\!+\!
\tilde{\mathbf{v}}^\mathrm{H}\mathbf{G}_{\mathrm{cu},k}\!\mathbf{W}_k\mathbf{G}_{\mathrm{cu},k}^\mathrm{H}
\tilde{\mathbf{v}}\!+\!\mathbf{h}_{\mathrm{d},k}^\mathrm{H}\!\mathbf{W}_k\mathbf{h}_{\mathrm{d},k}\notag\\
=&\mathrm{Tr}(\mathbf{v}^\mathrm{H}\mathbf{F}_k^{\mathrm{H}}\mathbf{W}_k\mathbf{F}_k\mathbf{v})=\mathrm{Tr}(\mathbf{F}_k^{\mathrm{H}}\mathbf{W}_k\mathbf{F}_k\mathbf{V}),
\end{align}
where $\mathbf{v}=$$[\tilde{\mathbf{v}}^{\mathrm{H}},v]^{\mathrm{H}}$, $|v|^2=1$, $\tilde{\mathbf{v}}=[\tilde{\alpha}_1,  \ldots, \tilde{\alpha}_N]^{\mathrm{H}}$, and $\mathbf{F}_k = \begin{bmatrix}\mathbf{G}_{\mathrm{cu},k}^\mathrm{H} &\mathbf{h}_{\mathrm{d},k}\end{bmatrix}$.
Then, by defining  $\mathbf{F}_{\mathrm{eve},j} = \begin{bmatrix}\mathbf{G}_{\mathrm{ce},j}^\mathrm{H} &\mathbf{h}_{\mathrm{ed},j}\end{bmatrix}$, $\mathbf{W}=\mathbf{w}\mathbf{w}^\mathrm{H}$, and $\mathbf{V}=\mathbf{v}\mathbf{v}^\mathrm{H}$, $\mathrm{SINR}_k$ in the objective function, $C_{k,j}$ in constraints C5, constraints C3, and constraints C4 in \eqref{proposed_formulation_origion} can be equivalently rewritten as
\begin{align}
    &\mathrm{SINR}_k
    \!\! =\!\frac{\mathrm{Tr}(\mathbf{F}_k^{\mathrm{H}}\mathbf{W}_k\mathbf{F}_k\!\mathbf{V})}{\sum\limits_{i\neq k}\!\!\mathrm{Tr}(\mathbf{F}_k^{\mathrm{H}}\mathbf{W}_i\mathbf{F}_k\!\mathbf{V}\!)\!+\!
    \mathrm{Tr}(\mathbf{F}_k^{\mathrm{H}}\mathbf{Z}\mathbf{F}_k\mathbf{V})\!+\!\sigma_k^2},\forall k,\label{SINR_new}\\
&C_{k,j}=\log_2\Big(1+\frac{\mathrm{Tr}(\mathbf{F}_{\mathrm{eve},j}^{\mathrm{H}}\mathbf{W}_k\mathbf{F}_{\mathrm{eve},j}\mathbf{V})}{\mathrm{Tr}(\mathbf{F}_{\mathrm{eve},j}^{\mathrm{H}}\mathbf{Z}\mathbf{F}_{\mathrm{eve},j}\mathbf{V})+\sigma^2_{\mathrm{eve},j}}\Big),\forall k,j\text{,}\notag\\
&\mathrm{C3}\hspace{-1mm}:\sum_{n=1}^N |{v^*_n}|P_{\mathrm{IRS}}(b)\hspace{-1mm}+\hspace{-1mm}P_{\mathrm{c}} \leq\underset{\Delta\mathbf{G}}{\min}\{  P_{\mathrm{PH}}\}\text{, and }\notag\\
 &\mathrm{C4}\hspace{-1mm}:v^*_{n} \in \tilde{\mathcal{F}} =\{0,e^{j0}, e^{j\bigtriangleup\theta},\ldots, e^{j\bigtriangleup\theta(B-1)}\},\forall n\in \mathcal{N} , \notag
\end{align}
respectively, where $v_n$ is the $n$-th element of vector $\mathbf{v}$.

Next, to handle the discrete optimization variable $\tilde{\alpha}_{n}$ in C3 and C4, we further introduce a binary mode selection optimization variable $s_{i,n} \in \mathbb{B}, \forall i \in \mathcal{I}= \{1,\ldots,B+1\}, n\in \mathcal{N}$. $s_{i,n} \in \mathbf{S}$, where $\mathbf{S}\in \mathbb{B}^{(B+1)\times N}$ is a mode selection binary matrix. In particular, $s_{i,n} = 1$ indicates that the $i$-th mode is selected for the $n$-th element. Otherwise, $s_{i,n}=0$. Thus, constraint C4 can be equivalently represented as:
\begin{align}
&\mathrm{C4a}\hspace{-1mm}:\sum_{i\in\mathcal{I}}s_{i,n} \le 1, \forall n \text{,  }\mathrm{C4b}\hspace{-1mm}:v_{n} = \sum_{i\in\mathcal{I}}s_{i,n}f_i^{*}, \forall n\in \mathcal{N},\label{constraint:c4c}\\
&\text{and }\mathrm{C4c}\hspace{-1mm}:s_{i,n} \in \{0,1\}, \forall i,n,\notag
\end{align}
where $f_i$ is the $i$-th element of the generalized mode selection set $\tilde{\mathcal{F}}$ defined in \eqref{constraint:c4c}.
Furthermore, to tackle the binary constraint on $s_{i,n}$ in \eqref{constraint:c4c}, we equivalently transform constraint $\mathrm{C4c}$ into the following two continuous constraints:
\begin{align}
\!\!\mathrm{C4d}\hspace{-1mm}: s_{i,n}\!-\!s_{i,n}^2\leq0,\forall i,n\text{, and } \mathrm{C4e}\hspace{-1mm}: 0\!\leq\! s_{i,n}\!\leq\! 1,\forall i,n.\label{C4dC4e}
\end{align}
Meanwhile, the first row of mode selection binary matrix $\mathbf{S}$, i.e., $[s_{1,1},\ldots,s_{1,n},\ldots,s_{1,N}]$, represents the activity statuses of all the IRS elements on energy harvesting mode.
Therefore, the total received RF power at the IRS in \eqref{y_eh} can be equivalently rewritten as follows:
\begin{align}
\!\!\!\!P_{\mathrm{PR}}\!\!
=\!\!\!\!\sum_{n\in\mathcal{N}}\!\!\mathrm{Tr}\Big(\!\mathbf{T}_n\mathbf{G}(\!\sum_{k\in\mathcal{K}}\!\mathbf{W}_k\!\!+\!
\mathbf{Z})\mathbf{G}^{\mathrm{H}}\mathbf{T}_n^{\mathrm{H}}\!\Big) s_{1,n}\!\!+\!\sigma_a^2\!\!\sum_{n\in\mathcal{N}}\!\!s_{1,n},\label{PpR}
\end{align}
 where $\mathbf{T}_n=\mathrm{diag}(\mathbf{t}_n)$ is a $N\times N$ matrix, $\mathbf{t}_n$ is a $N\times 1$ vector with one on the $n$-th element and zeros elsewhere, i.e., $\mathbf{t}_n=[\underbrace{0\ldots}_{1\text{ to }n-1} 1\underbrace{\ldots 0}_{n+1 \text{ to } N}]^{\mathrm{T}}$. As a result, constraint $\mathrm{C3}$ can be equivalently rewritten as
\begin{align}
\mathrm{\overline{C3}}\hspace{-1mm}: &  \Big(N- \sum_{n=1}^N (s_{1,n})\Big)P_{\mathrm{IRS}}(b)\!+\!P_{\mathrm{c}}\leq\underset{\Delta\mathbf{G}}{\min}\{  P_{\mathrm{PH}}\}\Longleftrightarrow \label{constraint:barc3a}
\end{align}
\begin{align}
\mathrm{\overline{C3a}}\hspace{-1mm}:&\hspace{-1mm}-\hspace{-2mm}\sum_{n=1}^N\!\! s_{1,n}\hspace{-1mm}+\hspace{-1mm}e^{-a(\beta_{\mathrm{PR}}-q)}\Big\{N\hspace{-1mm}+\hspace{-1mm}P_{\mathrm{c}}\hspace{-1mm}+\hspace{-1mm}\frac{M_{\mathrm{P}}\Omega}{(1\hspace{-1mm}
-\hspace{-1mm}\Omega)P_{\mathrm{IRS}}(b)}\!\!- \hspace{-1mm}\sum_{n=1}^N\!\!s_{1,n}\Big\}\hspace{-1mm}
\notag\\&\leq\hspace{-1mm}\frac{ M_{\mathrm{P}}\hspace{-1mm}-\hspace{-1mm}P_{\mathrm{c}}\hspace{-1mm}-\hspace{-1mm}P_{\mathrm{IRS}}(b)N}{P_{\mathrm{IRS}}(b)},\notag\\
\mathrm{\overline{C3b}}\hspace{-1mm}:& \beta_{\mathrm{PR}}\leq  \underset{\Delta\mathbf{G}}{\min}\Big\{\sum_{n\in\mathcal{N}}\mathrm{Tr}(\mathbf{T}_n\mathbf{G}(
\sum_{k\in\mathcal{K}}\mathbf{W}_k+\mathbf{Z})\mathbf{G}^{\mathrm{H}}\mathbf{T}_n^{\mathrm{H}}) s_{1,n}\Big\}\notag\\
&+\!\sigma_a^2\!\!\sum_{n\in\mathcal{N}}s_{1,n},\label{constraint:barc3b}
\end{align}
where $\beta_{\mathrm{PR}}$ is an auxiliary optimization variable to replace $P_{\mathrm{PR}}$.

Since $\log_2 ( \cdot )$ is an increasing function, we can establish a lower bound of the objective function in \eqref{proposed_formulation_origion} by applying the following commonly adopted inequality\cite{6638604,6626322}:
\begin{align}
&\!\!\!\!\!\min_{\substack{\Delta\mathbf{G}_{\mathrm{cu},k},\\\Delta\mathbf{h}_{\mathrm{d},k}}}\hspace{-1mm}\Big\{\!\!
\log_2\Big(\!1\!+\!\!\frac{\mathrm{Tr}(\mathbf{F}_k^{\mathrm{H}}\mathbf{W}_k\mathbf{F}_k\mathbf{V})}{\!\!\sum\limits^K_{i\neq k}\mathrm{Tr}(\mathbf{F}_k^{\mathrm{H}}\mathbf{W}_i\mathbf{F}_k\!\mathbf{V}\!)\!\!+\!\!
    \mathrm{Tr}(\mathbf{F}_k^{\mathrm{H}}\mathbf{Z}\mathbf{F}_k\!\mathbf{V}\!)\!+\!\sigma^2_k}\Big)\Big\}\label{objective_lowerbound}\\
\geq&\!\log_2\!\!\Big(\!1\!\!+\!\!\frac{\underset{\substack{\Delta\mathbf{G}_{\mathrm{cu},k},\Delta\mathbf{h}_{\mathrm{d},k}}}{\min}\Big\{\mathrm{Tr}(\mathbf{F}_k^{\mathrm{H}}\mathbf{W}_k\mathbf{F}_k\mathbf{V})\Big\}}
{\!\!\underset{\substack{\Delta\mathbf{G}_{\mathrm{cu},k},\\\Delta\mathbf{h}_{\mathrm{d},k}}}{\max}\Big\{\!\!\sum\limits^K_{i\neq k}\mathrm{Tr}(\mathbf{F}_k^{\mathrm{H}}\mathbf{W}_i\mathbf{F}_k\mathbf{V})\!\!+\!\!
    \mathrm{Tr}(\mathbf{F}_k^{\mathrm{H}}\mathbf{Z}\mathbf{F}_k\mathbf{V})\!\!+\!\!\sigma^2_k\Big\}}\Big).\notag
\end{align}
By replacing the objective function in \eqref{proposed_formulation_origion} with its lower bound \eqref{objective_lowerbound}, a performance lower bound of problem \eqref{proposed_formulation_origion} can be obtained by solving the following optimization problem:
\begin{align}
&\underset{\substack{\mathbf{W}_k\in\mathbb{H}^{M_{\mathrm{t}}}, \mathbf{Z}\in\mathbb{H}^{M_{\mathrm{t}}},\\\mathbf{V}\in\mathbb{H}^{N+1},\mathbf{v}, \mathbf{S},\xi_k,\iota_k,\beta_{\mathrm{PR}}}}{\mathrm{maximize}} \,\, \sum_{k\in\mathcal{K}} \log_2\Big(1+\frac{\xi_k}{\iota_k+\sigma_k^2}\Big)  \label{optimization_problem_transformed1} \\
\mathrm{s.t.}\,\,
&\mathrm{\overline{C3a}},\mathrm{\overline{C3b}},\mathrm{C4a},\mathrm{C4b},\mathrm{C4d},\mathrm{C4e},\mathrm{C5}\notag,\\
&\mathrm{C1}\hspace{-1mm}: \sum_{k\in\mathcal{K}}\mathrm{Tr}(\mathbf{W}_k)+\mathrm{Tr}(\mathbf{Z})\leq P_{\max},\mathrm{C6}\hspace{-1mm}: \mathbf{Z}\succeq\mathbf{0},\notag\\
&\mathrm{C7}\hspace{-1mm}: \mathbf{W}_k\succeq\mathbf{0}, \forall k,\hspace{2mm}\mathrm{C8}\hspace{-1mm}: \mathrm{Rank}(\mathbf{W}_k)\leq1, \forall k,\notag\\
&
\mathrm{C9}\hspace{-1mm}: \xi_k\hspace{-1mm}+ \hspace{-1mm}\underset{\Delta\mathbf{G}_{\mathrm{cu},k},\Delta\mathbf{h}_{\mathrm{d},k}}{\min}\hspace{-1mm}\Big\{
\hspace{-1mm}-\hspace{-1mm}\mathrm{Tr}(\mathbf{F}_k^{\mathrm{H}}\mathbf{W}_k\mathbf{F}_k\mathbf{V})\Big\}\hspace{-1mm}\leq\hspace{-1mm} 0,\forall k, \notag \\
&\mathrm{C10}\hspace{-1mm}:  \hspace{-2mm}\underset{\substack{\Delta\mathbf{G}_{\mathrm{cu},k}, \\ \Delta\mathbf{h}_{\mathrm{d},k}}}{\max}
\hspace{-1mm}\Big\{\hspace{-1mm}\sum\limits^K_{i\neq k}\!\mathrm{Tr}(\mathbf{F}_k^{\mathrm{H}}\mathbf{W}_i\mathbf{F}_k\mathbf{V}\!)\!\!+ \!\!  \mathrm{Tr}(\mathbf{F}_k^{\mathrm{H}}\mathbf{Z}\mathbf{F}_k\mathbf{V})\!\Big\} \!\!\leq\!\!\iota_k ,\forall k,\notag\\
&\mathrm{C11}\hspace{-1mm}:\mathbf{V}=\mathbf{v}\mathbf{v}^{\mathrm{H}},\notag
\end{align}
where $\xi_k$ and $\iota_k$ are two slack variables. Moreover, the non-convex constraint in $\mathrm{C11}$ in \eqref{optimization_problem_transformed1} can be replaced by the following two constraints \cite{moore1997rank}:
\begin{align}
\!\!\!\!\mathrm{C11a}\!: \mathbf{Y}=\begin{bmatrix} \mathbf{V} &&\mathbf{v}\\
\mathbf{v}^{\mathrm{H}} &&1\end{bmatrix}\succeq\mathbf{0} \text{ and }\mathrm{C11b}\!: \mathrm{Rank}(\mathbf{Y})\!=\!1.
\end{align}
Therefore, the optimization problem in \eqref{optimization_problem_transformed1} is equivalent to the following problem:
\begin{align}
&\underset{\substack{\mathbf{W}_k\in\mathbb{H}^{M_{\mathrm{t}}}, \mathbf{Z}\in\mathbb{H}^{M_{\mathrm{t}}},\\ \mathbf{Y}\in\mathbb{H}^{N+2},\mathbf{V}\in\mathbb{H}^{N+1},\mathbf{v}, \mathbf{S},\xi_k,\iota_k,\beta_{\mathrm{PR}}}}{\mathrm{maximize}} \,\, \sum_{k\in\mathcal{K}} \log_2\Big(1+\frac{\xi_k}{\iota_k+\sigma_k^2}\Big)  \label{optimization_problem_transformed} \\
&\mathrm{s.t.}\,\,\mathrm{C1},\mathrm{\overline{C3a}},
\mathrm{\overline{C3b}},\mathrm{C4a},\mathrm{C4b},\mathrm{C4d},
\mathrm{C4e},\mathrm{C5}-\mathrm{C11b}.\notag
\end{align}

\subsection{Optimization Variables Decoupling}
For the optimization problem in \eqref{optimization_problem_transformed}, the non-convexity arises from constraints $\mathrm{\overline{C3a}}$, $\mathrm{\overline{C3b}}$, $\mathrm{C5}$, $\mathrm{C9}$, $\mathrm{C10}$, and the objective function due to severely coupled variables. In particular, there are three types of variable couplings, i.e., the fractional form of two continuous variables in SINR, the multiplication between a binary and a continuous variable, and the multiplication between $\mathbf{W}_k$, $\mathbf{Z}$, and $\mathbf{V}$, e.g. $\mathrm{Tr}(\mathbf{F}_k^{\mathrm{H}}\mathbf{Z}\mathbf{F}_k\mathbf{V})$ and $\mathrm{Tr}(\mathbf{F}_k^{\mathrm{H}}\mathbf{W}_k\mathbf{F}_k\mathbf{V})$. To address the first type of coupling, we rewrite the $\mathrm{SINR}$ in \eqref{optimization_problem_transformed} in its equivalent form of D.C. form\cite{yu2019robust}:
\begin{align}
\underset{\substack{\mathbf{W}_k\in\mathbb{H}^{M_{\mathrm{t}}}, \mathbf{Z}\in\mathbb{H}^{M_{\mathrm{t}}},\\\mathbf{Y}\in\mathbb{H}^{N+2},\mathbf{V}\in\mathbb{H}^{N+1},\mathbf{v}, \\ \mathbf{S},\xi_k,\iota_k,\beta_{\mathrm{PR}}}}{\mathrm{maximize}} \,\, \sum_{k\in\mathcal{K}} \log_2\Big(1+\frac{\xi_k}{\iota_k+\sigma_k^2}\Big)\hspace{-2mm}\notag\\ \Leftrightarrow
\underset{\substack{\mathbf{W}_k\in\mathbb{H}^{M_{\mathrm{t}}}, \mathbf{Z}\in\mathbb{H}^{M_{\mathrm{t}}},\\\mathbf{Y}\in\mathbb{H}^{N+2},\mathbf{V}\in\mathbb{H}^{N+1},\mathbf{v}, \\ \mathbf{S},\xi_k,\iota_k,\beta_{\mathrm{PR}}}}{\mathrm{maximize}} \,\,N_{\mathrm{obj}}(\xi_k,\iota_k)-D_{\mathrm{obj}}(\iota_k), \label{objective_ND}
\end{align}
where $N_{\mathrm{obj}}(\xi_k,\iota_k) = \sum_{k\in\mathcal{K}} \log_2(\xi_k+\iota_k+\sigma_k^2) \text{ and }
D_{\mathrm{obj}}(\iota_k)= \sum_{k\in\mathcal{K}} \log_2(\iota_k+\sigma_k^2)$
are two functions which are both concave with respect to $\xi_k$ and $\iota_k$, respectively.

Then, by introducing two auxiliary optimization variables $U_n^{\mathrm{ES}}$, $U_{n}^{\mathrm{BS}}$, we adopt the big-M reformulation\cite{sun2018robust} to convert the logical constraints $\overline{\mathrm{C3a}}$ and $\overline{\mathrm{C3b}}$ to a set of equivalent constraints as follows:
\begin{align}
&\mathrm{ \overline{\overline{C3a}}}\hspace{-1mm}:\hspace{-1mm}-\hspace{-1mm}\sum_{n=1}^N (s_{1,n})\hspace{-1mm}+\hspace{-1mm}e^{\Big(\hspace{-1mm}-a(\beta_{\mathrm{PR}}-q)\Big)}\hspace{-1mm}\Big\{\hspace{-1mm}N\hspace{-1mm}+\hspace{-1mm}P_{\mathrm{c}}\hspace{-1mm}+\hspace{-1mm}\frac{M_{\mathrm{P}}\Omega}{(1\hspace{-1mm}-\hspace{-1mm}\Omega)P_{\mathrm{IRS}}(b)}\Big\}\hspace{-1mm}
\hspace{3mm}\notag\\
&\hspace{7mm} -\sum_{n=1}^N U_n^{\mathrm{ES}}\hspace{-1mm}\leq\hspace{-1mm}\frac{ M_{\mathrm{P}}\hspace{-1mm}-\hspace{-1mm}P_{\mathrm{c}}\hspace{-1mm}-\hspace{-1mm}P_{\mathrm{IRS}}(b)N}{P_{\mathrm{IRS}}(b)},\label{C3aoverlineoverline}\\
&\mathrm{\overline{\overline{C3b}}}\hspace{-1mm}: \beta_{\mathrm{PR}}\hspace{-1mm}\leq\hspace{-2mm}\sum_{n\in\mathcal{N}}\hspace{-1mm}U_{n}^{\mathrm{BS}}
\hspace{-1mm}+\hspace{-1mm}\sigma_a^2\sum_{n=1}^{N}s_{1,n},
\mathrm{\overline{C3c}}\hspace{-1mm}:U_n^{\mathrm{ES}}\leq e^{\Big(-a(\beta_{\mathrm{PR}}-q)\Big)}, \notag\\
& \mathrm{\overline{C3d}}\hspace{-1mm}:  U_n^{\mathrm{ES}}\geq0, \forall n,\mathrm{\overline{C3e}}\hspace{-1mm}:U_n^{\mathrm{ES}} \leq s_{1,n}M_{\mathrm{Big}}, \forall n,\notag\\
&
\mathrm{\overline{C3f}}\hspace{-1mm}:  {U}_{n}^{\mathrm{BS}} \leq\underset{\Delta\mathbf{g}_n}{\min}\Big\{\mathrm{Tr}(\mathbf{T}_n\mathbf{G}(\sum_{k\in\mathcal{K}}\mathbf{W}_k+\mathbf{Z})\mathbf{G}^{\mathrm{H}}\mathbf{T}_n^{\mathrm{H}}) \Big\}, \forall n,\notag\\
&\mathrm{\overline{C3g}}\hspace{-1mm}: {U}_{n}^{\mathrm{BS}}\geq0,\forall n \text{, and }\mathrm{\overline{C3h}}\hspace{-1mm}: {U}_{n}^{\mathrm{BS}}\leq s_{1,n}M_{\mathrm{Big}}, \forall n,\notag
\end{align}
where $M_{\mathrm{Big}}\gg1$ is a sufficiently large constant number and $\Delta\mathbf{g}_n$$ =$$ \mathrm{vec}(\mathbf{T}_n\Delta\mathbf{G})$ denotes the channel estimation error for the channel between the AP and the $n$-th IRS element.

Note that according to the definition of $\mathbf{T}_n$ in \eqref{PpR}, the term $\mathrm{Tr}(\mathbf{T}_n\mathbf{G}(\sum_{k\in\mathcal{K}}\mathbf{W}_k+\mathbf{Z})\mathbf{G}^{\mathrm{H}}\mathbf{T}_n^{\mathrm{H}}) s_{1,n},\forall n$ in \eqref{constraint:barc3b} only depends on the $n$-th row of the channel matrix $\mathbf{G}$, which corresponds to the channel vector between the AP and the $n$-th IRS element. Since the channel uncertainty of the AP to IRS elements links are assumed independent with each other, we have $\underset{\Delta\mathbf{G}}{\min}\Big\{\sum_{n\in\mathcal{N}}\mathrm{Tr}(\mathbf{T}_n\mathbf{G}(\sum_{k\in\mathcal{K}}\mathbf{W}_k+\mathbf{Z})\mathbf{G}^{\mathrm{H}}\mathbf{T}_n^{\mathrm{H}}) s_{1,n}\Big\}\Leftrightarrow \sum_{n\in\mathcal{N}}\underset{\Delta\mathbf{g}_n}{\min}\Big\{\mathrm{Tr}(\mathbf{T}_n\mathbf{G}(\sum_{k\in\mathcal{K}}\mathbf{W}_k+\mathbf{Z})\mathbf{G}^{\mathrm{H}}\mathbf{T}_n^{\mathrm{H}}) s_{1,n}\Big\}$.

Lastly, we adopt the transformation method in \cite{yu2020power} to decompose the coupled variables between $\{\mathbf{W}_k,\mathbf{V}\}$ and $\{\mathbf{Z},\mathbf{V}\}$ in the objective function and constraints $\mathrm{C5}$, $\mathrm{C9}$, and $\mathrm{C10}$ in \eqref{optimization_problem_transformed}. For illustration, we take constraint $\mathrm{C9}$ in \eqref{optimization_problem_transformed} as an example. The coupled terms on the left hand side of constraint $\mathrm{C9}$ can be rewritten to two separated variables as follows:
\begin{align}
&-\mathrm{Tr}(\mathbf{W}_k\mathbf{F}_k\mathbf{V}\mathbf{F}_k^{\mathrm{H}})\\
&
=\frac{1}{2}\|\mathbf{W}_k\!\!-\!\!\mathbf{F}_k\!\mathbf{V}\mathbf{F}_k^{\mathrm{H}}\!\|^2_{\mathrm{F}}\!\!-\!\!
\frac{1}{2}\mathrm{Tr}(\!\mathbf{W}_k^\mathrm{H}\mathbf{W}_k)\!\!-\!\!\frac{1}{2}
\mathrm{Tr}(\!\mathbf{F}_k\!\mathbf{V}^{\mathrm{H}}\mathbf{F}_k^{\mathrm{H}}\mathbf{F}_k\!\mathbf{V}\mathbf{F}_k^{\mathrm{H}}\!).\notag
\end{align}
Therefore, constraint $\mathrm{C9}$ can be rewritten as the following equivalent form without couplings:
\begin{align}
\mathrm{C9}\hspace{-1mm}:&\,\,\xi_k\hspace{-1mm}+\hspace{-1mm} \underset{\substack{\Delta \mathbf{h}_{\mathrm{d},k},  \\ \Delta\mathbf{G}_{\mathrm{cu},k}}}{\min}\hspace{-1mm}\Big\{\frac{\|\mathbf{W}_k\hspace{-1mm}-\hspace{-1mm}\mathbf{F}_k\mathbf{V}\mathbf{F}_k^{\mathrm{H}}\|^2_{\mathrm{F}}
\hspace{-1mm}
-\hspace{-1mm}\mathrm{Tr}(\mathbf{F}_k\mathbf{V}^{\mathrm{H}}\mathbf{F}_k^{\mathrm{H}}\mathbf{F}_k\mathbf{V}\mathbf{F}_k^{\mathrm{H}})}{2}
\Big\}\hspace{-1mm}\notag\\&-\hspace{-1mm}\frac{1}{2}\mathrm{Tr}(\mathbf{W}_k^\mathrm{H}\mathbf{W}_k)\leq 0,\forall k.
\end{align}
By using the same method, constraints $\mathrm{C5}$ and $\mathrm{C10}$ can be rewritten as
\begin{align}
\mathrm{C5}\hspace{-1mm}:&\hspace{-1mm}\underset{\substack{\Delta \mathbf{h}_{\mathrm{ed},j}, \\ \Delta\mathbf{G}_{\mathrm{ce},j}}}{\min}\Big[\frac{1}{2}\|\mathbf{W}_k+\mathbf{F}_{\mathrm{eve},j}\mathbf{V}\mathbf{F}_{\mathrm{eve},j}^{\mathrm{H}}\|^2_{\mathrm{F}}\hspace{-1mm}-\hspace{-1mm}\frac{1}{2}\mathrm{Tr}(\mathbf{W}_k^{\mathrm{H}}\mathbf{W}_k)
\notag\\&-
\frac{1}{2}\mathrm{Tr}(\mathbf{F}_{\mathrm{eve},j}\mathbf{V}^{\mathrm{H}}\mathbf{F}_{\mathrm{eve},j}^{\mathrm{H}}\mathbf{F}_{\mathrm{eve},j}\mathbf{V}\mathbf{F}_{\mathrm{eve},j}^{\mathrm{H}})\notag\\
&+(2^{\tau_{k,j}}-\hspace{-1mm}1)\Big\{\frac{1}{2}\|\mathbf{Z}-\mathbf{F}_{\mathrm{eve},j}\mathbf{V}\mathbf{F}_{\mathrm{eve},j}^{\mathrm{H}}\|^2_{\mathrm{F}}-\frac{1}{2}\mathrm{Tr}(\mathbf{Z}^{\mathrm{H}}\mathbf{Z})
\notag\\-
&\frac{1}{2}\mathrm{Tr}(\mathbf{F}_{\mathrm{eve},j}\mathbf{V}^{\mathrm{H}}\mathbf{F}_{\mathrm{eve},j}^{\mathrm{H}}\mathbf{F}_{\mathrm{eve},j}\mathbf{V}\mathbf{F}_{\mathrm{eve},j}^{\mathrm{H}}) \Big\}\Big]\notag\\
&-(2^{\tau_{k,j}}-1)\sigma^2_{\mathrm{eve},j}\leq0, \forall k,j\text{, and}
\end{align}
\begin{align}
\mathrm{C10}\hspace{-1mm}:&\hspace{-2mm} \underset{\Delta \mathbf{h}_{\mathrm{d},k},\Delta\mathbf{G}_{\mathrm{cu},k}}{\max}\Big\{\sum_{i\neq k}\{\frac{1}{2}\|\mathbf{W}_i+\mathbf{F}_k\mathbf{V}\mathbf{F}_k^{\mathrm{H}}\|^2_{\mathrm{F}}\notag\\
&-\frac{K}{2}\mathrm{Tr}(\mathbf{F}_k\mathbf{V}^{\mathrm{H}}\mathbf{F}_k^{\mathrm{H}}\mathbf{F}_k\mathbf{V}\mathbf{F}_k^{\mathrm{H}})
\}+ \frac{1}{2}\|\mathbf{Z}+\mathbf{F}_k\mathbf{V}\mathbf{F}_k^{\mathrm{H}}\|^2_{\mathrm{F}}
\Big\}\notag\\
&-\sum_{i\neq k}\frac{1}{2}\mathrm{Tr}(\mathbf{W}_i^\mathrm{H}\mathbf{W}_i)-\frac{1}{2}\mathrm{Tr}(\mathbf{Z}^\mathrm{H}\mathbf{Z}) \leq\iota_k ,\forall k,
\end{align}
 respectively.
\subsection{$\mathcal{S}$-Procedure}
On the other hand, due to the CSI uncertainty, there are infinity possibilities for constraints $\mathrm{\overline{C3f}}$, $\mathrm{C5}$, $\mathrm{C9}$, and $\mathrm{C10}$. As a result, we adopt the $\mathcal{S}$-procedure\cite{boyd2004convex,luo2004multivariate} to tackle these issues.
Firstly, to handle the challenge in constraints $\mathrm{\overline{C3f}}$, we apply the following lemma to convert $\mathrm{\overline{C3f}}$ into a finite number of linear matrix inequalities (LMIs).
\begin{Lem}\label{Lemma:1}($\mathcal{S}$-Procedure\cite{boyd2004convex}): Let a function $f_m(\mathbf{x}), m\in \{1,2\}$, be defined as
\begin{align}
f_m(\mathbf{x})=\mathbf{x}^{\mathrm{H}}\mathbf{A}_m\mathbf{x}+2\mathrm{Re}\{\mathbf{b}_m^{\mathrm{H}}\mathbf{x}\}+c_m,
\end{align}
where $\mathbf{A}_m \in \mathbb{H}^N, \mathbf{b}_m\in \mathbb{C}^{N\times 1}$, and $c_m\in \mathbb{R}$. Then, the implication $f_1(\mathbf{x})\geq0\Longrightarrow f_2(\mathbf{x})\geq0$ holds if and only if there exists an $\epsilon \geq0$ such that
\begin{align}
\begin{bmatrix}
\mathbf{A}_2 & \mathbf{b}_2 \\
\mathbf{b}_2^{\mathrm{H}} & c_2
\end{bmatrix}-\epsilon\begin{bmatrix}
\mathbf{A}_1 & \mathbf{b}_1 \\
\mathbf{b}_1^{\mathrm{H}} & c_1
\end{bmatrix}\succeq\mathbf{0}
\end{align}
provided that there exists a point $\hat{\mathbf{x}}$ such that $f_k(\hat{\mathbf{x}})<0$.\end{Lem}
To apply \textbf{Lemma \ref{Lemma:1}} to constraint $\mathrm{\overline{C3f}}$, we rewrite the channel uncertainty \eqref{ChannelG:error} and reformulate constraint $\mathrm{\overline{C3f}}$. In particular,
\begin{align}
&\Delta\mathbf{g}_n^{\mathrm{H}}\Delta\mathbf{g}_n\leq\rho_{\mathrm{G}}^2/N
\Rightarrow \\& \Delta\mathbf{g}_n^{\mathrm{H}}\bm{\mathcal{W}}\Delta\mathbf{g}_n+2\mathrm{Re}\{\hat{\mathbf{g}}_n^{\mathrm{H}}\bm{\mathcal{W}}\Delta\mathbf{g}_n\}+\hat{\mathbf{g}}_n^{\mathrm{H}}\bm{\mathcal{W}}\hat{\mathbf{g}}_n-{U}_{n}^{\mathrm{BS}}\geq0, \forall n,\notag
\end{align}
holds if and only if there exist $\epsilon^{\mathrm{C3f}}_{n}\geq0,\forall n$, such that the following LMI constraints hold:
\begin{align}
&\hspace{-2mm}\mathrm{\overline{\overline{C3fa}}}\hspace{-1mm}: {\mathcal{S}}_{\mathrm{C3f},n}(\bm{\mathcal{W}},s_{1,n},\epsilon^{\mathrm{C3f}}_{n}, {U}_{n}^{\mathrm{BS}})\\&
\hspace{5mm}\overset{\Delta}{=}\hspace{-1mm}
\begin{bmatrix}
 -\epsilon^{\mathrm{C3f}}_{n}\mathbf{I}_{NM_{\mathrm{t}}}&\mathbf{0}\\[-2mm]
 {0}&\hspace{-3mm}-{U}_{n}^{\mathrm{BS}}+\epsilon^{\mathrm{C3f}}_{n}\frac{\rho_{\mathrm{G}}^2}{N}
 \end{bmatrix}+\mathbf{O}_{\hat{\mathbf{g}}_n}^{\mathrm{H}}\bm{\mathcal{W}}\mathbf{O}_{\hat{\mathbf{g}}_n}\succeq\mathbf{0}, \forall n,\notag\\
&\hspace{-2mm}\mathrm{\overline{\overline{C3fb}}}\hspace{-1mm}:\epsilon^{\mathrm{C3f}}_{n}\geq0,\forall n,
\end{align}
where $\hat{\mathbf{g}}_n$$ =$$ \mathrm{vec}(\mathbf{T}_n\hat{\mathbf{G}})$, $\bm{\mathcal{W}}$$=$$\mathbf{I}_{N}\otimes (\sum_{k\in\mathcal{K}}\mathbf{W}_k+ \mathbf{Z})$, and $\mathbf{O}_{\hat{\mathbf{g}}_n} = \begin{bmatrix}\mathbf{I}_{NM_t} &\hat{\mathbf{g}}_n \end{bmatrix}$.

Then, for constraints $\mathrm{C5}$, $\mathrm{C9}$,  and $\mathrm{C10}$, there are an infinitely number of quadratic matrix inequalities (QMIs)\cite{sun2016multi,luo2004multivariate} due to the CSI uncertainty set. To overcome the problem, we first introduce slack optimization variables $\Big\{\mathbf{\Psi}_{\mathrm{C9},k},\forall k\Big\}$, $\Big\{\mathbf{\Psi}_{\mathrm{C10},k},\forall k\Big\}$, and $\Big\{\mathbf{\Psi}_{\mathrm{eve},j},\forall j\Big\}$ to replace the $\mathbf{F}_k\mathbf{V}\mathbf{F}_k^{\mathrm{H}}$ term in constraints $\mathrm{C9}$ and $\mathrm{C10}$ and the $\mathbf{F}_{\mathrm{eve},j}\mathbf{V}\mathbf{F}_{\mathrm{eve},j}^{\mathrm{H}}$ term in constraint $\mathrm{C5}$. Hence, constraints   $\mathrm{C5}$, $\mathrm{C9}$,  and $\mathrm{C10}$ can be equivalently rewritten as:
\begin{align}
&\mathrm{C5a}\hspace{-1mm}:\frac{1}{2}\|\mathbf{W}_k+\mathbf{\Psi}_{\mathrm{eve},j}\|^2_{\mathrm{F}}-\frac{1}{2}\mathrm{Tr}(\mathbf{W}_k^{\mathrm{H}}\mathbf{W}_k)
\label{constriant:c5ac9c10}\\
&\hspace{8mm}-2^{\tau_{k,j}-1}\mathrm{Tr}(\mathbf{\Psi}_{\mathrm{eve},j}^{\mathrm{H}}\mathbf{\Psi}_{\mathrm{eve},j})\notag\\
&\hspace{8mm}+(2^{\tau_{k,j}}-1)\Big\{\frac{1}{2}\|\mathbf{Z}-\mathbf{\Psi}_{\mathrm{eve},j} \|^2_{\mathrm{F}}-\frac{1}{2}\mathrm{Tr}(\mathbf{Z}^{\mathrm{H}}\mathbf{Z})\Big\}
\notag\\
&\hspace{8mm}-(2^{\tau_{k,j}}-1)\sigma^2_{\mathrm{eve},j}\leq0,\forall k,j,\notag\\
&\mathrm{C5b}\hspace{-1mm}: \mathbf{\Psi}_{\mathrm{eve},j}\succeq \underset{\Delta\mathbf{G}_{\mathrm{ce},j},\Delta\mathbf{h}_{\mathrm{ed},j}}{\max} \Big\{\mathbf{F}_{\mathrm{eve},j}\mathbf{V}\mathbf{F}_{\mathrm{eve},j}^{\mathrm{H}}\Big\},\forall j,\notag\\
&\overline{\mathrm{C9}}\hspace{-1mm}:\xi_k+ \frac{1}{2}\|\mathbf{W}_k-\mathbf{\Psi}_{\mathrm{C9},k}\|^2_{\mathrm{F}}
-\frac{1}{2}\mathrm{Tr}(\mathbf{\Psi}_{\mathrm{C9},k}^{\mathrm{H}}\mathbf{\Psi}_{\mathrm{C9},k})
\notag\\
&\hspace{8mm}-\frac{1}{2}\mathrm{Tr}(\mathbf{W}_k^\mathrm{H}\mathbf{W}_k)\leq 0,\forall k,\notag\\
&\overline{\mathrm{C10}}\hspace{-1mm}:  \sum_{i\in\mathcal{K}/ \{k\}}\Big\{\frac{1}{2}\|\mathbf{W}_i+\mathbf{\Psi}_{\mathrm{C10},k}\|^2_{\mathrm{F}}\Big\}-\frac{K}{2}\mathrm{Tr}(\mathbf{\Psi}_{\mathrm{C10},k}^{\mathrm{H}}\mathbf{\Psi}_{\mathrm{C10},k})
\notag\\
&\hspace{8mm}+ \frac{1}{2}\|\mathbf{Z}+\mathbf{\Psi}_{\mathrm{C10},k}\|^2_{\mathrm{F}}-\sum_{i\in\mathcal{K}/ \{k\}}\frac{1}{2}\mathrm{Tr}(\mathbf{W}_i^\mathrm{H}\mathbf{W}_i)\notag\\
&\hspace{8mm}-\frac{1}{2}\mathrm{Tr}(\mathbf{Z}^\mathrm{H}\mathbf{Z}) \leq\iota_k ,\forall k, \notag\\
&\mathrm{C12}\hspace{-1mm}:\mathbf{\Psi}_{\mathrm{C9},k}\hspace{-1mm}\preceq \hspace{-1mm} \underset{\Delta\mathbf{G}_{\mathrm{cu},k},\Delta\mathbf{h}_{\mathrm{d},k}}{\min}\{\mathbf{F}_k\mathbf{V}\mathbf{F}_k^{\mathrm{H}}\},\forall k\text{, and }
\notag\\
&\mathrm{C13}\hspace{-1mm}:\mathbf{\Psi}_{\mathrm{C10},k}\hspace{-1mm}\succeq \hspace{-1mm} \underset{\Delta\mathbf{G}_{\mathrm{cu},k},\Delta\mathbf{h}_{\mathrm{d},k}}{\max}\{\mathbf{F}_k\mathbf{V}\mathbf{F}_k^{\mathrm{H}}\},\forall k,\label{constriant:c13}\notag
\end{align}
respectively. Note that constraints $\mathrm{C5b},\mathrm{C12}$, and $\mathrm{C13}$ still involve an infinite number of inequality constraints. To circumvent this difficulty, we convert the infinity number of constraints $\mathrm{C5b},\mathrm{C12}$, and $\mathrm{C13}$ to an equivalent form with only a finite number of constraints by applying \textbf{Lemma 2}.
\begin{Lem}\label{Lemma:QMI}
(Generalized S-procedure\cite{luo2004multivariate}): Let $f(\mathbf{x})=\mathbf{X}^{\mathrm{H}}\mathbf{D}\mathbf{X}+\mathbf{X}^{\mathrm{H}}\mathbf{B}+\mathbf{B}^{\mathrm{H}}\mathbf{X}+\mathbf{E}$ and $\mathbf{J}\succeq\mathbf{0}$. For some $\epsilon\geq0, f(\mathbf{X})\succeq\mathbf{0},\forall \mathbf{X}\in\Big\{\mathbf{X}|\mathrm{Tr}(\mathbf{J}\mathbf{X}\mathbf{X}^{\mathrm{H}})\leq1\Big\}$, is equivalent to
\begin{align}
\begin{bmatrix}
\mathbf{E}&\mathbf{B}^{\mathrm{H}}\\
\mathbf{B}&\mathbf{D}
\end{bmatrix}-\epsilon\begin{bmatrix}
\mathbf{I}&\mathbf{0}\\
\mathbf{0}&-\mathbf{J}
\end{bmatrix}\succeq\mathbf{0}.
\end{align}
\end{Lem}
Let us take constraint $\mathrm{C12}$ as an example. After substituting \eqref{channelGcu:model} and \eqref{channelhd:model} into \eqref{constriant:c13}, for $\Delta\mathbf{F}_k\in$$\Big\{\Delta\mathbf{F}_k|\mathrm{Tr}(\frac{1}{\rho_{\mathrm{cu},k}^2+\rho_{\mathrm{d},k}^2}\Delta\mathbf{F}_k\Delta\mathbf{F}_k^{\mathrm{H}})\leq 1\}$, constraint $\mathrm{C12}$ is expressed as
\begin{align}
\mathrm{C12}\hspace{-1mm}: & \Delta\mathbf{F}_k\mathbf{V}\Delta\mathbf{F}_k^{\mathrm{H}}+\Delta\mathbf{F}_k\mathbf{V}\hat{\mathbf{F}}_k^{\mathrm{H}}+
\hat{\mathbf{F}}_k\mathbf{V}\Delta\mathbf{F}_k^{\mathrm{H}}\notag\\
&+\hat{\mathbf{F}}_k\mathbf{V}\hat{\mathbf{F}}_k^{\mathrm{H}}-\mathbf{\Psi}_{\mathrm{C9},k}\succeq\mathbf{0},\forall k,
\end{align}
where $\hat{\mathbf{F}}_k=\begin{bmatrix}\hat{\mathbf{G}}_{\mathrm{cu},k}^{\mathrm{H}}&\hat{\mathbf{h}}_{\mathrm{d},k}\end{bmatrix}$, $\forall k$ and $\Delta\mathbf{F}_k=\begin{bmatrix}\Delta\mathbf{G}_{\mathrm{cu},k}^{\mathrm{H}}&\Delta\mathbf{h}_{\mathrm{d},k}\end{bmatrix}$, $\forall k$. Then, by applying Lemma \ref{Lemma:QMI}, constraint $\mathrm{C12}$ is equivalently transformed as
\begin{align}
&\overline{\mathrm{C12a}}\hspace{-1mm}: {\mathcal{S}}_{\mathrm{C12}}(\mathbf{V},\epsilon^{\mathrm{\overline{C12}}}_{k})=\hspace{-1mm}\\
&
\begin{bmatrix}\hat{\mathbf{F}}_k\mathbf{V}\hat{\mathbf{F}}_k^{\mathrm{H}}\hspace{-1mm}-\hspace{-1mm}\mathbf{\Psi}_{\mathrm{C9},k}\hspace{-1mm}-\hspace{-1mm}\epsilon^{\mathrm{\overline{C12}}}_{k}\mathbf{I}_{M_\mathrm{t}}
&&\hat{\mathbf{F}}_k\mathbf{V}\\
\mathbf{V}\hat{\mathbf{F}}_k^{\mathrm{H}}&&\hspace{-9mm}\mathbf{V}\hspace{-1mm}+\hspace{-1mm}\epsilon^{\mathrm{\overline{C12}}}_{k}\frac{1}{\rho_{\mathrm{cu},k}^2+\rho_{\mathrm{d},k}^2}\mathbf{I}_{N+1}\end{bmatrix}\!\!\succeq\!\!\mathbf{0}, \forall k,\notag\\
&\overline{\mathrm{C12b}}\hspace{-1mm}: \epsilon^{\mathrm{\overline{C12}}}_{k}\geq0,\forall k.\notag
\end{align}
Similarly, constraints $\mathrm{C5b}$ and $\mathrm{C13}$ can be equivalently represented as
\begin{align}
&\overline{\mathrm{C5b}}\hspace{-1mm}: {\mathcal{S}}_{\mathrm{C5b}}(\mathbf{V},\epsilon^{\mathrm{\overline{C5b}}}_{j})=\hspace{-1mm}\\
&
\begin{bmatrix}\!-\!\hat{\mathbf{F}}_{\mathrm{eve},j}\!\mathbf{V}\hat{\mathbf{F}}_{\mathrm{eve},j}^{\mathrm{H}}\hspace{-1mm}+\hspace{-1mm}\mathbf{\Psi}_{\mathrm{eve},j}\hspace{-1mm}-\hspace{-1mm}\epsilon^{\mathrm{\overline{C5b}}}_{j}\mathbf{I}_{\mathrm{M}_t}
&&\!-\hat{\mathbf{F}}_{\mathrm{eve},j}\mathbf{V}\notag\\
-\mathbf{V}\hat{\mathbf{F}}_{\mathrm{eve},j}^{\mathrm{H}}&&\hspace{-8mm}\!\!\!\!\frac{\epsilon^{\mathrm{\overline{C5b}}}_{j}}{\rho_{\mathrm{ce},j}^2+\rho_{\mathrm{ed},j}^2}\mathbf{I}_{N+1}\hspace{-1mm}-\hspace{-1mm}\mathbf{V}\end{bmatrix}\!\succeq\!\mathbf{0}, \forall j,\\
&\overline{\mathrm{C5c}}\hspace{-1mm}: \epsilon^{\mathrm{\overline{C5b}}}_{j}\geq0, \forall j,\notag\\
&\overline{\mathrm{C13a}}\hspace{-1mm}: \bm{\mathcal{S}}_{\mathrm{C13}}(\mathbf{V},\epsilon^{\mathrm{\overline{C13}}}_{k})\hspace{-1mm}=
\hspace{-1mm}\notag\\
&\begin{bmatrix}-\hat{\mathbf{F}}_k\mathbf{V}\hat{\mathbf{F}}_k^{\mathrm{H}}\hspace{-1mm}+\hspace{-1mm}\mathbf{\Psi}_{\mathrm{C10},k}-\epsilon^{\mathrm{\overline{C13}}}_{k}\mathbf{I}_{\mathrm{M}_t}
&&-\hat{\mathbf{F}}_k\mathbf{V}\notag\\
-\mathbf{V}\hat{\mathbf{F}}_k^{\mathrm{H}}&&\hspace{-8mm}\frac{\epsilon^{\mathrm{\overline{C13}}}_{k}}{\rho_{\mathrm{cu},k}^2+\rho_{\mathrm{d},k}^2}\mathbf{I}_{N+1}-\mathbf{V}\end{bmatrix}\hspace{-1mm}\succeq\hspace{-1mm}\mathbf{0} , \forall k,\notag\\
&\text{and }\overline{\mathrm{C13b}}\hspace{-1mm}: \epsilon^{\mathrm{\overline{C13}}}_{k}\geq0, \forall k.\notag
\end{align}
where $\Delta\mathbf{F}_{\mathrm{eve},j}=\begin{bmatrix}\Delta\mathbf{G}_{\mathrm{ce},j}^{\mathrm{H}}&\Delta\mathbf{h}_{\mathrm{ed},j}\end{bmatrix}$, $\forall j$, for $\Delta\mathbf{F}_{\mathrm{eve},j}\in$
$\Big\{\Delta\mathbf{F}_{\mathrm{eve},j}|\mathrm{Tr}(\frac{1}{\rho_{\mathrm{ce},j}^2+\rho_{\mathrm{ed},j}^2}$$\Delta\mathbf{F}_{\mathrm{eve},j}\Delta\mathbf{F}_{\mathrm{eve},j}^{\mathrm{H}})\leq 1\}$ and $\hat{\mathbf{F}}_{\mathrm{eve},j}=\begin{bmatrix}\hat{\mathbf{G}}_{\mathrm{ce},j}^{\mathrm{H}}&\hat{\mathbf{h}}_{\mathrm{ed},j}\end{bmatrix}$.

\subsection{SCA- and SDR-based Iterative Algorithm }
Next, to tackle the non-convexity of constraints $\mathrm{C4d}$, $\mathrm{C5a}$, $\overline{\mathrm{C9}}$, $\overline{\mathrm{C10}}$, ${\mathrm{C11b}}$, and the objective function in \eqref{optimization_problem_transformed}, we first transform them into a D.C. form such that SCA can be applied to obtain a suboptimal solution.
Firstly, the rank-one constraint ${\mathrm{C11b}}$ is rewritten using \textbf{lemma 3}:
\begin{Lem}
The rank-one constraint $\mathrm{C11b}$ is equivalent to the following D.C. form constraint:
\begin{align}
\widetilde{\mathrm{C11b}}\hspace{-1mm}:\|\mathbf{Y}\|_{*}-\|\mathbf{Y}\|_2\leq0. \label{rank1_v_inequality}
\end{align}
\end{Lem}
\begin{proof}
Since $\mathbf{Y}$ is a Hermitian matrix, the inequality $\|\mathbf{Y}\|_{*}=\sum_i \varrho_i\geq\|\mathbf{Y}\|_2=\underset{i}{\max}\{\varrho_i\}$ holds, where the $i$-th singular value of $\mathbf{Y}$ is denoted by $\varrho_i$. So the inequality in \eqref{rank1_v_inequality} holds if and only if $\mathbf{Y}$ has a unit rank\cite{yu2019robust}.
\end{proof}
Now, we aim to propose an iterative algorithm based on SCA to tackle the D.C. form constraint $\widetilde{\mathrm{C11b}}$. To this end, for any feasible point $\mathbf{Y}^{(t)}$, where $(t)$ denotes the iteration index for the proposed algorithm summarized in \textbf{Algorithm \ref{alg}} (to be discussed in detail later), a lower bound of $\|\mathbf{Y}\|_2$ is given by its first-order Taylor approximation:
\begin{align}
\hspace{-3mm}\|\mathbf{Y}\|_2\!\geq \!\!\|\mathbf{Y}^{(t)}\!\|_2\!+\!\!\mathrm{Tr}\Big(\!\bm{\lambda}_{\max}(\mathbf{Y}^{(t)})\!\!\times\!\! \bm{\lambda}_{\max}^{\mathrm{H}}(\mathbf{Y}^{(t)}\!)(\mathbf{Y}\!\!-\!\mathbf{Y}^{(t)})\!\Big).\hspace{-2mm}
\end{align}
By applying SCA, a subset of constraint $\widetilde{\mathrm{C11b}}$ can be obtained which is given by
\begin{align}
\overline{\mathrm{C11b}}\hspace{-1mm}:&\|\mathbf{Y}\|_*-\|\mathbf{Y}^{(t)}\|_2\\
&-\mathrm{Tr}\Big(\bm{\lambda}_{\max}(\mathbf{Y}^{(t)}) \bm{\lambda}_{\max}^{\mathrm{H}}(\mathbf{Y}^{(t)})\times (\mathbf{Y}-\mathbf{Y}^{(t)})\Big)\leq0.\notag
\end{align}
Since $\overline{\mathrm{C11b}}\Rightarrow{\mathrm{C11b}}$, replacing $\mathrm{C11b}$ with $\overline{\mathrm{C11b}}$ can ensure that the former is satisfied when the proposed algorithm converges. On the other hand, it can be observed that the non-convex objective function in \eqref{objective_ND}, non-convex constraint $\overline{{\mathrm{C3c}}}$ in \eqref{C3aoverlineoverline}, and non-convex constraint $\mathrm{C4d}$ in \eqref{C4dC4e} are also in D.C. form that are differentiable. Similarly, by deriving the first-order Taylor expansions corresponding to the non-convex component, for any feasible point $\iota_k^{(t)}$, $\beta_{\mathrm{PR}}^{(t)}$, and $s_{i,n}^{(t)}$, we have the following inequalities
\begin{align}
&D_{\mathrm{obj}}(\iota_k)\leq\!\!\sum_{k\in\mathcal{K}}\!\!\Big( \nabla^{\mathrm{H}}_{\mathbf{\iota}_k}D_{\mathrm{obj}}(\iota_k^{(t)})(\iota_k-\iota_k^{(t)})\Big)\!+\!D_{\mathrm{obj}}(\iota_k^{(t)}),\label{SCA:D}\\
&e^{-a(\beta_{\mathrm{PR}}-q)}\!\geq\!e^{-a(\beta_{\mathrm{PR}}^{(t)}-q)} -ae^{-a(\beta_{\mathrm{PR}}^{(t)}-q}(\beta_{\mathrm{PR}}\!-\!\beta_{\mathrm{PR}}^{(t)})\text{, and }\notag\\
&s_{i,n}^2\geq \,(s_{i,n}^{(t)})^2+2s_{i,n}^{(t)}(s_{i,n}-s_{i,n}^{(t)}),\forall i,n,\notag
\end{align}
respectively, where $\nabla^{\mathrm{H}}_{\mathbf{\iota}_k}D_{\mathrm{obj}}(\iota_k)=\frac{1}{(\ln2)(\sigma^2_k+\iota_k)}$.  As such, a lower bound of the objective function \eqref{objective_ND} and subsets of constraints $\overline{\mathrm{C3c}}$ and $\mathrm{C4d}$ are given by
\begin{align}
&N_{\mathrm{obj}}(\xi_k,\iota_k)\!-\!\!\sum_{k\in\mathcal{K}}\!\!\Big(\! \nabla^{\mathrm{H}}_{\mathbf{\iota}_k}D_{\mathrm{obj}}(\iota_k^{(t)})(\iota_k\!-\!\iota_k^{(t)})\!\Big)\!\!-\!\!D_{\mathrm{obj}}(\iota_k^{(t)}),\\
&\overline{\mathrm{\overline{C3c}}}\hspace{-1mm}:\! -e^{-a(\beta_{\mathrm{PR}}^{(t)}-q)} \!\!+\!ae^{-a(\beta_{\mathrm{PR}}^{(t)}-q)}(\beta_{\mathrm{PR}}\!-\!\beta_{\mathrm{PR}}^{(t)})\!+\!U_n^{\mathrm{ES}}\!\!\leq\!0,\forall n, \notag\\
&\text{and }\overline{\mathrm{C4d}}\hspace{-1mm}:s_{i,n}-(s_{i,n}^{(t)})^2-2s_{i,n}^{(t)}(s_{i,n}-s_{i,n}^{(t)})\leq0,\forall i,n,\notag
\end{align}
respectively.
\begin{table}[t] \vspace*{-4mm}
\scriptsize
\linespread{1.01}
\begin{algorithm} [H]
\caption{Proposed Resource Allocation Scheme } \label{alg}
\begin{algorithmic} [1]

\STATE Initialize the maximum number of iteration $t_{\max}$, the initial iteration index $t=0$, and variables $\iota^{t}_k$, $\beta^{t}_{\mathrm{PR}}$, $s^{t}_{i,n}$, $\mathbf{W}^{t}_k$, $\mathbf{Z}^{t}$, $\mathbf{\Psi}^{t}_{\mathrm{C9},k}$, $\mathbf{\Psi}^{t}_{\mathrm{C10},k}$, $\mathbf{\Psi}^{t}_{\mathrm{eve},j}$, and $\mathbf{Y}^{t},\forall k,j,i,n$.

\REPEAT[Main Loop: SCA]

\STATE Solve problem \eqref{optimization_problem_sca} with given variables $\iota^{t}_k$, $\beta^{t}_{\mathrm{PR}}$, $s^{t}_{i,n}$, $\mathbf{W}^{t}_k$, $\mathbf{Z}^{t}$, $\mathbf{\Psi}^{t}_{\mathrm{C9},k}$, $\mathbf{\Psi}^{t}_{\mathrm{C10},k}$, $\mathbf{\Psi}^{t}_{\mathrm{eve},j}$, and $\mathbf{Y}^{t},\forall k,j,i,n$, to obtain $\iota^{t+1}_k$, $\beta^{t+1}_{\mathrm{PR}}$, $s^{t+1}_{i,n}$, $\mathbf{W}^{t+1}_k$, $\mathbf{Z}^{t+1}$, $\mathbf{\Psi}^{t+1}_{\mathrm{C9},k}$, $\mathbf{\Psi}^{t+1}_{\mathrm{C10},k}$, $\mathbf{\Psi}^{t+1}_{\mathrm{eve},j}$, and $\mathbf{Y}^{t+1},\forall k,j,i,n$;
\STATE Set $t=t+1$ and update variables;
\UNTIL{convergence or $t=t_{\max}$}.
\end{algorithmic}
\end{algorithm}
\vspace*{-12mm}
\end{table}
Similarly, to overcome the non-convexity of  $\mathrm{C5a}$, $\overline{\mathrm{C9}}$, and $\overline{\mathrm{C10}}$ in \eqref{constriant:c5ac9c10}, 
 we construct their corresponding global underestimators by their first-order Taylor approximations, respectively. In particular, for any feasible point $\mathbf{W}_i^{(t)}$, a lower bound of $\mathrm{Tr}(\mathbf{W}_i^\mathrm{H}\mathbf{W}_i)$ is given by
\begin{align}
&\mathrm{Tr}(\mathbf{W}_i^\mathrm{H}\mathbf{W}_i)\geq-\|\mathbf{W}_i^{(t)}\|^2_{\mathrm{F}}+2\mathrm{Tr}\Big((\mathbf{W}_i^{(t)})^{\mathrm{H}}\mathbf{W}_i\Big).
\end{align}
Following the same method, the lower bound of   $\mathrm{Tr}(\mathbf{W}_k^\mathrm{H}\mathbf{W}_k)$,  $\mathrm{Tr}(\mathbf{\mathbf{\Psi}}_{\mathrm{eve},j}^\mathrm{H}\mathbf{\mathbf{\Psi}}_{\mathrm{eve},j})$,  $\mathrm{Tr}(\mathbf{\mathbf{\Psi}}_{\mathrm{C9},k}^\mathrm{H}\mathbf{\mathbf{\Psi}}_{\mathrm{C9},k})$,  $\mathrm{Tr}(\mathbf{Z}^\mathrm{H}\mathbf{Z})$, and $\mathrm{Tr}(\mathbf{\mathbf{\Psi}}_{\mathrm{C10},k}^\mathrm{H}\mathbf{\mathbf{\Psi}}_{\mathrm{C10},k})$ can be found. Therefore, the subsets of $\mathrm{C5a}$, $\overline{\mathrm{C9}}$, and $\overline{\mathrm{C10}}$  in \eqref{constriant:c5ac9c10} are given by
\begin{align}
&\overline{\mathrm{C5a}}\hspace{-1mm}:\|\mathbf{W}_k\!+\!\mathbf{\Psi}_{\mathrm{eve},j}\|^2_{\mathrm{F}}\!+\!\|\mathbf{W}_k^{(t)}\|^2_{\mathrm{F}}\!-\!2\mathrm{Tr}\Big(\!(\mathbf{W}_k^{(t)})^{\mathrm{H}}\mathbf{W}_k\!\Big)\\
&\hspace{7mm}+(2^{\tau_{k,j}})\|\mathbf{\Psi}_{\mathrm{eve},j}^{(t)}\|^2_{\mathrm{F}}-(2^{\tau_{k,j}+1})\mathrm{Tr}\Big((\mathbf{\Psi}_{\mathrm{eve},j}^{(t)})^{\mathrm{H}}\mathbf{\Psi}_{\mathrm{eve},j}\Big)\notag\\
&\hspace{7mm}-(2^{\tau_{k,j}+1}-2)\sigma^2_\mathrm{eve}+(2^{\tau_{k,j}}-1)\Big\{\|\mathbf{Z}-\mathbf{\Psi}_{\mathrm{eve},j} \|^2_{\mathrm{F}}
+ \notag\\
&\hspace{7mm} \|\mathbf{Z}^{(t)}\|^2_{\mathrm{F}}-2\mathrm{Tr}\Big((\mathbf{Z}^{(t)})^{\mathrm{H}}\mathbf{Z}\Big)\Big\}\leq0, \forall k,j,\notag\\
&\overline{\overline{\mathrm{C9}}}\hspace{-1mm}:\xi_k +\!\!\frac{1}{2}\Big(\|\mathbf{W}_k-\mathbf{\Psi}_{\mathrm{C9},k}\|^2_{\mathrm{F}}+\|\mathbf{\Psi}_{\mathrm{C9},k}^{(t)}\|^2_{\mathrm{F}}\!+\!\|\mathbf{W}_k^{(t)}\|^2_{\mathrm{F}}\Big)\!
\\
&\hspace{5mm}-\!\mathrm{Tr}\Big((\mathbf{\Psi}_{\mathrm{C9},k}^{(t)})^{\mathrm{H}}\mathbf{\Psi}_{\mathrm{C9},k}\Big)\!-\!\mathrm{Tr}\Big(\!(\mathbf{W}_k^{(t)})^{\mathrm{H}}\mathbf{W}_k\!\Big)\!\leq \! 0,\forall k,\text{ and }\notag\\
&\overline{\overline{\mathrm{C10}}}\hspace{-1mm}: \hspace{-3mm}\sum_{i\in\mathcal{K}/ \{k\}}\hspace{-3mm}\Big\{\frac{1}{2}\|\mathbf{W}_i+\mathbf{\Psi}_{\mathrm{C10},k}\|^2_{\mathrm{F}}\Big\}
+\frac{1}{2}\|\mathbf{Z}+\mathbf{\Psi}_{\mathrm{C10},k}\|^2_{\mathrm{F}}\\
&+\frac{K}{2}\|\mathbf{\Psi}_{\mathrm{C10},k}^{(t)}\|^2_{\mathrm{F}}\!-
\!K\mathrm{Tr}\Big(\!(\mathbf{\Psi}_{\mathrm{C10},k}^{(t)})^{\mathrm{H}}\mathbf{\Psi}_{\mathrm{C10},k}\!\Big)\!-\!\mathrm{Tr}\Big(\!(\mathbf{Z}^{(t)})^{\mathrm{H}}\mathbf{Z}\!\Big)\notag\\
&+\hspace{-4mm}\sum_{i\in\mathcal{K}/ \{k\}}\hspace{-3mm}\Big\{\frac{1}{2}\|\mathbf{W}_i^{(t)}\|^2_{\mathrm{F}}\!-\!\mathrm{Tr}\Big(\!(\mathbf{W}_i^{(t)})^{\mathrm{H}}\mathbf{W}_i\!\Big)\Big\}\!+\!\frac{1}{2}\|\mathbf{Z}^{(t)}\|^2_{\mathrm{F}} \!\leq\!\iota_k ,\forall k,\notag
\end{align}
respectively, where $\mathbf{W}_k^{(t)}$, $\mathbf{Z}^{(t)}$, $\mathbf{\Psi}_{\mathrm{eve},j}^{(t)}$, $\mathbf{\Psi}_{\mathrm{C9},k}^{(t)}$, and $\mathbf{\Psi}_{\mathrm{C10},k}^{(t)}$ are the solutions obtained in the $t$-th iteration.
For notational simplicity, we define  $\epsilon=\{\epsilon^{\mathrm{C3f}}_{k,n},\epsilon^{\mathrm{\overline{C5b}}}_{j},\epsilon^{\mathrm{\overline{C12}}}_{k},\epsilon^{\mathrm{\overline{C13}}}_{k}\}$, $\mathbf{\Psi}=\{\mathbf{\Psi}_{\mathrm{C9},k},\mathbf{\Psi}_{\mathrm{C10},k},\mathbf{\Psi}_{\mathrm{eve},j}\}$, and $\mathcal{U}=\{U_{n}^{\mathrm{ES}},U_{n}^{\mathrm{BS}}\}$. A lower bound of \eqref{optimization_problem_transformed} can be obtained via solving the following optimization problem:
\begin{align}
\underset{\substack{\mathbf{W}_k\in\mathbb{H}^{M_{\mathrm{t}}}, \\ \mathbf{Z}\in\mathbb{H}^{M_{\mathrm{t}}},\,\\  \mathbf{S}, \mathbf{Y}\in\mathbb{H}^{N+2},\\ \mathbf{V}\in\mathbb{H}^{N+1},\\ \mathbf{v}, \xi_k,\iota_k,\\ \beta_{\mathrm{PR}},\,\mathcal{U},\epsilon,\mathbf{\Psi}}}{\mathrm{maximize}} &\,\, N_{\mathrm{obj}}(\xi_k,\iota_k)\!-\!\!\!\sum_{k\in\mathcal{K}}\! \nabla^{\mathrm{H}}_{\mathbf{\iota}_k}\hspace{-1mm}D_{\mathrm{obj}}(\iota_k^{(t)})(\iota_k\!-\!\iota_k^{(t)})\!-\!D_{\mathrm{obj}}(\iota_k^{(t)})
  \label{optimization_problem_sca} \\[-14mm]
&\mathrm{s.t.}\,\,\mathrm{C1},\mathrm{\overline{\overline{C3a}}},\mathrm{\overline{\overline{C3b}}},\overline{\mathrm{\overline{C3c}}},\mathrm{\overline{C3d}},
\mathrm{\overline{C3e}},\mathrm{\overline{\overline{C3fa}}},\mathrm{\overline{\overline{C3fb}}},\notag\\ &\hspace{6mm}\mathrm{{\overline{C3g}}},\mathrm{\overline{C3h}},
\mathrm{C4a},\mathrm{C4b},\overline{\mathrm{C4d}},\mathrm{C4e},\overline{\mathrm{C5a}},\overline{\mathrm{C5b}},\notag\\ &\hspace{6mm}\overline{\mathrm{C5c}},\mathrm{C6},
\mathrm{C7},\mathrm{C8},\overline{\overline{\mathrm{C9}}},\overline{\overline{\mathrm{C10}}},\mathrm{C11a},
\overline{\mathrm{C11b}},\notag\\ &\hspace{6mm}\overline{\mathrm{C12a}},\overline{\mathrm{C12b}},\overline{\mathrm{C13a}},
\overline{\mathrm{C13b}}.\notag
\end{align}
Since rank-one constraint $\mathrm{C8}$ is the only non-convex part in optimization problem \eqref{optimization_problem_sca}, we apply the SDR technique to drop the rank constraint in $\mathrm{C8}$, i.e., \cancel{$\mathrm{C8}\hspace{-1mm}: \mathrm{Rank}(\mathbf{W}_k)\leq1, \forall k$}, such that the rank constraint relaxed version of \eqref{optimization_problem_sca} can be solved by using standard numerical solver for convex programming. In the following theorem, the tightness of the adopted SDR is revealed.
\begin{Thm}\label{Thm:rank_one_W}
For $P_{\max}>0$ and if \eqref{optimization_problem_sca} is feasible, the rank one constraint of $\mathrm{Rank}(\mathbf{W}_k)\leq1$ in \eqref{optimization_problem_sca} can always be obtained.
\end{Thm}
\begin{proof}
Please refer to Appendix \ref{Proof:rankoneW} for the proof of Theorem \ref{Thm:rank_one_W}.
\end{proof}

Due to the use of SCA, solving the problem in \eqref{optimization_problem_sca} provides a lower bound for the problem in \eqref{proposed_formulation_origion}. To tighten the obtained performance lower bound, we iteratively update the feasible solution by solving the optimization problem in \eqref{optimization_problem_sca} in the $t$-th iteration. The proposed SCA-based algorithm is shown in \textbf{Algorithm \ref{alg}} and the proof of its convergence to a suboptimal solution can be found in \cite{opial1967weak} which is omitted here for brevity.
\begin{table*}[t]
\scriptsize
\centering
\caption{Computational complexity comparison between different algorithms}\vspace{-0mm}
\label{tab:CC_R}
\begin{tabular}{|m{3cm}|m{10cm}|}
\hline
\textbf{Algorithm }             & \textbf{Computational Complexity}\cite{polik2010interior}   \\ \hline
Exhaustive search&
\begin{align}
\mathcal{O}\Big(2^NB^N\frac{P_{\max}}{\Delta}^{M_t^2K+M_t^2}\Big)\label{eq:BigO_ES}\\[-7mm]\notag
\end{align}\\
\hline
AO algorithm &
\begin{align}
\mathcal{O}\Bigg(\!\!\Big(\!N_{\mathcal{O},\mathrm{AO1}}M_{\mathcal{O},\mathrm{AO1}}^3
\!+\!M_{\mathcal{O},\mathrm{AO1}}^2N_{\mathcal{O},\mathrm{AO1}}^2\!+\!N_{\mathcal{O},\mathrm{AO1}}^3\!\Big)\label{eq:BigO_AO}\\[-2mm]
\!\times t_{\max}^{\mathrm{AO1}}\sqrt{M_{\mathcal{O},\mathrm{AO1}}}\log\frac{1}{\varrho_1}\!+\!N_{\mathcal{O},\mathrm{AO2}}^3 M_{\mathcal{O},\mathrm{AO2}}^2\times t_{\max}^{\mathrm{AO2}}\log\frac{1}{\varrho_2}\Bigg)\notag\\[-7mm]\notag
\end{align}\\
\hline
The proposed algorithm &
\begin{align}
\mathcal{O}\Bigg(\!\!\Big(\!N_{\mathcal{O},\mathrm{P}}M_{\mathcal{O},\mathrm{P}}^3\!+\!
M_{\mathcal{O},\mathrm{P}}^2N_{\mathcal{O},\mathrm{P}}^2\!+\!N_{\mathcal{O},\mathrm{P}}^3\!\Big)
\!\!\times\!\!\sqrt{M_{\mathcal{O},\mathrm{P}}}\log\frac{1}{\varrho}\!\!\Bigg)\label{eq:BigO_PS}\\[-7mm]\notag
\end{align}\\
\hline
\end{tabular}
\end{table*}

In practice, several methods \cite{boyd2004convex} can be used to solve the convex problem in \eqref{proposed_formulation_origion}. Table \ref{tab:CC_R} shows the computational complexity and performance comparison of the exhaustive search method, the AO approach, and the proposed algorithm of the optimization problem in \eqref{proposed_formulation_origion}. Since the precoder adopted at the AP, $\mathbf{w}_k, \forall k$, and artificial noise vector at the AP, $\mathbf{z}$, are continuous variables, to solve the problem in \eqref{proposed_formulation_origion}, the exhaustive search is performed by a grid search with a grid resolution of $\Delta>0$. In particular, the exhaustive search enumerates all the variable combinations within the feasible set of problem in \eqref{proposed_formulation_origion} and then returns the variable combination with the maximum system sum-rate. As such, the computational complexity of the exhaustive search method for the problem in \eqref{proposed_formulation_origion} is given by \eqref{eq:BigO_ES} in Table \ref{tab:CC_R}, where $\mathcal{O}$ is the big-O notation.
 Although the exhaustive search can obtain an optimal solution of the problem in \eqref{proposed_formulation_origion}, it is computationally intractable even for a moderate size system. In the following, we analyze the computational complexity of the AO algorithm and the proposed algorithm in details.\\
 On the other hand, the computational complexity of each iteration of the conventional AO algorithm to solve problem \eqref{proposed_formulation_origion} is given by \eqref{eq:BigO_AO} in Table \ref{tab:CC_R} \cite{polik2010interior}, where
 \begin{align}
M_{\mathcal{O},\mathrm{AO1}}=&KJ+7K+2J+2\text{ and } \\ N_{\mathcal{O},\mathrm{AO1}}=&(3K+J+1)M_t^2+3K+2J\notag
\end{align}
are the number of inequalities and the number of variables (dominated terms) of subproblem 1 of the AO algorithm, respectively. Besides,
\begin{align}
M_{\mathcal{O},\mathrm{AO2}}\hspace{-1mm}=&3\hspace{-1mm}+\hspace{-1mm}6N\hspace{-1mm}+\hspace{-1mm}4BN\hspace{-1mm}+\hspace{-1mm}2K\hspace{-1mm}+\hspace{-1mm}KJ\text{ and }\\
N_{\mathcal{O},\mathrm{AO2}}\hspace{-1mm}=&N^2\hspace{-1mm}+\hspace{-1mm}4N\hspace{-1mm}+\hspace{-1mm}2\hspace{-1mm}+\hspace{-1mm}BN\hspace{-1mm}+\hspace{-1mm}4K\hspace{-1mm}+\hspace{-1mm}J\hspace{-1mm}
+\hspace{-1mm}(2K\hspace{-1mm}+\hspace{-1mm}J)\hspace{-0mm}M_t^2\notag
\end{align} are the number of inequalities and the number of variables of subproblem 2 of \eqref{eq:BigO_ES} the AO algorithm, respectively. Note that $\varrho_1$ and $\varrho_2$ is the threshold of convergence tolerance of the subproblems $1$ and $2$ of the AO algorithm, respectively. Besides, $t_{\max}^{\mathrm{AO1}}$ and $t_{\max}^{\mathrm{AO2}}$ are the maximum iteration times of the subproblems $1$ and $2$ of the AO algorithm, respectively.
Moreover, the computational complexity of each iteration of the proposed algorithm is
given by \eqref{eq:BigO_PS} in Table \ref{tab:CC_R} \cite{polik2010interior},
where $M_{\mathcal{O},\mathrm{P}}=5+11N+2BN+KJ+2J+7K$ and $N_{\mathcal{O},\mathrm{P}}= (3K+J+1)M_t^2+7N+BN+N^2+4K+2$ are the number of inequalities and the number of variables of the proposed scheme, respectively. Besides, $\varrho$ is the threshold of convergence tolerance of the proposed algorithm. {Note that the computational complexity of the proposed algorithm and the AO method both enjoy a polynomial time computational complexity which is known as fast algorithms for implementation \cite{cormen2009introduction} and the conclusion of their complexity comparison depends on the adopted parameters. It is worth to mention that the performance of the proposed algorithm outperforms that of the AO algorithm, which will be shown in the following simulations. Indeed, rather than separately and alternatively optimizing variables in each subproblem, the proposed algorithm can jointly optimize all optimization variables in every iteration for a better exploitation of the problem structure.}

\section{Simulation Results}
\subsection{Simulation Setup}
\begin{figure}[t]
  \centering
  \includegraphics[width=3.4in]{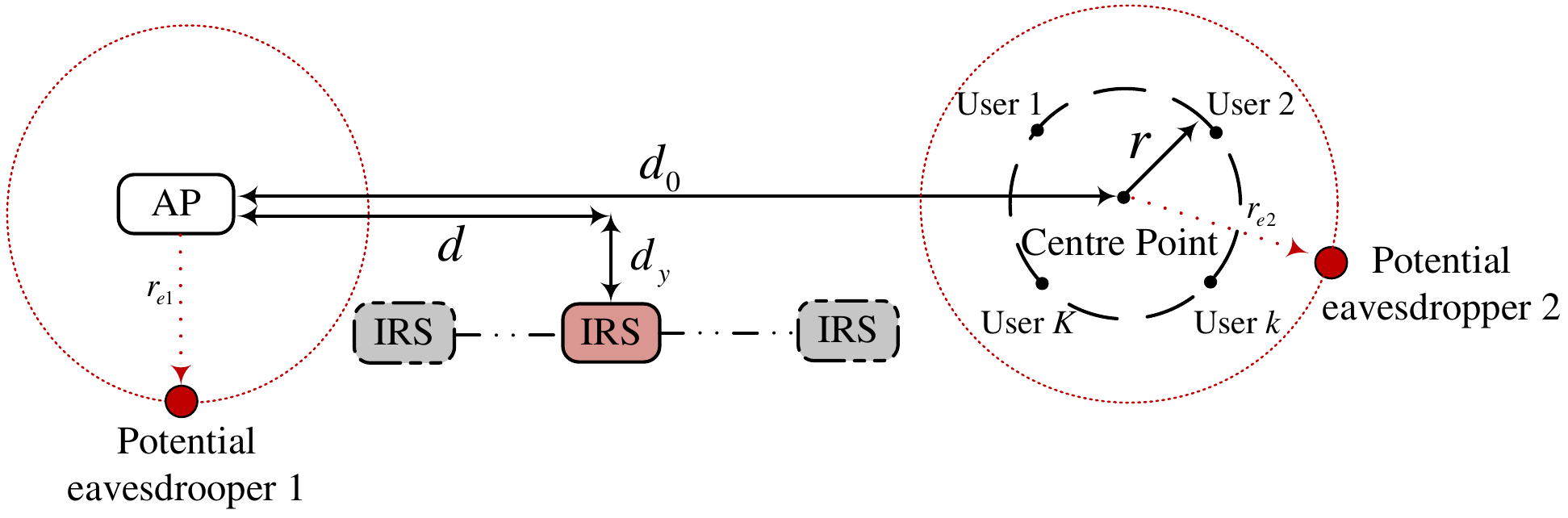}
  \caption{Simulation setup.}
  \label{distance_model} 
\end{figure}
This section evaluates the system performance of the proposed self-sustainable IRS scheme via simulation. The system setup is shown in Fig. \ref{distance_model}. The $K$ users and $J-1$ the potential eavesdroppers are randomly distributed on the circumference of two concentric circles with radii of $r = 1$ m and $r_{\mathrm{e2}} = 20$ m, respectively. The other potential eavesdropper is located on the circumference of a circle centered at the AP with a radii of $r_{\mathrm{e1}} = 80$ m. The AP and the center point of the circle formed by users are located $d_0 = 60$ m apart. Also, the IRS lies on a horizontal line that is in parallel to the one that connects the AP and the centre point of the users formed concentric circle. The vertical distance between these two lines is  $d_y = 1$ m and the horizontal distance between the AP and the IRS is denoted by $d$. The distance-dependent path loss model\cite{goldsmith2005wireless} is adopted with $10$ m reference distance.
For ease of presentation, we define $\kappa_{\mathrm{cu},k}, \kappa_{\mathrm{G}}, \kappa_{\mathrm{d},k}, \kappa_{\mathrm{ce},j}$, and $\kappa_{\mathrm{ed},j}$ as the maximum normalized estimation error for the channel $\mathbf{G}_{\mathrm{cu},k}, \mathbf{G}, \mathbf{h}_{\mathrm{d},k}, \mathbf{G}_{\mathrm{ce},j}$, and $\mathbf{h}_{\mathrm{ed},j}$, i.e., $\kappa_{\mathrm{cu},k}=\frac{\rho_{\mathrm{cu},k}}{\|\mathbf{G}_{\mathrm{cu},k}\|_{\mathrm{F}}}, \kappa_{\mathrm{G}}=\frac{\rho_{\mathrm{G}}}{\|\mathbf{G}\|_{\mathrm{F}}}, \kappa_{\mathrm{d},k}=\frac{\rho_{\mathrm{d},k}}{\|\mathbf{h}_{\mathrm{d},k}\|_{2}}, \kappa_{\mathrm{ce},j}=\frac{\rho_{\mathrm{ce},j}}{\|\mathbf{G}_{\mathrm{ce},j}\|_{\mathrm{F}}}$, and $\kappa_{\mathrm{ed},j}=\frac{\rho_{\mathrm{ed},j}}{\|\mathbf{h}_{\mathrm{ed},j}\|_{2}}$, respectively. Unless further specified, we fix $\kappa=\kappa_{\mathrm{cu},k}=\kappa_{\mathrm{G}}=\kappa_{\mathrm{d},k}=\kappa_{\mathrm{ce},j}=\kappa_{\mathrm{ed},j}$ and $\kappa^2 = 0.1$, which is the normalized mean square error based on channel estimation method in \cite{hu2019two}. The maximum of tolerable channel capacity of potential eavesdroppers is $\tau =\tau_{k,j}=1.5$ bits/s/Hz. Other important parameters are summarized in Table \ref{Table:system_para}.
\begin{table*}[]
\caption{System parameters}
\label{Table:system_para}
\scriptsize
\center
\begin{tabular}{|m{8cm}|m{5cm}|}
\hline
Antenna gains at the AP, the IRS and receivers & $20$ dBi, $0$ dBi, and $0$ dBi \\ \hline
Systen carrier centre frequency & $2.4$ GHz \\ \hline
Number of antennas at AP $M$ &  $10$\\ \hline
Number of the IRS elements $N$ &  $50$\\ \hline
Number of users $K$ and the number of eavesdroppers $J$ & $K=2$ and $J=2$ \\ \hline
\hspace{-2mm}\begin{tabular}{l}Path loss exponents of AP-user links, AP-Eve links,\\ AP-IRS link, IRS-user links, and IRS-Eve links\end{tabular} & \hspace{-2mm}\begin{tabular}{l}$\alpha_{\mathrm{AU}}=\alpha_{\mathrm{AE}}=3.6$\cite{mardeni2010optimised},\\$
\alpha_{\mathrm{AI}}=\alpha_{\mathrm{IU}}=\alpha_{\mathrm{IE}}=2.2$ \end{tabular}\\ \hline
\hspace{-2mm}\begin{tabular}{l}Rician factors of of AP-user links, AP-Eve links, \\ AP-IRS link, IRS-user links, and IRS-Eve links \end{tabular} &  \hspace{-2mm}\begin{tabular}{l}$
\beta_{\mathrm{AU}}=\beta_{\mathrm{AE}}=0$,\\$\beta_{\mathrm{AI}}=\beta_{\mathrm{IU}}=\beta_{\mathrm{IE}}=3$ \end{tabular}\\ \hline
Non-linear energy harvesting circuit parameters &\hspace{-2mm}\begin{tabular}{l} $M_{\mathrm{p}}=80$ mW, $a = 1\hspace{-0.5mm},\hspace{-0.5mm}500$, \\and $q = 0.0022$ \cite{guo2012improved} \end{tabular}\\ \hline
Maximum power budget at the AP $P_{\max}$ &  $P_{\max} = 38$ dBm\\ \hline
Noise power at users and eavesdroppers & $\sigma_{k}^2=\sigma_{\mathrm{eve},j}^2=-90$ dBm \\ \hline
Power consumption consumed by power conversion circuit $P_\mathrm{c}$ & $P_\mathrm{c} = 2.1$ $\mu W$\cite{yao2008fully} \\ \hline
Phase shifter bit resolution of IRS elements $b$& $b=3$ bits \\ \hline
Power consumption of each IRS reflection element $P_{\mathrm{IRS}}(b)$ & $P_{\mathrm{IRS}}(b)=1.5$ mW for $b=3$ \cite{huang2019reconfigurable} \\ \hline
The maximum of tolerable channel capacity of eavesdroppers & $\tau =\tau_{k,j}=1.5$ bits/s/Hz \\ \hline
\end{tabular}
\end{table*}

For comparison, we evaluate the system performance of five baseline schemes: 1) ``Upper bound $1$" has the same setting as the considered model expect that all IRS elements in this scheme adopt continuous phase shifters but without consuming any energy; 2) ``Upper bound $2$" is achieved by the IRS-assisted system with an energy harvesting IRS but without potential eavesdroppers; 3) ``Baseline scheme $1$" is a non-robust design for the case when the IRS is not deployed and the maximum ratio transmission (MRT) with respect to users is adopted as the precoder at the AP. In particular, this scheme treats the imperfect CSI as the perfect one and the direction of the precoder for the $k$-th user is fixed to $\frac{\mathbf{h}_{\mathrm{d},k}^{\mathrm{H}}}{\|\mathbf{h}_{\mathrm{d},k}\|}$. Then, we optimize the power allocation for each user subject to constraint C1 in problem \eqref{proposed_formulation_origion} for the maximization of system sum-rate; 4) ``Baseline scheme $2$" is the system with a self-sustainable IRS adopting the same MRT precoder at the AP as in baseline scheme $1$ for non-robust transmission. Yet, its phase shifts of the IRS and the power allocation at the AP are jointly optimized by applying the method developed by the proposed scheme; 5) ``Baseline scheme $3$" is the same as the proposed scheme expect that it treats the estimated channel information as perfect channel information; 6) ``Baseline scheme $4$" focuses on the same setting as the proposed scheme except that all IRS elements in the former adopt random phase shifts but without consuming any energy; 7) ``Baseline scheme 5'' is the same as the proposed scheme expect that it does not take the existence of eavesdroppers into account. In other words, there is no communication security guaranteed; 8) ``AO scheme'' is the conventional algorithm which has been widely used in the literature \cite{lyu2020optimized,zou2020wireless,abeywickrama2020intelligent,wu2019intelligent}. It has the same setting as the considered model except that AO is applied which splits the original problem into two sub-problems and optimizes the grouped variables $\{\mathbf{w}_k$, $\mathbf{Z}\}$ and $\{\theta_n,$ $\alpha_n\}$ in the subproblems, respectively, via an alternating manner. The resource allocation for the eight baselines scheme can be obtained by modifying the proposed scheme accordingly. Note that for all the schemes, if the joint design of precoder and phase shifts of IRS are unable to meet the requirements of constraints in \eqref{proposed_formulation_origion}, we set the system sum-rate for that channel realization to zero to account the penalty for the corresponding failure.
\subsection{Average System Sum-Rate versus the Distance Between AP and IRS}
\begin{figure}[t] 
\centering
  \includegraphics[width=\linewidth]{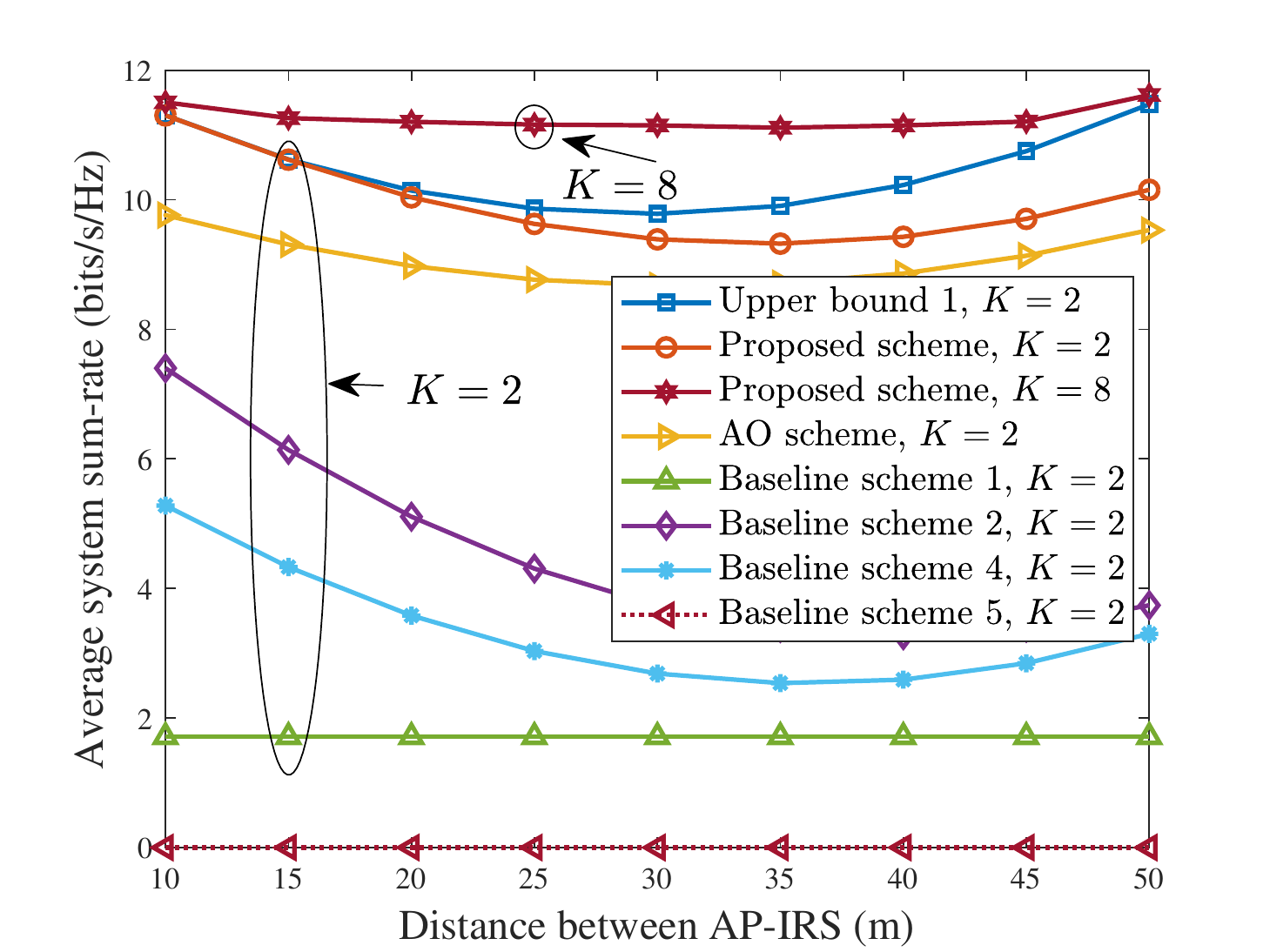}
  \caption{Average system sum-rate (bits/s/Hz) versus the distance between the AP and the IRS (m) for $J=2$.}
  \label{results:distanceSR} 
\end{figure}

\begin{figure}[t] 
\centering
  \includegraphics[width=\linewidth]{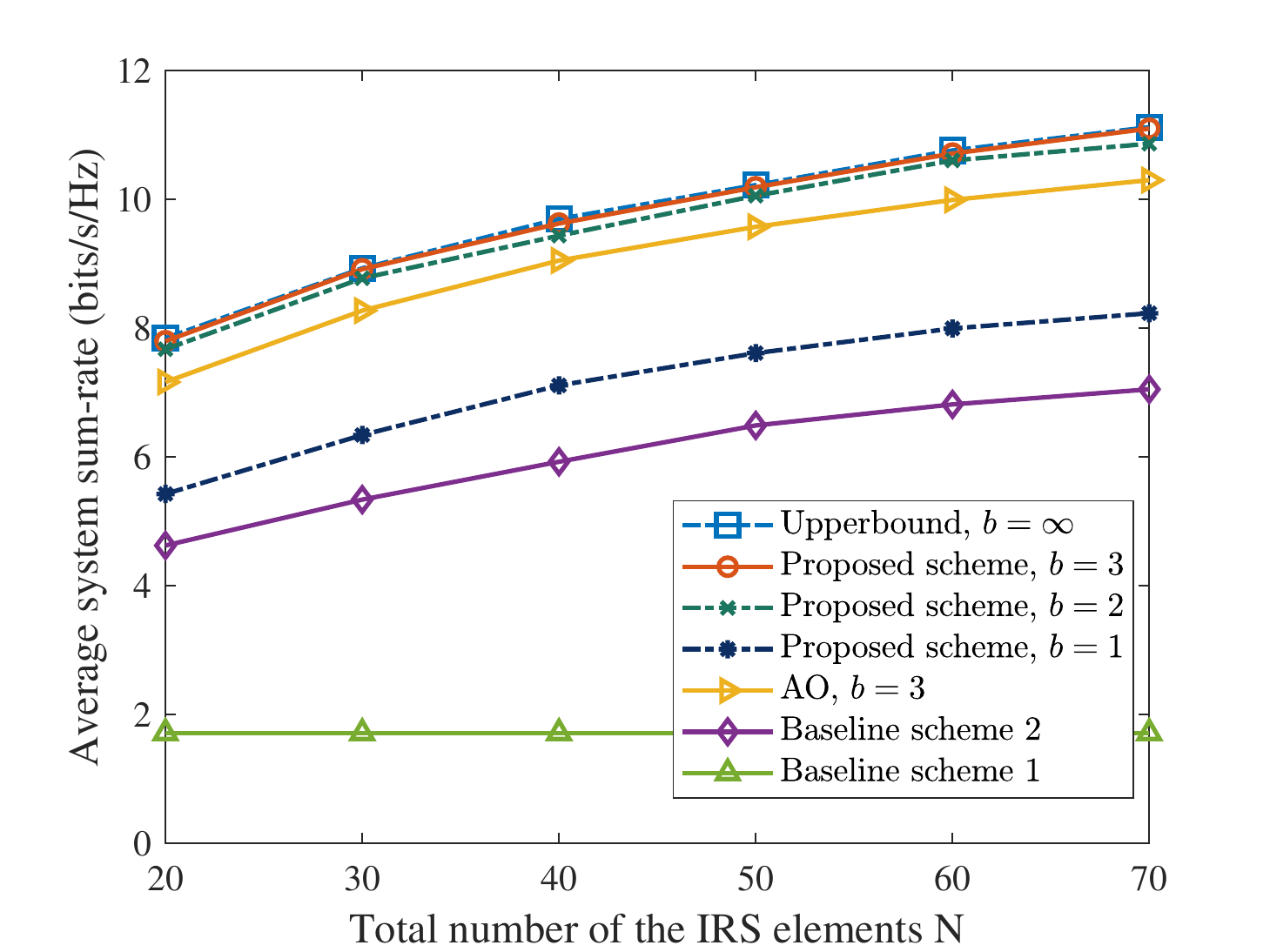}
  \caption{Average system sum-rate (bits/s/Hz) versus the number of IRS elements $N$ for $P_{\max}=36$ dBm.}
  \label{results:numN} 
\end{figure}

Fig. \ref{results:distanceSR} depicts the average system sum-rate versus the horizontal distance between the AP and the IRS for different schemes.
It can be observed that the proposed scheme can achieve a substantially higher sum-rate than that of baseline scheme $1$, which the IRS is not available. Indeed, the IRS provides an additional path gain, $\mathbf{h}_{\mathrm{r},k}^{\mathrm{H}}\mathbf{\Theta}\mathbf{G},\forall k$, which can be exploited and optimized by the proposed scheme to improve the system performance. Besides, due to the joint optimization of the transmit beamforming and phase shifts, the proposed scheme can achieve a considerable performance gain compared with baseline scheme $2$ adopting a fixed beamforming direction.
In addition, we note that baseline scheme $4$ has unsatisfactory performance compared with the proposed scheme. Indeed, since baseline scheme $4$ adopts random phase shifts, the information beams reflected by the IRS do not always align with the direction of the desired channels such that the additional path gain cannot be fully exploited. Furthermore, since baseline scheme 5 does not take the existence of eavesdroppers into account while optimizing the precoder and IRS phase shifts, there is no power allocated to the artificial noise vector, $\mathbf{z}$, which generally results in a high channel capacity at the potential eavesdroppers. Therefore, the resulting joint design of baseline scheme 5 is unable to fulfill the security QoS requirement, i.e., constraint C5 of the problem in \eqref{proposed_formulation_origion}, which leads to an unsatisfactory performance. Moreover, the proposed scheme outperforms the AO scheme. Indeed, AO splits the original problem in \eqref{proposed_formulation_origion} into two subproblems and optimizes the grouped  variables $\{\mathbf{w}_k$, $\mathbf{Z}\}$ and $\{\theta_n,$ $\alpha_n\}$ in the two subproblems, respectively, via an alternating manner. In particular, it is well-known that the AO method can be easily trapped in inefficient stationary points over iterations in practice, which leads to possible severe performance degradation. In contrast, our proposed algorithm is able to jointly optimize all these optimization variables in each iteration, which can better exploit the structure of the design problem resulting in a higher sum-rate than the conventional AO approach. Furthermore, since introducing more users to the system offers more multi-user diversity \cite{tse2005fundamentals,wei2017optimal} for the proposed resource allocation design to exploit, the system sum-rate of the case $K=8$ and $J=2$ is higher than that of $K=2$ and $J=2$.\\
On the other hand, for all schemes with IRS deployed except baseline scheme 5, the average system sum-rate is at its lowest when the IRS is close to the middle between the AP and the center point of users. In fact, when the IRS is neither close to the AP nor the users, both the AP-IRS path and the IRS-users paths would experience significant attenuations that decrease the capability of the IRS in focusing the reflected signals to the desired users. Besides, we can also observe from Fig. \ref{results:distanceSR} that when the IRS is in close proximity to the AP, the performance of the proposed scheme approaches that of the upper bound. In contrast, as the distance $d$ further increases, the sum-rate gap between the proposed scheme and the upper bound is slightly enlarged. This is because as $d$ increases, each IRS element would harvest less power on average. As can be expected from constraint ${\mathrm{C}3}$ in \eqref{proposed_formulation_origion}, more IRS elements are switched to the power harvesting mode to maintain the sustainability of the IRS resulting in a less number of IRS elements for improving the system sum-rate via signal reflection.
\subsection{Average System Sum-Rate versus the Number of IRS Elements}
\begin{figure}[t]
    \centering
     \begin{subfigure}[b]{0.485\textwidth}
      \centering
      \includegraphics[width=\linewidth]{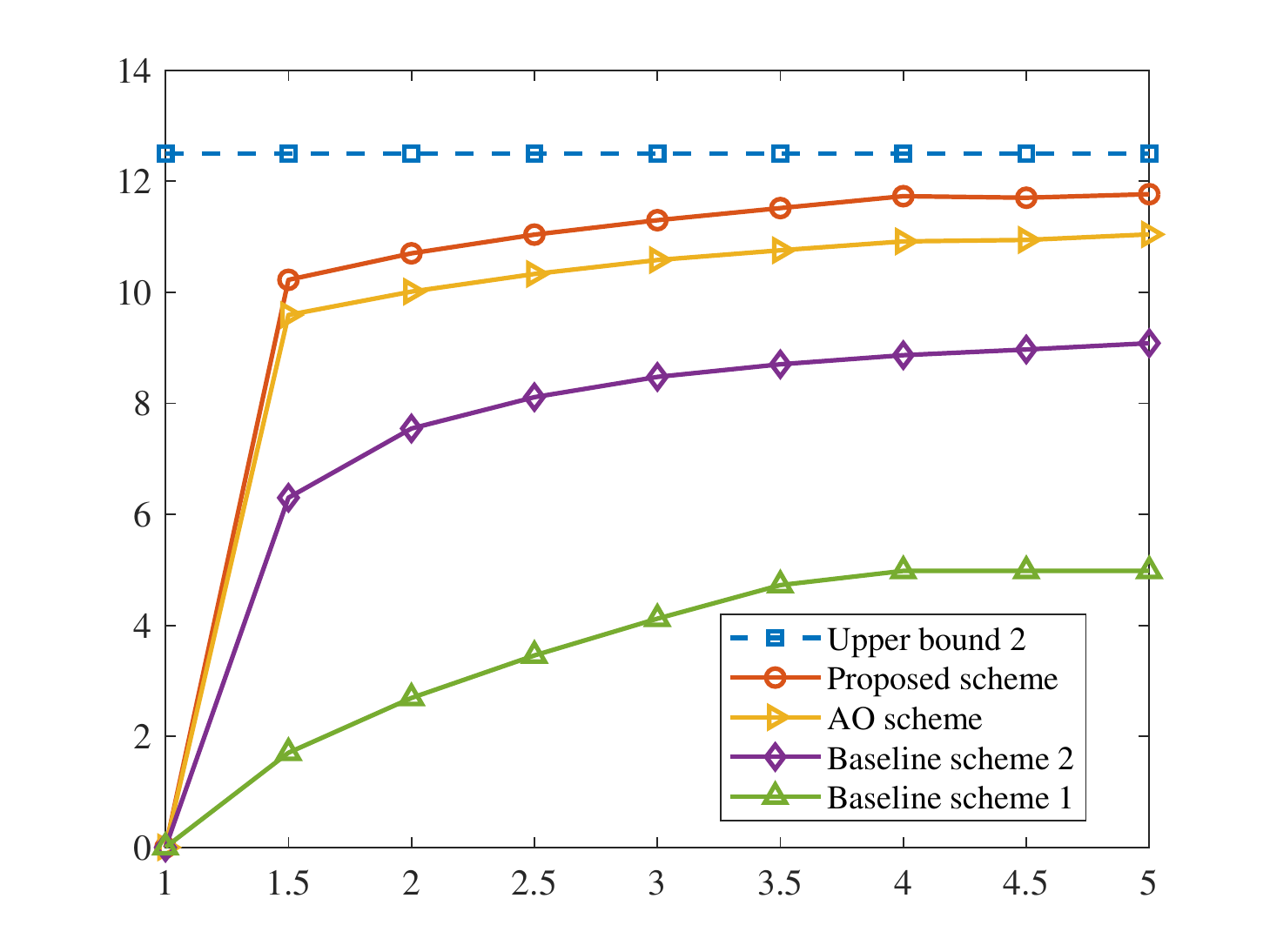}
      \caption{\scriptsize Average system sum-rate (bits/s/Hz) versus the maximum tolerable channel capacity of the potential eavesdroppers (bits/s/Hz).}
      \label{Result:RCVsTau}
    \end{subfigure}\hfill
    \begin{subfigure}[b]{0.485\textwidth}
      \centering
      \includegraphics[width=\linewidth]{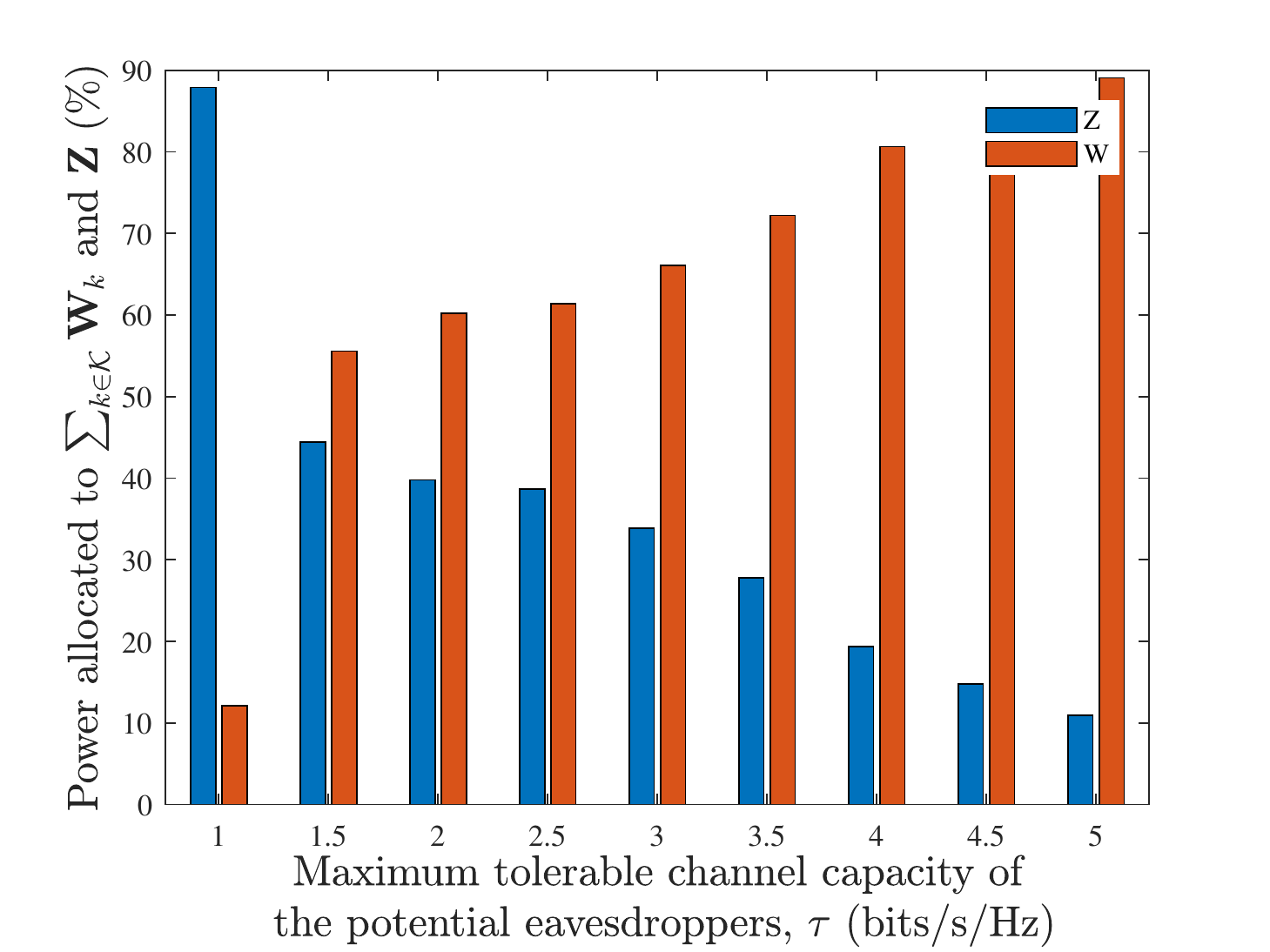}
      \caption{\scriptsize Power allocated to information beamforming and AN of the proposed scheme versus the maximum tolerable channel capacity of the potential eavesdroppers (bits/s/Hz).}
      \label{Result:PrecentageVsTau}
    \end{subfigure}
    \caption{The system parameters are set as $N=20, J=2, K=2, d=15$ m, $\kappa^2=0.1$, and $P_{\max} = 36$ dBm.}
\end{figure}
Fig. \ref{results:numN} shows the variation of average system sum-rate with different numbers of IRS elements $N$ at $d = 15$ m. It can be observed that with an increasing number of IRS elements $N$, the average system sum-rate of the proposed scheme increases. 
Indeed, the extra spatial degrees of freedom offered by the increased number of reflecting IRS elements provides a high flexibility in beamforming to enhance the channel quality of the end-to-end AP-IRS-user link for improving the system sum-rate. More interestingly, the system sum-rate is gradually saturated as the number of IRS elements $N$ increases. In practice, the number of activated IRS elements for reflection is constrained by both the total amount of energy harvested from the AP and the hardware limitation of the energy harvesting circuit. Although the increase of IRS elements can collect more energy at the input of the energy harvesting circuit, an exceeding large power can saturate the circuit due to its non-linearity. As such, the harvestable energy reaches its maximum which is unable to further support more activated IRS reflection elements. Also, Fig. \ref{results:numN} compares the performance of the proposed scheme with different bit resolutions, $b$, of each IRS phase shifter. Note that the upper bound in Fig. \ref{results:numN} is the previously mentioned upper bound scheme $1$ but with $b=\infty$. It can be observed that the performance of the IRS phase shifts with a bit resolution of $b = 2,3$ bits can approach that of the upper bound. In particular, increasing the bit resolution above $2$ bits would only provide a marginal improvement in system sum-rate. In fact, the IRS-user links are dominated by LoS components in Rician fading channels. Thus, a small bit resolution of phase shifts is sufficient to facilitate the beamformer aligning the desired signals with the dominant channels.
\subsection{Sum-Rate Versus the Maximum Tolerable Channel Capacity of the Eavesdroppers}
Fig. \ref{Result:RCVsTau} depicts the average system sum-rate versus the maximum tolerable channel capacity of the potential eavesdroppers, $\tau_{k,j}$. As can be observed, thanks to the optimization of the transmit beanformer at the AP and the phase shifts at the IRS, the average system sum-rate of the proposed scheme outperforms that of the both baseline schemes $1$ and $2$. Moreover, the performance gap between the proposed scheme and upper bound $2$ is generally small for $\tau>1.5$, which confirms the effectiveness of the proposed scheme to strike a balance between system sum-rate and security provisioning. Furthermore, the system sum-rate grows as $\tau$ increases. Indeed, as shown in Fig. \ref{Result:PrecentageVsTau}, when $\tau$ is small, a large portion of the transmit power is allocated to transmitting AN to deteriorate the capacity of the potential eavesdroppers. As a result, less power is allocated to maximize the system sum-rate. As for large $\tau$, the information leakage constraints are less stringent, thus more power can be assigned to the precoder $\mathbf{W}_k$, which can effectively exploit the DoF brought by the AP antennas and IRS elements to create powerful and sharp beamforming for improving the system sum-rate.
\subsection{CSI Uncertainties}
\begin{figure}[t] 
  \centering
  \includegraphics[width=1\linewidth]{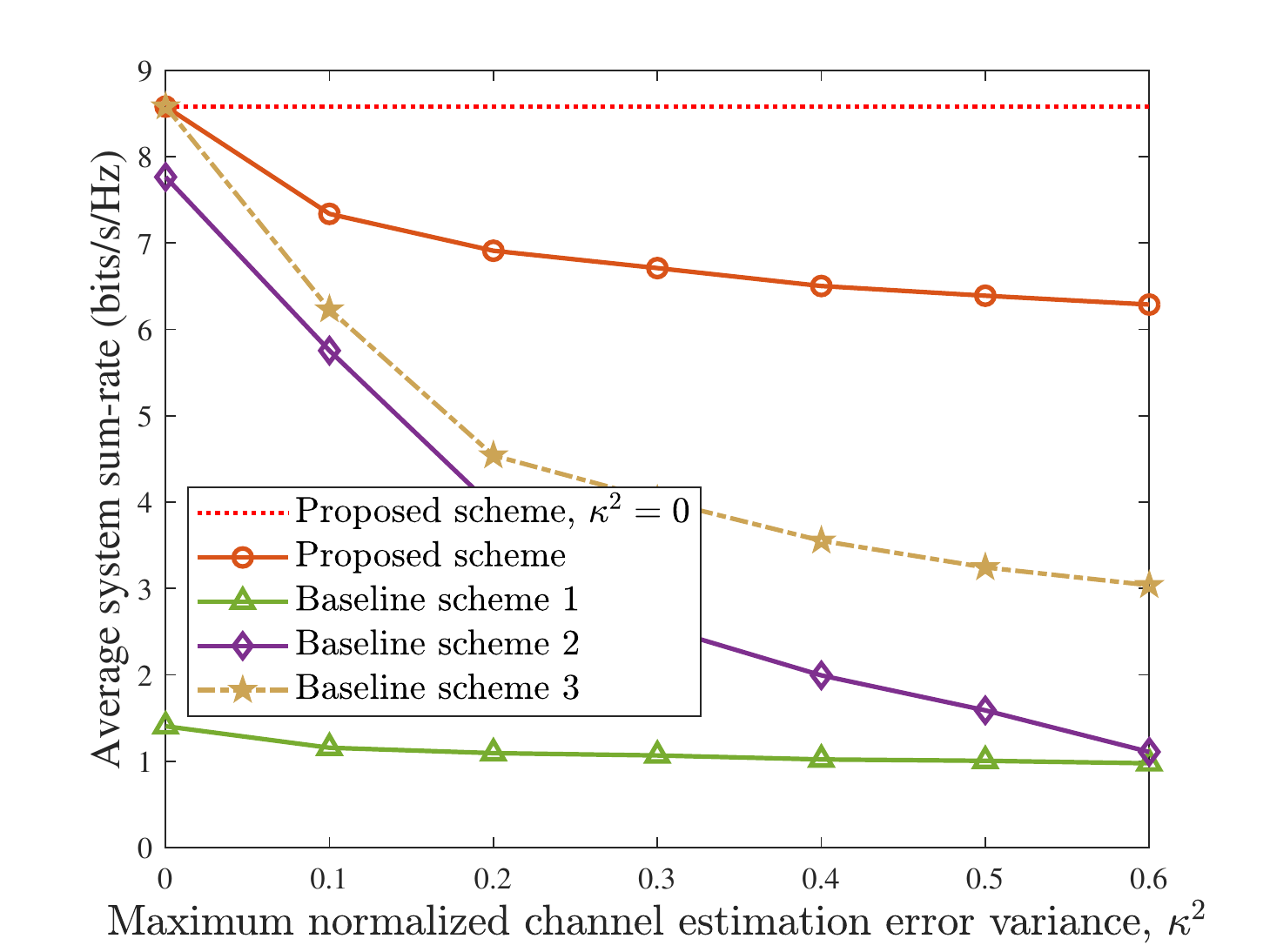}
  \caption{Average system sum-rate (bits/s/Hz) versus the maximum normalized channel estimation error variance, $\kappa^2$.}
  \label{Result:RsVserror} 
\end{figure}

Fig. \ref{Result:RsVserror} shows the average system sum-rate against different maximum normalized channel estimation error variance $\kappa^2$ when $d=10$ m and $P_{\max} = 30$ dBm. As can be observed, for all the schemes, the average sum-rates decline as the quality of CSI is degraded. This is because with the deteriorated CSI quality, both the AP and the IRS become less capable in creating energy-focused information carrying beams
and more power is allocated to AN jamming to satisfy the secrecy constraint. On the other hand, since the proposed scheme can jointly optimize the IRS phase shifters and precoders, which exploits the spatial DoF more efficiently than that of baseline schemes $2$ and $3$, the performance of the proposed scheme significantly outperforms all baseline schemes over the entire range of the different channel error estimation variances, which confirms the robustness against the CSI imperfectness. Moreover, although the performance of baseline scheme 3 is as good as the proposed scheme at $\kappa^2=0$,  its sum-rate decreases rapidly as the value of $\kappa^2$ increases. Since baseline 3 treats imperfect CSI as the prefect one for beamforming design, the resource allocation mismatches are magnified and it is unable to fulfill the QoS requirement leading to unsatisfactory performance.
\section{Conclusion}
In this paper, we formulated the resource allocation algorithm design for robust and secure multiuser MISO downlink systems with self-sustainable IRS as a non-convex optimization problem. To maximize the system sum-rate, the transmit beamformers, AN covariance matrix, IRS phase shifts, and the power harvesting schedule were jointly optimized while satisfying secrecy constraints for the potential eavesdroppers and taking into account the imperfection of  CSI of all channels. The proposed algorithm can achieve an effective solution to the formulated non-convex design problem. Simulation results demonstrated that the proposed scheme offers significant performance gain compared with the conventional MISO systems without IRS. Moreover, our results also unveiled the non-trivial trade-off between achieving IRS self-sustainability and the system sum-rate. In particular, the system sum-rate of a IRS-assisted system is gradually saturated as the number of energy harvesting IRS elements increased. Additionally,  we confirmed that a small bit resolution of the IRS phase shifters can achieve a considerable average system sum-rate of the ideal case with continuous phase shifters. Lastly, our results verified the potential of IRSs to improve the physical layer security and confirmed that the robustness of the proposed scheme against the CSI imperfectness.
\appendices
\section{Proof of Theorem \ref{Thm:rank_one_W}}\label{Proof:rankoneW}
To facilicate the proof, we first introduce three slack optimization variables, i.e.  $\mathbf{C}_{\overline{\mathrm{C5a}},k,j}$, $\mathbf{C}_{\overline{\overline{\mathrm{C9}}},k}$, and $\mathbf{C}_{\overline{\overline{\mathrm{C10}}},i,k}$, and three constraints, i.e. $\mathrm{C14}$, $\mathrm{C15}$, and $\mathrm{C16}$. Then, the rank constraint relaxed version of problem \eqref{optimization_problem_sca} can be transformed into its equivalent form, as shown in \eqref{optimization_problem_proof} at the top of next page.
\begin{figure*}
\begin{align}
&\underset{\substack{\mathbf{W}_k\in\mathbb{H}^{M_{\mathrm{t}}}, \mathbf{Z}\in\mathbb{H}^{M_{\mathrm{t}}},\, \mathbf{S}, \mathbf{Y}\in\mathbb{H}^{N+2},\mathbf{V}\in\mathbb{H}^{N+1},\mathbf{v},  \xi_k,\iota_k,\beta_{\mathrm{PR}},\,\mathcal{U},\epsilon,\mathbf{\Psi}}}{\mathrm{maximize}} \,\,N_{\mathrm{obj}}(\xi_k,\iota_k)\!\!-\!\!\!\!\sum_{k\in\mathcal{K}}\!\! \nabla^{\mathrm{H}}_{\mathbf{\iota}_k}\hspace{-1mm}D_{\mathrm{obj}}(\iota_k^{(t)}\!)(\iota_k\!-\!\iota_k^{(t)}\!)\!\!-\!\!D_{\mathrm{obj}}(\iota_k^{(t)}\!)
  \label{optimization_problem_proof} \\
&\mathrm{s.t.}\,\,\mathrm{C1},\mathrm{\overline{\overline{C3a}}},\mathrm{\overline{\overline{C3b}}},\overline{\mathrm{\overline{C3c}}},\mathrm{\overline{C3d}},
\mathrm{\overline{C3e}},\mathrm{\overline{\overline{C3fa}}},\mathrm{\overline{\overline{C3fb}}},\mathrm{{\overline{C3g}}},\mathrm{\overline{C3h}},
\mathrm{C4a},\mathrm{C4b},\overline{\mathrm{C4d}},\mathrm{C4e},\overline{\mathrm{C5b}},\overline{\mathrm{C5c}},\mathrm{C6},\mathrm{C7},\mathrm{C11a},\overline{\mathrm{C11b}},\overline{\mathrm{C12a}},\notag\\
&\hspace{6mm}\overline{\mathrm{C12b}},\overline{\mathrm{C13a}},\overline{\mathrm{C13b}}.\notag\\
&\widetilde{\overline{\mathrm{C5a}}}\hspace{-1mm}:\|\mathbf{C}_{\overline{\mathrm{C5a}},k,j}\|^2_{\mathrm{F}}+\|\mathbf{W}_k^{(t)}\|^2_{\mathrm{F}}-2\mathrm{Tr}\Big((\mathbf{W}_k^{(t)})^{\mathrm{H}}\mathbf{W}_k\Big)+(2^{\tau_{k,j}})
\|\mathbf{\Psi}_{\mathrm{eve},j}^{(t)}\|^2_{\mathrm{F}}
-(2^{\tau_{k,j}+1})\mathrm{Tr}\Big((\mathbf{\Psi}_{\mathrm{eve},j}^{(t)})^{\mathrm{H}}\mathbf{\Psi}_{\mathrm{eve},j}\Big)-(2^{\tau_{k,j}+1}-2)\sigma^2_\mathrm{eve}\notag\\
&\hspace{7mm}+(2^{\tau_{k,j}}-1)\Big\{\|\mathbf{Z}-\mathbf{\Psi}_{\mathrm{eve},j} \|^2_{\mathrm{F}}
+ \|\mathbf{Z}^{(t)}\|^2_{\mathrm{F}}-2\mathrm{Tr}\Big((\mathbf{Z}^{(t)})^{\mathrm{H}}\mathbf{Z}\Big)\Big\}\leq0, \forall k,j,\notag\\
&\widetilde{\overline{\overline{\mathrm{C9}}}}\hspace{-1mm}:\xi_k+ \frac{1}{2}\|\mathbf{C}_{\overline{\overline{\mathrm{C9}}},k}\|^2_{\mathrm{F}}
+\frac{1}{2}\|\mathbf{\Psi}_{\mathrm{C9},k}^{(t)}\|^2_{\mathrm{F}}-\mathrm{Tr}\Big((\mathbf{\Psi}_{\mathrm{C9},k}^{(t)})^{\mathrm{H}}\mathbf{\Psi}_{\mathrm{C9},k}\Big)+\frac{1}{2}\|\mathbf{W}_k^{(t)}\|^2_{\mathrm{F}}-\mathrm{Tr}\Big((\mathbf{W}_k^{(t)})^{\mathrm{H}}\mathbf{W}_k\Big)\leq 0,\forall k,\notag\\
&\widetilde{\overline{\overline{\mathrm{C10}}}}\hspace{-1mm}:  \sum_{i\in\mathcal{K}/ \{k\}}\Big\{\frac{1}{2}\|\mathbf{C}_{\overline{\overline{\mathrm{C10}}},i,k}\|^2_{\mathrm{F}}\Big\}
+\frac{1}{2}\|\mathbf{Z}+\mathbf{\Psi}_{\mathrm{C10},k}\|^2_{\mathrm{F}}+\frac{K}{2}\|\mathbf{\Psi}_{\mathrm{C10},k}^{(t)}\|^2_{\mathrm{F}}-K\mathrm{Tr}\Big((\mathbf{\Psi}_{\mathrm{C10},k}^{(t)})^{\mathrm{H}}\mathbf{\Psi}_{\mathrm{C10},k}\Big)\notag\\
&\hspace{7mm}+\sum_{i\in\mathcal{K}/ \{k\}}\Big\{\frac{1}{2}\|\mathbf{W}_i^{(t)}\|^2_{\mathrm{F}}-\mathrm{Tr}\Big((\mathbf{W}_i^{(t)})^{\mathrm{H}}\mathbf{W}_i\Big)\Big\}+\frac{1}{2}\|\mathbf{Z}^{(t)}\|^2_{\mathrm{F}}-\mathrm{Tr}\Big((\mathbf{Z}^{(t)})^{\mathrm{H}}\mathbf{Z}\Big) \leq\iota_k ,\forall k, \notag\\
& \mathrm{C14}\hspace{-1mm}:\,\, \mathbf{C}_{\overline{\mathrm{C5a}},k,j}\succeq \mathbf{W}_k+\mathbf{\Psi}_\mathrm{eve,j},\forall k,j,\notag\\
&\mathrm{C15}\hspace{-1mm}:\,\, \mathbf{C}_{\overline{\overline{\mathrm{C9}}},k}\succeq \mathbf{W}_k-\mathbf{\Psi}_{\overline{\overline{\mathrm{C9}}},k},\forall k,\hspace{7mm}\mathrm{C16}\hspace{-1mm}:\,\, \mathbf{C}_{\overline{\overline{\mathrm{C10}}},i,k}\succeq \mathbf{W}_i+\mathbf{\Psi}_{\mathrm{C10},k},\forall i,k.\notag
\end{align}
\hrulefill
\end{figure*}

Since \eqref{optimization_problem_proof} is jointly convex with respect to the optimization variables and satisfies Slater's constraint qualification, its strong duality holds. The Lagrangian function of problem \eqref{optimization_problem_proof} in terms of $\mathbf{W}_k$ is given in \eqref{Lagrangian_modify}.
\begin{figure*}
\begin{align}
{\mathcal{L}} = & -\delta_{\mathrm{C1}}\mathrm{Tr}(\sum_{k\in\mathcal{K}}\mathbf{W}_k)+
\sum_{n\in\mathcal{N}}\mathrm{Tr}(\bm{\mathcal{S}}_{{{\mathrm{C3f}}},n}\mathbf{D}_{{{\mathrm{C3f}}},n})+\sum_{k\in\mathcal{K}}\sum_{j\in\mathcal{J}}2\delta_{\widetilde{\overline{\mathrm{C5a}}},k,j}
\mathrm{Tr}((\mathbf{W}_k^{(t)})^{\mathrm{H}}\mathbf{W}_k)+\sum_{k\in\mathcal{K}}\mathrm{Tr}(\mathbf{M}_{\mathrm{C7},k}\mathbf{W}_k)\label{Lagrangian_modify}\\
&+\sum_{k\in\mathcal{K}}\delta_{\widetilde{\overline{\overline{\mathrm{C9}}}},k}\mathrm{Tr}((\mathbf{W}_k^{(t)})^{\mathrm{H}}\mathbf{W}_k)
+\sum_{k\in\mathcal{K}}\delta_{\widetilde{\overline{\overline{\mathrm{C10}}}},k}\sum_{i\in\mathcal{K}/\{k\}}\mathrm{Tr}((\mathbf{W}_i^{(t)})^{\mathrm{H}}\mathbf{W}_i)\Big]-\hspace{-1mm}\sum_{j\in\mathcal{J}}\sum_{k\in\mathcal{K}}\mathrm{Tr}
(\mathbf{M}_{\mathrm{C14},k,j}\mathbf{W}_k)\hspace{-1mm}\notag\\
&-\hspace{-1mm}\sum_{k\in\mathcal{K}}\mathrm{Tr}
(\mathbf{M}_{\mathrm{C15},k}\mathbf{W}_k)-\sum_{i\in\mathcal{K}}\sum_{k\in\mathcal{K}}
\mathrm{Tr}(\mathbf{M}_{\mathrm{C16},i,k}\mathbf{W}_k)+\Delta.\notag
\end{align}
\hrulefill
\end{figure*}
 $\Delta$ in \eqref{Lagrangian_modify} denotes the collection of terms that are irrelevant to $\mathbf{W}_k$. In \eqref{Lagrangian_modify}, $\delta_{\mathrm{C1},k}, \delta_{\widetilde{\overline{\mathrm{C5a}}},k,j}, \delta_{\widetilde{\overline{\overline{\mathrm{C9}}}},k}$, and $\delta_{\widetilde{\overline{\overline{\mathrm{C10}}}},k}$ are the Lagrange multiplier for the constraints $\mathrm{C1}, \widetilde{\overline{\mathrm{C5a}}}, \widetilde{\overline{\overline{\mathrm{C9}}}}$, and $\widetilde{\overline{\overline{\mathrm{C10}}}}$, respectively. $\mathbf{D}_{{{\mathrm{C3f}}},n}$, $\mathbf{M}_{\mathrm{C7},k}$, $\mathbf{M}_{\mathrm{C14},k,j}$, $\mathbf{M}_{\mathrm{C15},k}$, and $\mathbf{M}_{\mathrm{C16},i,k}$ are the Lagrange multiplier matrices corresponding to constraints $\overline{\overline{\mathrm{C3f}}}$, $\mathrm{C7}$, $\mathrm{C14}$, $\mathrm{C15}$, and $\mathrm{C16}$ respectively.
Then, examining the Karush-Kuhn-Tucker (KKT) conditions for problem \eqref{optimization_problem_proof} yields
\begin{align}
\mathrm{K1}\hspace{-1mm}:&\, \mathbf{D}_{{{\mathrm{C3f}}},n}^*\!\succeq\!\mathbf{0}, \mathbf{M}_{\mathrm{C7},k}^*\!\succeq\!\mathbf{0}, \mathbf{M}_{\mathrm{C14},k,j}^*\!\succeq\!\mathbf{0},\mathbf{M}_{\mathrm{C15},k}^*\!\succeq\!\mathbf{0},\label{KKT:K1_modify}\\
&\mathbf{M}_{\mathrm{C16},i,k}^*\!\succeq\!\mathbf{0},\delta_{\mathrm{C1},k}\!\geq\!0, \delta_{\overline{\mathrm{C5a}},k,j}\!\geq\!0, \delta_{\overline{\overline{\mathrm{C9}}},k}\geq0,\delta_{\overline{\overline{\mathrm{C10}}},k}\!\geq\!0  \notag\\
\mathrm{K2}\hspace{-1mm}:&\,\mathbf{M}_{\mathrm{C7},k}^*\mathbf{W}_k^*\!=\!\mathbf{0}, \mathbf{M}_{\mathrm{C14},k,j}^*\mathbf{W}_k^*\!=\!\mathbf{0}, \mathbf{M}_{\mathrm{C15},k}^*\mathbf{W}_k^*\!=\!\mathbf{0}, \label{KKT:K2_modify}\\ &\mathbf{M}_{\mathrm{C16},i,k}^*\mathbf{W}_k^*=\mathbf{0}, \forall i,k\text{, and }\notag\\
\mathrm{K3}\hspace{-1mm}:&\,\nabla_{\mathbf{W}_k}\mathcal{L} = \mathbf{0}. \label{KKT:K3_modify}
\end{align}
For ease of presentation, the K3 in \eqref{KKT:K3_modify} can be recast as:
\begin{align}
 \delta_{\mathrm{C1},k}^*\mathbf{I}_{M_t}-\mathbf{\Upsilon}^*_k = \mathbf{M}_{\mathrm{C7},k}^*, \label{K3_rewritten_modify}
\end{align}
where $\mathbf{\Upsilon}^*_k = \mathbf{B}^* +2\sum_{k\in\mathcal{K}}\sum_{j\in\mathcal{J}}\delta_{\overline{\mathrm{C5a}},k,j}^*\mathbf{W}_k^{(t)}-\sum_{k\in\mathcal{K}}\mathbf{M}_{\mathrm{C15},k}^*+ \sum_{k\in\mathcal{K}}\delta_{\overline{\overline{\mathrm{C9}}},k}^*\mathbf{W}_k^{(t)}+\delta_{\overline{\overline{\mathrm{C10}}},k}^* \sum_{i\in\mathcal{K}/\{k\}}\mathbf{W}_i^{(t)}-\sum_{j\in\mathcal{J}}\sum_{k\in\mathcal{K}}\mathbf{M}_{\mathrm{C14},k,j}^*-\sum_{i\in\mathcal{K}}\sum_{k\in\mathcal{K}}\mathbf{M}_{\mathrm{C16},i,k}^*$
and $\mathbf{B}^*=\mathrm{Tr}(\mathbf{O}_{\hat{\mathbf{g}}_n}\mathbf{D}_{{{\mathrm{C3f}}},n}^*\mathbf{O}_{\hat{\mathbf{g}}_n}^{\mathrm{H}})\mathbf{I}_{M_t}$.

From \eqref{KKT:K2_modify}, we know that matrix $\mathbf{W}^*_k$ lies in the null space of $\mathbf{M}_{\mathrm{C7},k}^*$. To reveal the rank of $\mathbf{W}^*_k$, we first investigate the structure of $\mathbf{M}_{\mathrm{C7},k}^*$.
When $\delta_{\mathrm{C1},k}^* =0$, $\mathbf{\Upsilon}^*_k \preceq\mathbf{0}$ which leads to the dual problem of problem \eqref{optimization_problem_proof}  unbounded.
Thus $\delta_{\mathrm{C1},k}^* >0$ holds. According to \eqref{K3_rewritten_modify}, if $\lambda_{\max}(\mathbf{\Upsilon}^*_k )> \delta_{\mathrm{C1},k}^*$, then $\lambda_{\max}(\mathbf{M}_{\mathrm{C7},k}^*)<0$, which contradicts $\mathrm{K1}$ in \eqref{KKT:K1_modify}. If $\lambda_{\max}(\mathbf{\Upsilon}^*_k )< \delta_{\mathrm{C1},k}^*$, $\mathbf{M}_{\mathrm{C7},k}^*$ is a positive definite matrix with full rank, which yields the solution $\mathbf{W}^*_k=\mathbf{0}$ and $\mathrm{Rank}(\mathbf{W}^*_k)=0$. If $\lambda_{\max}(\mathbf{\Upsilon}^*_k )=\delta_{\mathrm{C1},k}^*$ holds at the optimal solution, in order to have a bounded optimal dual solution, the null space of  $\mathbf{M}_{\mathrm{C7},k}^*$ is spanned by vector $\mathbf{r}_{\max}\in\mathbb{C}^{M_t+1}$, which is a unit-norm eigenvector of $\mathbf{\Upsilon}^*_k$ associated with eigenvalue  $\lambda_{\max}(\mathbf{\Upsilon}^*_k)$. As a result, there exists the optimal beamforming matrix given by
$\mathbf{W}^*_k = \nu \mathbf{r}_{\max} \mathbf{r}_{\max}^{\mathrm{H}}$,
where $\nu$ is a parameter which makes the power consumption of transmitter satisfies constraint C1.

\vspace{-2mm}
\bibliographystyle{IEEEtran}\vspace{-2mm}
\bibliography{IRS-EH}

\begin{thebibliography}{10}

\bibitem{hu2020sum}
S.~Hu, Z.~Wei, Y.~Cai, D.~W.~K. Ng, and J.~Yuan, ``Sum-rate maximization for
  multiuser {MISO} downlink systems with self-sustainable {IRS},'' in {\em
  GLOBECOM 2020 - 2020 IEEE Global Commun. Conf.}, pp.~1--7, Dec. 2020.

\bibitem{zhang20196g}
Z.~Zhang, Y.~Xiao, Z.~Ma, M.~Xiao, Z.~Ding, X.~Lei, G.~K. Karagiannidis, and
  P.~Fan, ``6{G} {wireless networks: Vision, requirements, architecture, and
  key technologies},'' {\em IEEE Veh. Technol. Mag.}, vol.~14, pp.~28--41, Sep.
  2019.

\bibitem{wu2017overview}
Q.~Wu, G.~Y. Li, W.~Chen, D.~W.~K. Ng, and R.~Schober, ``An overview of
  sustainable green {5G} networks,'' {\em IEEE Wirel. Commun.}, vol.~24,
  pp.~72--80, Aug. 2017.

\bibitem{shafi20175g}
M.~Shafi, A.~F. Molisch, P.~J. Smith, T.~Haustein, P.~Zhu, P.~De~Silva,
  F.~Tufvesson, A.~Benjebbour, and G.~Wunder, ``{5G}: A tutorial overview of
  standards, trials, challenges, deployment, and practice,'' {\em IEEE J.
  Select. Areas Commun.}, vol.~35, pp.~1201--1221, Apr. 2017.

\bibitem{wu2019intelligent}
Q.~Wu and R.~Zhang, ``{Intelligent reflecting surface enhanced wireless network
  via joint active and passive beamforming},'' {\em IEEE Trans. Wireless
  Commun.}, vol.~18, pp.~5394--5409, Aug. 2019.

\bibitem{wu2020intelligent}
Q.~Wu, S.~Zhang, B.~Zheng, C.~You, and R.~Zhang, ``Intelligent reflecting
  surface aided wireless communications: A tutorial,'' {\em arXiv preprint
  arXiv:2007.02759}, 2020.

\bibitem{wu2019weighted}
Q.~{Wu} and R.~{Zhang}, ``{Weighted sum power maximization for intelligent
  reflecting surface aided {SWIPT}},'' {\em IEEE Wireless Commun. Lett.},
  pp.~1--1, Dec. 2019.

\bibitem{liu2021deep}
C.~Liu, X.~Liu, Z.~Wei, S.~Hu, D.~W.~K. Ng, and J.~Yuan, ``Deep
  learning-empowered predictive beamforming for {IRS}-assisted multi-user
  communications,'' {\em arXiv preprint arXiv:2104.12309}, 2021.

\bibitem{yu2021smart}
X.~Yu, V.~Jamali, D.~Xu, D.~W.~K. Ng, and R.~Schober, ``Smart and
  reconfigurable wireless communications: From {IRS} modeling to algorithm
  design,'' {\em arXiv preprint arXiv:2103.07046}, 2021.

\bibitem{9145224}
Y.~{Cai}, Z.~{Wei}, S.~{Hu}, D.~W.~K. {Ng}, and J.~{Yuan}, ``Resource
  allocation for power-efficient {IRS}-assisted {UAV} communications,'' in {\em
  Proc. IEEE Intern. Commun. Conf. Workshops (ICC Workshops)}, pp.~1--7, Jun.
  2020.

\bibitem{abeywickrama2020intelligent}
S.~Abeywickrama, R.~Zhang, Q.~Wu, and C.~Yuen, ``Intelligent reflecting
  surface: Practical phase shift model and beamforming optimization,'' {\em
  IEEE Trans. Commun.}, vol.~68, no.~9, pp.~5849--5863, Jun. 2020.

\bibitem{pan2020multicell}
C.~Pan, H.~Ren, K.~Wang, W.~Xu, M.~Elkashlan, A.~Nallanathan, and L.~Hanzo,
  ``Multicell {MIMO} communications relying on intelligent reflecting
  surfaces,'' {\em IEEE Trans. Wireless Commun.}, vol.~19, no.~8,
  pp.~5218--5233, May 2020.

\bibitem{zhang2020joint}
Z.~Zhang and L.~Dai, ``A joint precoding framework for wideband reconfigurable
  intelligent surface-aided cell-free network,'' {\em arXiv preprint
  arXiv:2002.03744}, 2020.

\bibitem{mendez2016hybrid}
R.~M{\'e}ndez-Rial, C.~Rusu, N.~Gonz{\'a}lez-Prelcic, A.~Alkhateeb, and R.~W.
  Heath, ``Hybrid {MIMO} architectures for millimeter wave communications:
  Phase shifters or switches?,'' {\em IEEE Access}, vol.~4, pp.~247--267, Jan.
  2016.

\bibitem{ribeiro2018energy}
L.~N. Ribeiro, S.~Schwarz, M.~Rupp, and A.~L. de~Almeida, ``Energy efficiency
  of mmwave massive {MIMO} precoding with low-resolution {DACs},'' {\em IEEE J.
  Sel. Topics Signal Process.}, vol.~12, no.~2, pp.~298--312, May 2018.

\bibitem{huang2018energy}
C.~Huang, G.~C. Alexandropoulos, A.~Zappone, M.~Debbah, and C.~Yuen, ``Energy
  efficient multi-user {MISO} communication using low resolution large
  intelligent surfaces,'' in {\em 2018 IEEE Globe. Commun. Workshops},
  pp.~1--6, Dec. 2018.

\bibitem{huang2019reconfigurable}
C.~{Huang}, A.~{Zappone}, G.~C. {Alexandropoulos}, M.~{Debbah}, and C.~{Yuen},
  ``Reconfigurable intelligent surfaces for energy efficiency in wireless
  communication,'' {\em IEEE Trans. Wireless Commun.}, vol.~18, pp.~4157--4170,
  Jun. 2019.

\bibitem{wu2019towards}
Q.~Wu and R.~Zhang, ``Towards smart and reconfigurable environment: Intelligent
  reflecting surface aided wireless network,'' {\em IEEE Commun. Mag.},
  vol.~58, pp.~106--112, Nov. 2019.

\bibitem{clerckx2018fundamentals}
B.~Clerckx, R.~Zhang, R.~Schober, D.~W.~K. Ng, D.~I. Kim, and H.~V. Poor,
  ``Fundamentals of wireless information and power transfer: From {RF} energy
  harvester models to signal and system designs,'' {\em IEEE J. Select. Areas
  Commun.}, vol.~37, pp.~4--33, Sep. 2018.

\bibitem{lyu2020optimized}
B.~Lyu, P.~Ramezani, D.~T. Hoang, S.~Gong, Z.~Yang, and A.~Jamalipour,
  ``Optimized energy and information relaying in self-sustainable
  {IRS}-empowered {WPCN},'' {\em arXiv preprint arXiv:2004.03108}, Apr. 2020.

\bibitem{zou2020wireless}
Y.~Zou, S.~Gong, J.~Xu, W.~Cheng, D.~T. Hoang, and D.~Niyato, ``Wireless
  powered intelligent reflecting surfaces for enhancing wireless
  communications,'' {\em IEEE Trans. Veh. Technol}, vol.~69, no.~10,
  pp.~12369--12373, Jul. 2020.

\bibitem{le2008efficient}
T.~Le, K.~Mayaram, and T.~Fiez, ``Efficient far-field radio frequency energy
  harvesting for passively powered sensor networks,'' {\em IEEE J. Solid-State
  Circuits}, vol.~43, pp.~1287--1302, Apr. 2008.

\bibitem{guo2012improved}
J.~Guo and X.~Zhu, ``An improved analytical model for {RF-DC} conversion
  efficiency in microwave rectifiers,'' in {\em 2012 IEEE/MTT-S International
  Microwave Symposium Digest}, pp.~1--3, IEEE, Aug. 2012.

\bibitem{sun2019physical}
X.~Sun, D.~W.~K. Ng, Z.~Ding, Y.~Xu, and Z.~Zhong, ``Physical layer security in
  {UAV} systems: Challenges and opportunities,'' {\em IEEE Wirel. Commun.},
  vol.~26, pp.~40--47, Oct. 2019.

\bibitem{dong2009cooperative}
L.~Dong, Z.~Han, A.~P. Petropulu, and H.~V. Poor, ``Cooperative jamming for
  wireless physical layer security,'' in {\em 2009 IEEE/SP 15th Workshop
  Statistical Signal Process.}, pp.~417--420, IEEE, Aug. 2009.

\bibitem{sun2018robust}
Y.~Sun, D.~W.~K. Ng, J.~Zhu, and R.~Schober, ``Robust and secure resource
  allocation for full-duplex {MISO} multicarrier {NOMA} systems,'' {\em IEEE
  Trans. Commun.}, vol.~66, pp.~4119--4137, Sep. 2018.

\bibitem{shen2019secrecy}
H.~Shen, W.~Xu, S.~Gong, Z.~He, and C.~Zhao, ``Secrecy rate maximization for
  intelligent reflecting surface assisted multi-antenna communications,'' {\em
  IEEE Wireless Commun. Lett.}, vol.~23, pp.~1488--1492, Jun. 2019.

\bibitem{cui2019secure}
M.~Cui, G.~Zhang, and R.~Zhang, ``Secure wireless communication via intelligent
  reflecting surface,'' {\em IEEE Wireless Commun. Lett.}, vol.~8,
  pp.~1410--1414, May 2019.

\bibitem{yu2019robust}
X.~Yu, D.~Xu, Y.~Sun, D.~W.~K. Ng, and R.~Schober, ``Robust and secure wireless
  communications via intelligent reflecting surfaces,'' {\em IEEE J. Select.
  Areas Commun.}, vol.~38, pp.~2637--2652, Jul. 2020.

\bibitem{zhou2020robust}
G.~Zhou, C.~Pan, H.~Ren, K.~Wang, M.~Di~Renzo, and A.~Nallanathan, ``Robust
  beamforming design for intelligent reflecting surface aided {MISO}
  communication systems,'' {\em IEEE Commun. Lett.}, vol.~9, pp.~1658--1662,
  Jun. 2020.

\bibitem{sun2016multi}
Y.~Sun, D.~W.~K. Ng, J.~Zhu, and R.~Schober, ``Multi-objective optimization for
  robust power efficient and secure full-duplex wireless communication
  systems,'' {\em IEEE Trans. Wireless Commun.}, vol.~15, Apr. 2016.

\bibitem{zhao2015joint}
X.~Zhao, J.~Xiao, Q.~Li, Q.~Zhang, and J.~Qin, ``Joint optimization of
  {AN}-aided transmission and power splitting for {MISO} secure communications
  with {SWIPT},'' {\em IEEE Commun. Lett.}, vol.~19, pp.~1969--1972, Sep. 2015.

\bibitem{zhang2016secure}
H.~Zhang, Y.~Huang, C.~Li, and L.~Yang, ``Secure beamforming design for {SWIPT}
  in {MISO} broadcast channel with confidential messages and external
  eavesdroppers,'' {\em IEEE Trans. Wireless Commun.}, vol.~15, pp.~7807--7819,
  Sep. 2016.

\bibitem{arunabha2010fundamentals}
G.~Arunabha, J.~Zhang, J.~G. Andrews, and R.~Muhamed, ``Fundamentals of
  {LTE},'' {\em The Prentice Hall communications engineering and emerging
  technologies series}, 2010.

\bibitem{9045989}
R.~{Li}, Z.~{Wei}, L.~{Yang}, D.~W.~K. {Ng}, J.~{Yuan}, and J.~{An}, ``Resource
  allocation for secure {Multi-UAV} communication systems with
  multi-eavesdropper,'' {\em IEEE Trans. Commun.}, vol.~68, pp.~4490--4506,
  Mar. 2020.

\bibitem{wu2019beamforming}
Q.~{Wu} and R.~{Zhang}, ``{Beamforming optimization for intelligent reflecting
  surface with discrete phase shifts},'' in {\em Proc. IEEE Intern. Conf. on
  Acoustics, Speech and Signal Process.}, pp.~7830--7833, May 2019.

\bibitem{abdelhalem2013rf}
S.~H. Abdelhalem, P.~S. Gudem, and L.~E. Larson, ``An {RF-DC} converter with
  wide-dynamic-range input matching for power recovery applications,'' {\em
  IEEE Trans. Circuits Syst. II Express Briefs IEEE T CIRCUITS-II}, vol.~60,
  pp.~336--340, Apr. 2013.

\bibitem{hameed2014hybrid}
Z.~Hameed and K.~Moez, ``Hybrid forward and backward threshold-compensated
  {RF-DC} power converter for {RF} energy harvesting,'' {\em IEEE Trans. Emerg.
  Sel. Topics Circuits Syst.}, vol.~4, no.~3, pp.~335--343, Jul. 2014.

\bibitem{nintanavongsa2012design}
P.~Nintanavongsa, U.~Muncuk, D.~R. Lewis, and K.~R. Chowdhury, ``Design
  optimization and implementation for {RF} energy harvesting circuits,'' {\em
  IEEE Trans. Emerg. Sel. Topics Circuits Syst.}, vol.~2, no.~1, pp.~24--33,
  Feb. 2012.

\bibitem{boshkovska2017resource}
E.~{Boshkovska}, D.~W.~K. {Ng}, N.~{Zlatanov}, and R.~{Schober}, ``Practical
  non-linear energy harvesting model and resource allocation for {SWIPT}
  systems,'' {\em IEEE Commun. Lett.}, vol.~19, pp.~2082--2085, Sep. 2015.

\bibitem{chen2019channel}
J.~Chen, Y.-C. Liang, H.~V. Cheng, and W.~Yu, ``Channel estimation for
  reconfigurable intelligent surface aided multi-user {MIMO} systems,'' {\em
  arXiv preprint arXiv:1912.03619}, 2019.

\bibitem{7106496}
D.~W.~K. {Ng} and R.~{Schober}, ``Secure and green {SWIPT} in distributed
  antenna networks with limited backhaul capacity,'' {\em IEEE Trans. Wireless
  Commun.}, vol.~14, pp.~5082--5097, May 2015.

\bibitem{hu2019two}
C.~Hu, L.~Dai, S.~Han, and X.~Wang, ``Two-timescale channel estimation for
  reconfigurable intelligent surface aided wireless communications,'' {\em IEEE
  Trans. Commun.}, pp.~1--1, Apr. 2021.

\bibitem{mukherjee2012detecting}
A.~Mukherjee and A.~L. Swindlehurst, ``Detecting passive eavesdroppers in the
  {MIMO} wiretap channel,'' in {\em 2012 IEEE Int. Conf. Acoustics, Speech and
  Signal Process. (ICASSP)}, pp.~2809--2812, IEEE, Mar. 2012.

\bibitem{zhao2015robust}
P.~Zhao, M.~Zhang, H.~Yu, H.~Luo, and W.~Chen, ``Robust beamforming design for
  sum secrecy rate optimization in {MU-MISO} networks,'' {\em IEEE Trans. Inf.
  Forensics Secur.}, vol.~10, pp.~1812--1823, Apr. 2015.

\bibitem{8333737}
F.~{Zhou}, Z.~{Chu}, H.~{Sun}, R.~Q. {Hu}, and L.~{Hanzo}, ``Artificial noise
  aided secure cognitive beamforming for cooperative {MISO-NOMA} using
  {SWIPT},'' {\em IEEE J. Select. Areas Commun.}, vol.~36, Apr. 2018.

\bibitem{5963524}
H.~{Qin}, X.~{Chen}, Y.~{Sun}, M.~{Zhao}, and J.~{Wang}, ``Optimal power
  allocation for joint beamforming and artificial noise design in secure
  wireless communications,'' in {\em 2011 IEEE Int. Conf. Commun. Workshops
  (ICC)}, Jul. 2011.

\bibitem{6638604}
D.~{Li}, C.~{Shen}, and Z.~{Qiu}, ``Sum rate maximization and energy harvesting
  for two-way af relay systems with imperfect {CSI},'' in {\em 2013 IEEE Int.
  Conf. Acoustics, Speech and Signal Process.}, pp.~4958--4962, May 2013.

\bibitem{6626322}
Y.~{Tang}, J.~{Xiong}, D.~{Ma}, and X.~{Zhang}, ``Robust artificial noise aided
  transmit design for {MISO} wiretap channels with channel uncertainty,'' {\em
  IEEE Commun. Lett.}, vol.~17, pp.~2096--2099, Oct. 2013.

\bibitem{moore1997rank}
J.~B. Moore and D.~Jiang, ``A rank preserving flow algorithm for quadratic
  optimization problems subject to quadratic equality constraints,'' in {\em
  1997 IEEE Int. Conf. Acoustics, Speech, and Signal Process.}, vol.~1,
  pp.~67--70, IEEE, 1997.

\bibitem{yu2020power}
X.~Yu, D.~Xu, D.~W.~K. Ng, and R.~Schober, ``Power-efficient resource
  allocation for multiuser {MISO} systems via intelligent reflecting
  surfaces,'' {\em arXiv preprint arXiv:2005.06703}, 2020.

\bibitem{boyd2004convex}
S.~Boyd, S.~P. Boyd, and L.~Vandenberghe, {\em Convex optimization}.
\newblock Cambridge university press, 2004.

\bibitem{luo2004multivariate}
Z.-Q. Luo, J.~F. Sturm, and S.~Zhang, ``Multivariate nonnegative quadratic
  mappings,'' {\em SIAM J. Optim.}, vol.~14, no.~4, pp.~1140--1162, 2004.

\bibitem{opial1967weak}
Z.~Opial, ``Weak convergence of the sequence of successive approximations for
  nonexpansive mappings,'' {\em Bull. Am. Math. Soc.}, vol.~73, no.~4,
  pp.~591--597, 1967.

\bibitem{polik2010interior}
I.~P{\'o}lik and T.~Terlaky, ``Interior point methods for nonlinear
  optimization,'' in {\em Nonlinear optimization}, pp.~215--276, Springer, Jul.
  2010.

\bibitem{cormen2009introduction}
T.~H. Cormen, C.~E. Leiserson, R.~L. Rivest, and C.~Stein, {\em Introduction to
  algorithms}.
\newblock MIT press, Jul. 2009.

\bibitem{goldsmith2005wireless}
A.~Goldsmith, {\em {Wireless Communications}}.
\newblock Cambridge university press, 2005.

\bibitem{mardeni2010optimised}
R.~Mardeni and T.~S. Priya, ``Optimised {COST-231} hata models for {WiMAX} path
  loss prediction in suburban and open urban environments,'' {\em Mod. Appl.
  Sci.}, vol.~4, no.~9, p.~75, Sep. 2010.

\bibitem{yao2008fully}
Y.~Yao, J.~Wu, Y.~Shi, and F.~F. Dai, ``A fully integrated 900-{MHz} passive
  {RFID} transponder front end with novel zero-threshold {RF-DC} rectifier,''
  {\em IEEE Trans. Ind. Electron.}, vol.~56, pp.~2317--2325, Nov. 2008.

\bibitem{tse2005fundamentals}
D.~Tse and P.~Viswanath, {\em Fundamentals of wireless communication}.
\newblock Cambridge university press, 2005.

\bibitem{wei2017optimal}
Z.~Wei, D.~W.~K. Ng, J.~Yuan, and H.-M. Wang, ``Optimal resource allocation for
  power-efficient {MC-NOMA} with imperfect channel state information,'' {\em
  IEEE Trans. Commun.}, vol.~65, pp.~3944--3961, May 2017.

\end{thebibliography}
\vspace{-5mm}
\begin{IEEEbiography}[{\includegraphics[width=0.95in,height=1.3in]{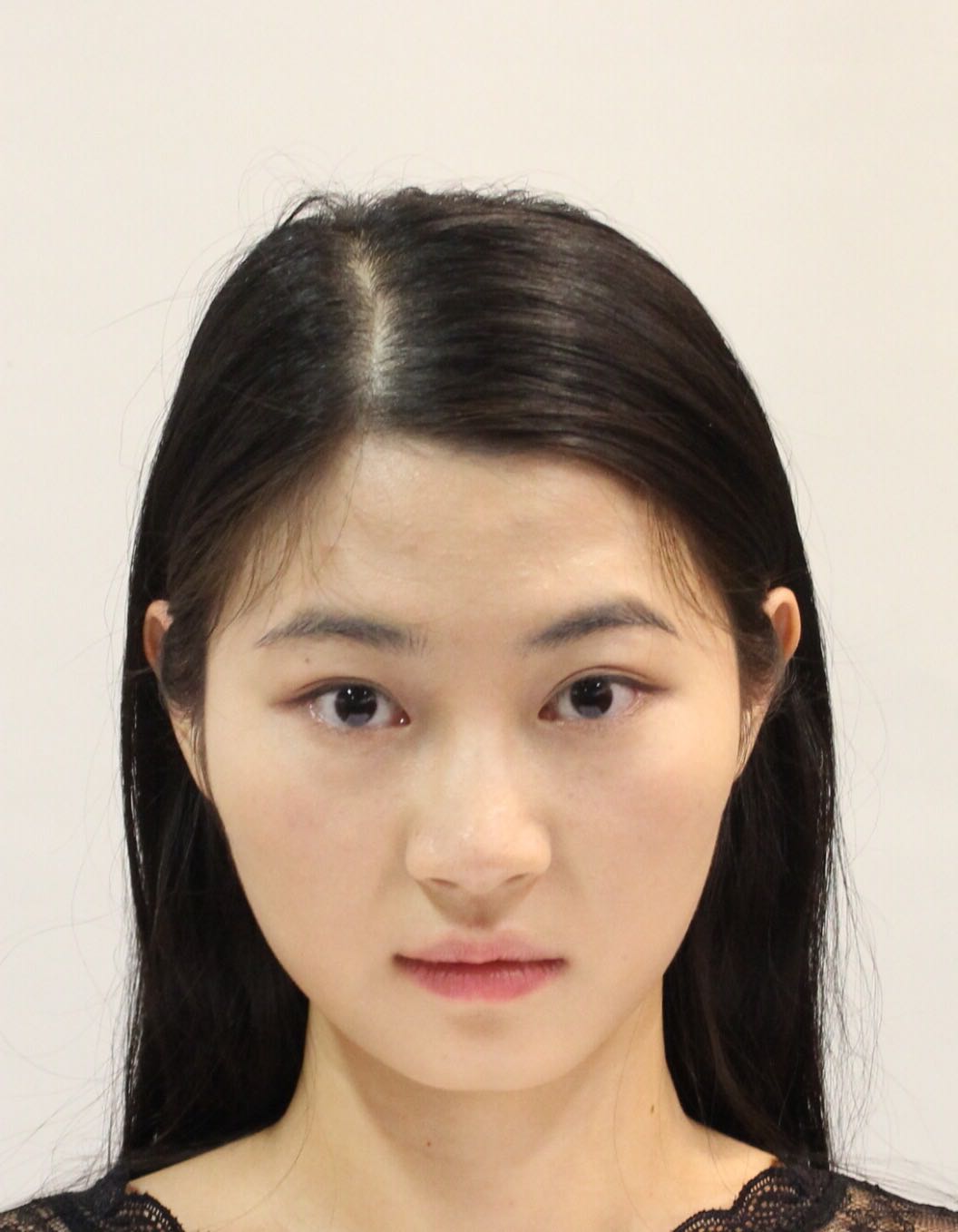}}]{Shaokang Hu (S'20)} received the B.S. degree in Engineering and Telecommunication from the University of New South Wales, Australia, Sydney in 2017 and the M.E. degree in Electrical Engineering and Telecommunication from the University of New South Wales, Sydney, Australia, in 2018. She is currently pursuing the Ph.D. degree in Telecommunication at the University of New South Wales, Sydney, Australia. Her current research interests include convex and non-convex optimization, IRS-assisted communication, resource allocation, physical-layer security, and green (energy-efficient) wireless communications.
\end{IEEEbiography}
\vspace{-5mm}
\begin{IEEEbiography}[{\includegraphics[width=1in,height=1.1in]{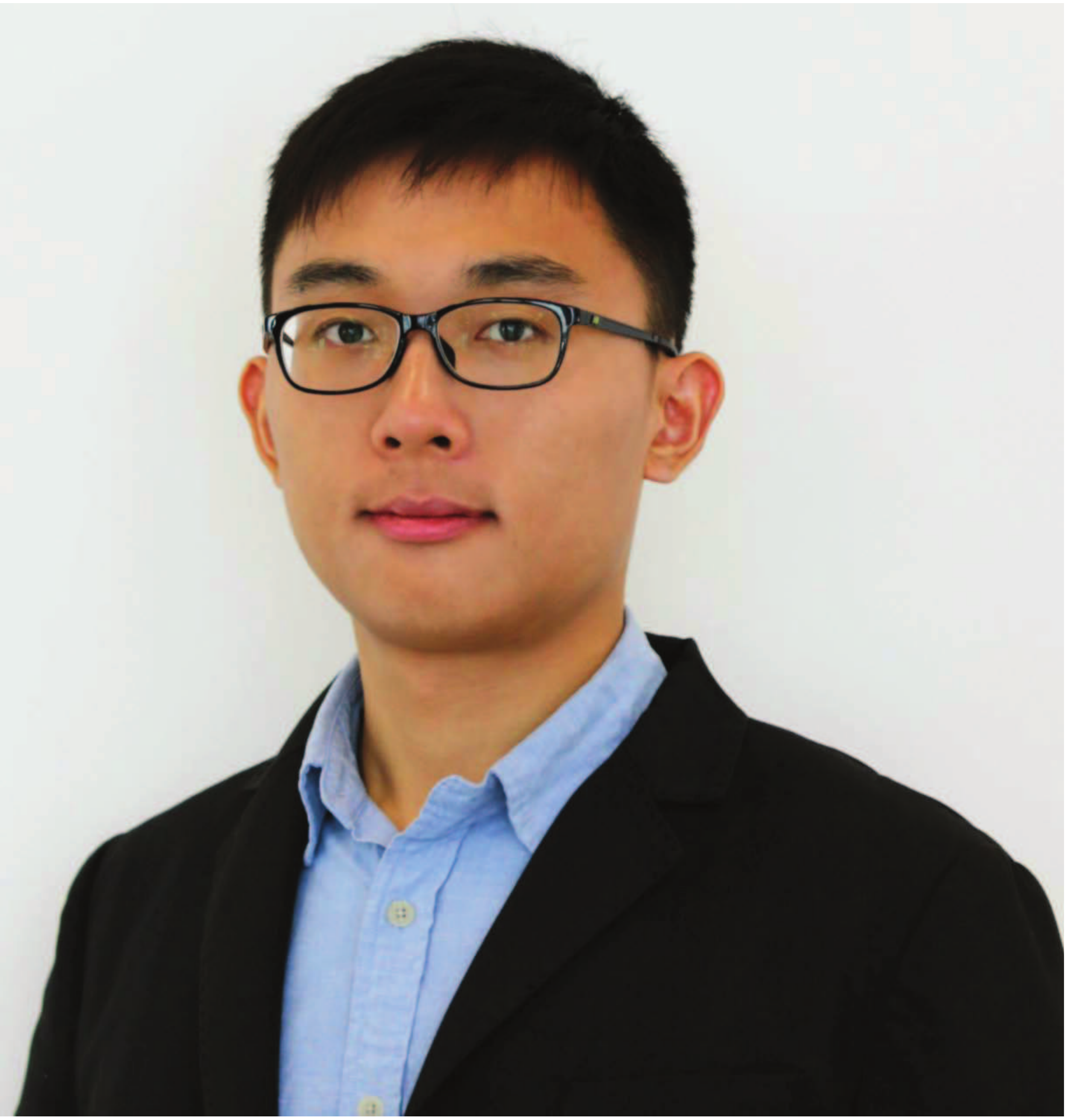}}]{Zhiqiang Wei (S’16–M’19)} received the B.E. degree in information engineering from Northwestern Polytechnical University (NPU), Xi’an, China, in 2012, and the Ph.D. degree in electrical engineering and telecommunications from the University of New South Wales, Sydney, Australia, in 2019. From 2019 to 2020, he was a Postdoctoral Research Fellow with the University of New South Wales. He is currently a Humboldt Postdoctoral Research Fellow with the Friedrich-Alexander University Erlangen-Nuremberg. His current research interests include convex/non-convex optimization, resource allocation design, intelligent reflecting surface, millimeter-wave communications, and orthogonal time-frequency space (OTFS) modulation. He received the Best Paper Award from the IEEE International Conference on Communications (ICC), in 2018.
\end{IEEEbiography}
\vspace{-5mm}
\begin{IEEEbiography}[{\includegraphics[width=0.95in,height=1.3in]{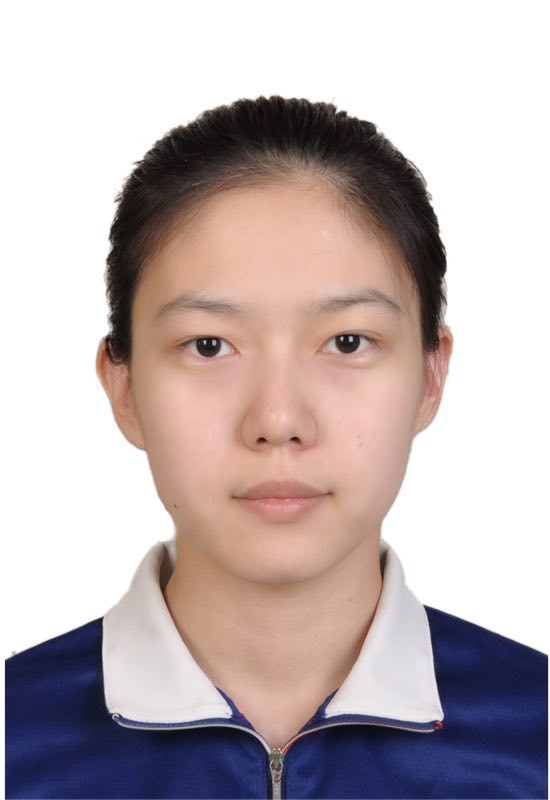}}]{Yuanxin Cai (S'19)} received the B.S. degree in Optical Information Science and Technology from the University of Electronic and Technology of China, Sichuan, China, in 2015 and the M.E. degree in Electrical Engineering and Telecommunication from the University of New South Wales, Sydney, Australia, in 2018. She is currently pursuing the Ph.D. degree in Telecommunication with the University of New South Wales, Sydney, Australia. Her current research interests include convex and non-convex optimization, UAV-assisted communication, resource allocation, physical-layer security, and green (energy-efficient) wireless communications.
\end{IEEEbiography}
\vspace{-5mm}
\begin{IEEEbiography}[{\includegraphics[width=0.95in,height=1.3in]{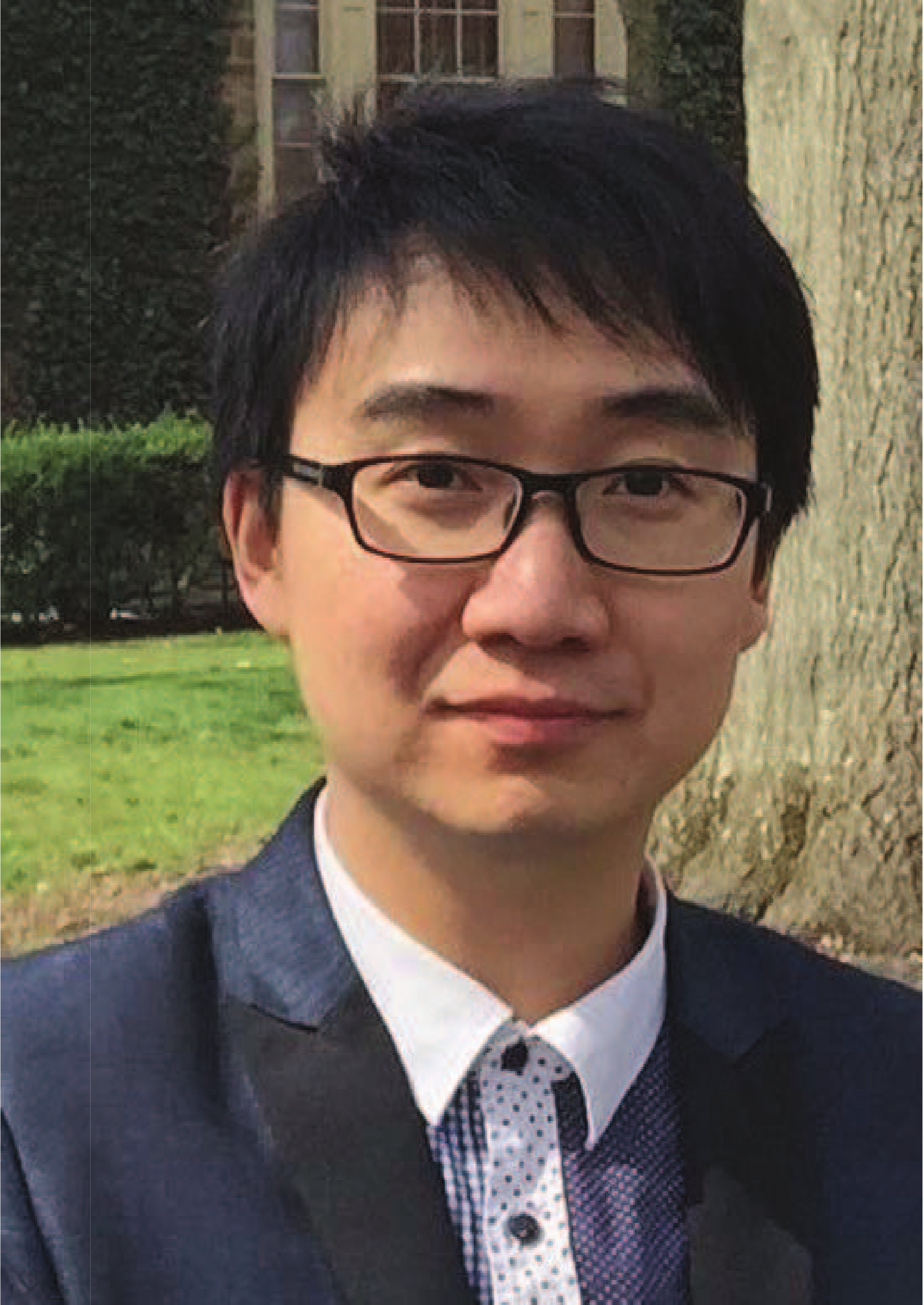}}]{Chang Liu (M’19)} received the B.E. degree in electronic information engineering from Dalian Maritime University, Dalian, China, in 2012, and the Ph.D. degree in signal and information processing from Dalian University of Technology, China, in 2017. He is currently a Postdoctoral Research Fellow with the University of New South Wales, Sydney, Australia. He was a visiting scholar in department of electrical engineering and computer science at University of Tennessee, Knoxville, USA, and a postdoctoral research fellow with the National Key Laboratory of Science and Technology on Communications, University of Electronic Science and Technology of China, Chengdu, China. His research interests include Machine Learning for Communications, Statistical Signal Processing, IRS-assisted Communications, UAV-assisted Communications, Internet-of-Things, and Spectrum Sensing and Sharing in Cognitive Radio.
\end{IEEEbiography}
\vspace{-5mm}
\begin{IEEEbiography}[{\includegraphics[width=0.95in,height=1.4in]{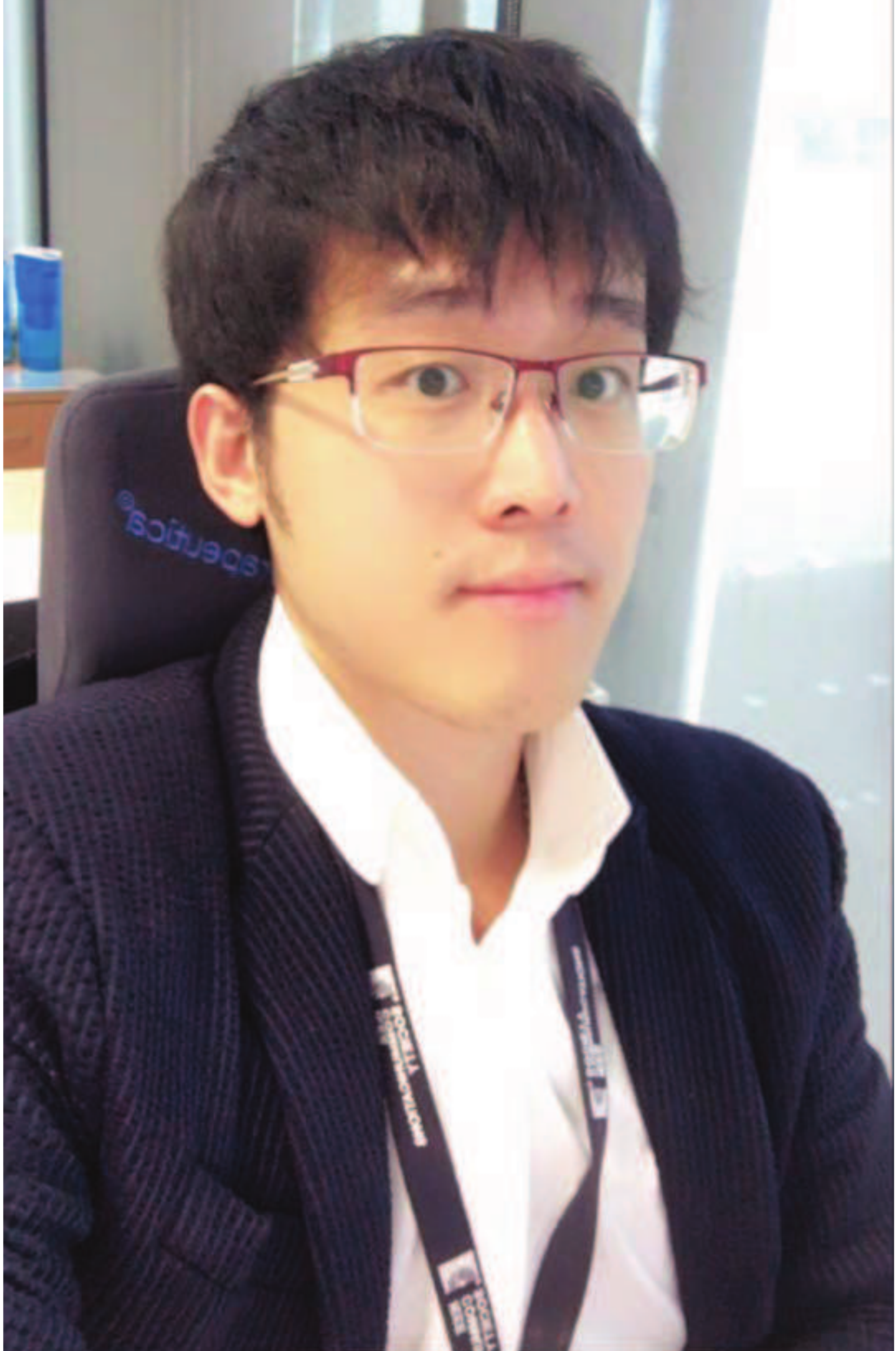}}]{Derrick
Wing Kwan Ng (S'06-M'12-SM'17-F'21)  } received the bachelor degree with first-class honors and the Master of Philosophy (M.Phil.) degree in electronic engineering from the Hong Kong University of Science and Technology (HKUST) in 2006 and 2008, respectively. He received his Ph.D. degree from the University of British Columbia (UBC) in 2012. He was a senior postdoctoral fellow at the Institute for Digital Communications, Friedrich-Alexander-University Erlangen-N\"urnberg (FAU), Germany. He is now working as a Senior Lecturer and a Scientia Fellow at the University of New South Wales, Sydney, Australia.  His research interests include convex and non-convex optimization, physical layer security, IRS-assisted communication, UAV-assisted communication, wireless information and power transfer, and green (energy-efficient) wireless communications.

Dr. Ng received the Australian Research Council (ARC) Discovery Early Career Researcher Award 2017,   the Best Paper Awards at the WCSP 2020,  IEEE TCGCC Best Journal Paper Award 2018, INISCOM 2018, IEEE International Conference on Communications (ICC) 2018,  IEEE International Conference on Computing, Networking and Communications (ICNC) 2016,  IEEE Wireless Communications and Networking Conference (WCNC) 2012, the IEEE Global Telecommunication Conference (Globecom) 2011, and the IEEE Third International Conference on Communications and Networking in China 2008.  He has been serving as an editorial assistant to the Editor-in-Chief of the IEEE Transactions on Communications from Jan. 2012 to Dec. 2019. He is now serving as an editor for the IEEE Transactions on Communications,  the IEEE Transactions on Wireless Communications, and an area editor for the IEEE Open Journal of the Communications Society. Also, he has been listed as a Highly Cited Researcher by Clarivate Analytics since 2018.
\end{IEEEbiography}
\vspace{-5mm}
\begin{IEEEbiography}[{\includegraphics[width=0.95in,height=1.3in]{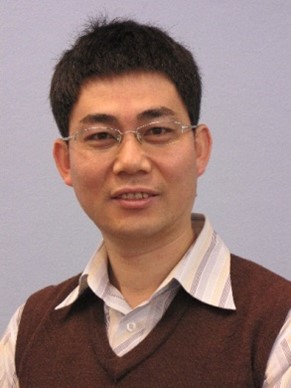}}]{Jinhong Yuan (M'02--SM'11--F'16)} received the B.E. and Ph.D. degrees in electronics engineering from the Beijing Institute of Technology, Beijing, China, in 1991 and 1997, respectively. From 1997 to 1999, he was a Research Fellow with the School of Electrical Engineering, University of Sydney, Sydney, Australia. In 2000, he joined the School of Electrical Engineering and Telecommunications, University of New South Wales, Sydney, Australia, where he is currently a Professor and Head of Telecommunication Group with the School. He has published two books, five book chapters, over 300 papers in telecommunications journals and conference proceedings, and 50 industrial reports. He is a co-inventor of one patent on MIMO systems and two patents on low-density-parity-check codes. He has co-authored four Best Paper Awards and one Best Poster Award, including the Best Paper Award from the IEEE International Conference on Communications, Kansas City, USA, in 2018, the Best Paper Award from IEEE Wireless Communications and Networking Conference, Cancun, Mexico, in 2011, and the Best Paper Award from the IEEE International Symposium on Wireless Communications Systems, Trondheim, Norway, in 2007. He is an IEEE Fellow and currently serving as an Associate Editor for the IEEE Transactions on Wireless Communications and IEEE Transactions on Communications. He served as the IEEE NSW Chapter Chair of Joint Communications/Signal Processions/Ocean Engineering Chapter during 2011-2014 and served as an Associate Editor for the IEEE Transactions on Communications during 2012-2017. His current research interests include error control coding and information theory, communication theory, and wireless communications.
\end{IEEEbiography}

\end{document}